\documentclass{div}
\usepackage[open=true]{bookmark}

\newcommand{\ObsHazOnetoTwo}{43.6} 

\newcommand{\EstHazOnetoTwo}{33.6} 
\newcommand{\EstHazTwotoThree}{75} 

\newcommand{\NoticeOnJFinit}{7 } 
\newcommand{\agecutoff}{21 to 64 } 

\newcommand{\hOShort}{14.7} 
 
\newcommand{\hOTwelveShort}{48.7} 
\newcommand{\hOTwelveLong}{56.6} 

\newcommand{\outdir}{./output}

\newcommand{\mythanks}{I am grateful to Theodore Papageorgiou, Kevin Lang, Shakeeb Khan, and Arthur Lewbel for their invaluable guidance and support. I would also like to thank all participants of the Dissertation Workshop at Boston College for their helpful comments. I am solely responsible for any errors or omissions.}
\title{\vspace{-1cm} \Large \setstretch{1.25} \sc{Duration Dependence and Heterogeneity: \\ Learning from Early Notice of Layoff}\thanks{\mythanks}\vspace{-0.75em}}
\author{\textit{\sc{Div Bhagia}\thanks{California State University, Fullerton; dbhagia@fullerton.edu}}}
\date{\vspace{-0.25cm} \today}

\begin{document}
\maketitle
\vspace{-0.75cm}

\begin{abstract}
This paper presents a novel approach to distinguishing the impact of duration-dependent forces and adverse selection on unemployment exit rates using variation in layoff notice lengths. I formulate a Mixed Hazard model in discrete time, specifying conditions for identifying structural duration dependence while allowing for arbitrary worker heterogeneity. Using data from the Displaced Worker Supplement (DWS), I estimate the model and find the decline in exit rates over the first 48 weeks is mainly due to the worsening composition of surviving jobseekers. Moreover, an individual's exit probability decreases initially, then increases until unemployment benefit exhaustion, and remains steady thereafter.
\end{abstract}

\begin{center}
\begin{minipage}{0.875\textwidth}
\setstretch{1.1} \noindent
\vspace{-1em}

\end{minipage}
\end{center}

\setlength{\abovedisplayshortskip}{-0.1cm}
\setlength{\belowdisplayshortskip}{0.25cm}
\setlength{\abovedisplayskip}{0.25cm}
\setlength{\belowdisplayskip}{0.25cm}


\clearpage
\section{Introduction}

A well-established empirical regularity is that the exit rate out of unemployment decreases over the spell of unemployment, except for a spike at the time of unemployment insurance (UI) exhaustion. The decline in the exit rate may represent negative duration dependence, meaning that the longer a worker remains unemployed, the less likely they are to exit unemployment. This would be true if employers discriminate against long-term unemployed workers or if workers lose valuable skills and connections over time, which would otherwise assist them in finding employment. However, workers with different unemployment durations, who may appear similar to researchers, may actually be quite different from each other. Factors such as employability, the urgency to find a job, or the ability to secure employment may vary across individuals. Such heterogeneity across workers would imply that the observed exit rate declines even in the absence of structural duration dependence.  As more employable workers leave unemployment, the remaining pool of unemployed individuals increasingly consists of those who are less likely to exit unemployment.

Understanding how the likelihood of exiting unemployment evolves over the unemployment spell and the extent of heterogeneity across workers is crucial for the design of unemployment policies.\footnote{See \cite{ShimerWerning2006}, \cite{PavoniViolante2007}, \cite{Pavoni2009}, and \cite{KolsrudEtAl2018}.} Furthermore, the magnitude and direction of structural duration dependence have implications for the incidence of long-term unemployment and the speed of recovery from economic downturns \citep{Pissarides1992}. Given its significance, a substantial body of literature has attempted to disentangle the sources of the decline in the exit rate from unemployment. However, it has proven to be challenging to do so using observational data.

In this paper, I develop and implement a novel approach to empirically disentangle the contributions of structural duration dependence and unobserved heterogeneity in explaining the dynamics of the exit rate from unemployment. My approach relies on leveraging variation in the length of notice workers receive from their employers before being laid off. Using data from the Displaced Worker Supplement (DWS), I compare workers with a notice period of more than two months (referred to as long notice) to workers with a notice period of one to two months (referred to as short notice). To ensure comparability across the two groups, I use inverse probability weighting (IPW) to achieve balance on a comprehensive set of observable characteristics. The analysis reveals that during the initial 12 weeks, the exit rate out of unemployment is \NoticeOnJFinit percentage points higher for long-notice workers. This difference is primarily due to a larger proportion of long-notice workers transitioning directly to their next job without experiencing a period of unemployment. However, beyond the first 12 weeks, the exit rate for workers with the longer notice is actually lower.

I argue that the lower exit rate for long-notice workers at later durations is due to the composition of this group becoming relatively worse as a larger proportion of individuals exit early in the spell. This indicates the presence of heterogeneity among workers. In the presence of heterogeneity, such as differences in employability, those who are more employable exit unemployment earlier. As more workers exit early from the long-notice group, the surviving workers from this group will have a lower proportion of highly employable workers compared to the short-notice group. Conversely, if there is no heterogeneity, a larger proportion of workers exiting early from the long-notice group will not alter the composition of long versus short-notice workers at later durations. Consequently, there would be no discernible difference in the exit rates of the two groups.\footnote{It is also possible that receiving a longer notice directly affects a worker's exit probability, even at later durations. I discuss this possibility below while addressing the robustness of my findings.}  Thus, the difference in the exit rates of short and long-notice workers is indicative of the extent of underlying heterogeneity. This is the fundamental idea behind my approach, which enables me to pin down the extent of heterogeneity and estimate structural duration dependence.




I operationalize this intuition by formulating a Mixed Hazard (MH) model \citep{Lancaster1979} in discrete time with multiple notice lengths. Within this framework, I specify the probability of an individual exiting unemployment as a product of their unobservable type, a function of observable factors, and a structural hazard that varies with the duration of unemployment and notice length. I show that structural duration dependence, characterized by how the structural hazard varies with the duration of unemployment, can be identified when two conditions hold. The first condition, commonly referred to as unconfoundedness \citep{RosenbaumRubin1983},  requires that the length of notice is independent of the worker's unobservable type when conditioned on observable characteristics. The second condition states that the length of the notice period does not impact the structural hazard at later durations of unemployment. In other words, while a longer notice period before layoff may affect the probability of exiting at the beginning of the unemployment spell, it does not influence the probability of exiting at later durations.  

Building on the identification result, I develop a method for estimating the model using the Generalized Method of Moments (GMM). The estimation method utilizes moments weighted by the inverse of propensity scores to ensure that the distribution of observable characteristics is similar across different notice lengths. Additionally, the method incorporates right-censored duration data to account for spells that are still incomplete at the time of the survey.

Relative to the existing literature on the identification and estimation of the Mixed Hazard model, my approach goes further in several dimensions. Firstly, I do not impose any functional form restrictions on the distribution of heterogeneity. This is crucial because misspecification of unobserved heterogeneity can significantly impact estimates of structural duration dependence, as demonstrated by \cite{HeckmanSinger1984}. Secondly, identification in my model stems from variation in a variable---the length of notice---that is exogenous conditional on observables and is assumed to affect the structural hazard in a way that aligns with economic intuition.\footnote{Existing non-parametric identification results for the Mixed Hazard model rely on variation in an exogenous variable that enters the structural hazard multiplicatively \citep{ElbersRidder1982, HeckmanSinger1984}. The practical implementation of these results has been limited due to the challenge of locating a variable that meets this criterion, as well as the absence of a convenient estimator. Another approach to identification is using multiple spell data \citep{Honore1993}. However, this approach assumes that the unobserved characteristics of the jobseeker remain constant across repeated spells.} The introduction of unconfoundedness in the Mixed Hazard framework is a novel innovation. Lastly, I provide a root-$n$ consistent estimator for the parameters of my model. My model is analogous to a Mixed Hazard model with a time-varying exogenous variable, for which \cite{Brinch2007} provides a non-constructive proof of identification. The key differences here are the assumption of unconfoundedness instead of exogeneity and the exposition being in discrete time, the latter of which leads to a consistent estimator for the model’s parameters via GMM.\footnote{Although the discrete-time model provides a simple, consistent estimator, it must be noted that one of its disadvantages is invariance to the definition of periods.} To the best of my knowledge, \cite{AlvarezEtAl2021} is the only other study that utilizes moment conditions from a discrete version of the Mixed Hazard model and constructs a GMM estimator. However, their identification result and estimator pertain to multiple spell data.\footnote{\cite{VanDenBerg1990} also set up a discrete-time MPH model and study identification using cohort effects; however, they use a maximum likelihood estimator on aggregate data for estimation.}

I estimate the model using weighted moments from the DWS data. The estimates uncover substantial heterogeneity in individual exit probabilities. I find that a quarter of the \ObsHazOnetoTwo\% decline in the observed exit rate over the first five months is due to the changing composition of workers over the spell of unemployment. Moreover, after the first five months, an individual worker's exit probability increases until the time of their unemployment benefit exhaustion and remains constant after. This is in contrast to the observed exit rate, which continues to decline even after benefit exhaustion. Recently, researchers have proposed behavioral modifications to standard search theory to explain this decline \citep{BoonevanOurs2012, DellaVignaEtAl2017, DellaVignaEtAl2021}. However, I highlight an alternative explanation: as a substantial number of individuals exit unemployment right at the point of benefit exhaustion, the composition of the remaining unemployed workers becomes significantly worse. This compositional change contributes to the observed decline in the exit rate after benefit exhaustion. Finally, I calibrate a partial equilibrium search model with a non-stationary environment \citep{Mortensen1986, VanDenBerg1990} and show that my findings can be rationalized in this framework with a decline in returns to search early in the spell.

Under the identifying assumptions specified for the Mixed Hazard model in my framework, the lower exit rate among long-notice workers after the initial 12 weeks is attributed to the presence of heterogeneity. However, two alternative explanations are possible. First, there could be unobservable differences between long- and short-notice workers. Second, a longer notice period may reduce a worker's exit probability at later durations. To address these concerns, in \ref{app_gen}, I relax the assumptions of my model to allow for arbitrary differences between the two groups and for the structural hazards at later durations to vary by notice length up to a specific constant. Although I cannot show that all the parameters of this more general model are identified, I estimate the model by varying the additional parameters and find the values that minimize residuals.\footnote{I verify that the numerical error function is locally convex in all cases.} The estimated values that minimize the residuals suggest no mean differences in the unobservable heterogeneity distribution between the two groups nor any difference in the structural hazard beyond the first 12 weeks. This points towards the validity of the identifying assumptions employed in the analyses.

This paper contributes to the extensive literature on the dynamics of job-finding over the spell of unemployment. Previous empirical studies utilizing the Mixed Hazard model have had to make strong functional form assumptions due to challenges with estimation. Consequently, the evidence on structural duration dependence from these studies is mixed, as highlighted by \cite{MachinManning1999} in their review. Recently, \cite{AlvarezEtAl2016} revived this strand of work by estimating a Mixed Hitting-Time (MHT) model \citep{Abbring2012} using Austrian social security data. They focus on a selected sample of workers with multiple unemployment spells and are able to estimate the extent of heterogeneity across workers that is fixed between spells. A relative advantage of my approach is that it captures spell-specific heterogeneity.\footnote{For instance, a worker's savings or UI eligibility may change over the months or years by the time this worker becomes unemployed again.} Another related study, \cite{MuellerEtAl2021}, utilizes variation in expectations about job-finding from survey data to pin down variation in actual job-finding rates. While both of these studies also document substantial heterogeneity across jobseekers, my estimator for structural duration dependence is flexible enough to capture changes around UI exhaustion.\footnote{\cite{AlvarezEtAl2016} utilizes an optimal-stopping model; a worker finds a job at an optimal stopping time when a Brownian motion with drift hits a barrier. Their model generates an inverse Gaussian distribution of duration for each worker. \cite{MuellerEtAl2021} restrict the structural hazard to be monotonic over the spell of unemployment, and their estimator yields a practically flat hazard.}

Given the challenges with estimating structural duration dependence, researchers have instead focused on assessing its determinants. \cite{KroftEtAl2013} conduct an audit study and find that the likelihood of receiving a callback for an interview declines with the duration of unemployment. However, they note that since they cannot measure worker behavior or employers' ultimate hiring decisions, their estimates only shed light on one determinant of structural duration dependence.\footnote{Using a structural model, \cite{JaroschPilossoph2019} argue that if employers statistically discriminate against those with longer durations, then a decline in callback rates only has a marginal effect on workers' exit rates.} Several papers have also documented how search effort or reservation wages evolve over the spell of unemployment \citep{KruegerMueller2011, MarinescuSkandalis2021, DellaVignaEtAl2021}. The evidence provided in this paper suggests that, while call-back rates or other factors affecting returns to search matter initially, a worker's optimizing behavior determines the likelihood of exiting unemployment at later durations.

Finally, a substantial body of literature highlights a spike in exit rates at UI exhaustion, where exit rates increase until benefit exhaustion and decline thereafter.\footnote{\cite{KatzMeyer1990} first documented the spike in exit rates at benefit exhaustion in the context of the US. Some recent papers that document this pattern using administrative data are \cite{DellaVignaEtAl2017} (Hungary), \cite{GanongNoel2019} (US), and \cite{MarinescuSkandalis2021} (France).} While the initial increase is consistent with standard search theory, the subsequent decline is not. My estimates reproduce the increase in individual exit probabilities leading up to UI exhaustion but do not find evidence of a decline thereafter. \cite{BoonevanOurs2012} propose storable job offers as an explanation for the spike, while \cite{DellaVignaEtAl2017} argue that search models incorporating reference dependence predict a decrease in search effort after benefit exhaustion.  My estimates suggest that the decline in the exit rate after UI exhaustion can be attributed to a shift in the composition of surviving workers, as a significant proportion of workers exit unemployment right at benefit exhaustion. However, the individual exit probability remains constant, consistent with the predictions of standard search models.


\section{Context and Data}\label{sec_data}

In this section, I describe the institutional setting and the data and document how the exit rate out of unemployment varies with the length of notice. 


\subsection{Institutional Details}

Under certain circumstances, US employers are required to give notice of layoff. The federal WARN Act mandates that employers with 100 or more full-time employees provide a 60-day advance notice for plant closings and mass layoffs. A plant closing is defined as the shutdown of a site or units within it that results in 50 or more employees losing their jobs within a 30-day period, while a mass layoff is the loss of employment for 500 or more employees during a 30-day period, or 50-499 employees if they constitute one-third or more of the employer's active workforce. The law only applies to layoffs exceeding six months, excluding discharges for cause, voluntary departures, or retirements. Some states, such as California, New York, and Illinois, have implemented their own WARN laws that expand the coverage of employment losses beyond what the federal law requires.\footnote{It is not possible to exploit policy variation across states, say in a differences-in-difference framework, due to confounding pre-trends; both California and New York implemented these laws in the aftermath of a national recession.}

When it comes to unemployment insurance (UI), US workers who are terminated without cause are typically eligible to receive benefits for a limited duration. Although the UI program is a federal program, each state sets its own benefit levels and durations. Eligibility and benefits may depend on a combination of earnings, hours worked, or weeks worked during a base period, depending on the specific rules of the state's UI program. Typically, this base period consists of the first four out of five completed calendar quarters preceding the claim filing date. In most states, the maximum period for receiving benefits is 26 weeks. Nine states have a uniform benefit duration of 26 weeks, while the benefit durations in the remaining states vary depending on the applicant's earnings history. Additionally, a program for extended benefits has been in place since a 1970 amendment to the Federal Unemployment Tax Act (FUTA), which can be triggered by the state unemployment rate. Temporary programs have also been implemented to extend benefits during recessions.


\subsection{Data Description and Sample Construction}\label{subsec_description}

I use data from the Displaced Worker Supplement (DWS) for the years 1996-2020. DWS is fielded biennially along with the basic monthly Current Population Survey (CPS) in January or February. The survey is administered to individuals who report having \textit{lost or left} a job within the past three years due to a plant closure, their position being abolished, or having insufficient work at their previous employment. Apart from details on workers' lost and current jobs, DWS also collects the length of the notice period workers received before being laid off and the length of time they took to find another job. 

For the analysis, the sample is limited to individuals aged \agecutoff years old. It excludes individuals who expected to be recalled to their previous jobs and those whose lost job was self-employment. To focus on workers who lost stable, permanent employment, the sample only includes individuals who were employed full-time for at least six months at their previous job and had health insurance benefits provided by that employer. The question on notice length in the DWS is categorical and specifies whether someone did not receive a notice or received a notice of <1 month, 1-2 months, or >2 months. For individuals who did not receive a notice, it is uncertain whether they were displaced or quit their jobs voluntarily, so they are excluded from the sample.\footnote{As noted by \cite{Farber2017}, while the DWS seeks to capture separations caused by economic difficulties within firms, it is hard to distinguish layoffs from voluntary quits. This is because financially distressed firms might reduce hours or wages rather than lay off workers, prompting some workers to voluntarily quit in search of better opportunities.} Additionally, I exclude individuals with less than a month's notice from the main analysis, as they significantly differ from the other two groups in observable characteristics, casting doubt on the validity of the unconfoundedness assumption. Nevertheless, I present results including these groups in Section \ref{subsec_robust} as a robustness check. Section \ref{app_data_cnstr} in the Appendix provides additional details on data construction. Specifically, Table \ref{tab_sample} summarizes the sample selection procedure, and Table \ref{tab_cps_comparison} compares the characteristics of the workers in my sample to all workers in the CPS and DWS.

\begin{table}[p]
\setstretch{1.15}
\caption{Descriptives by Notice Length}\label{tab_sum_stats}
\begin{threeparttable}
\begin{tabularx}{\textwidth}{p{0.245\textwidth}YYYp{0.15cm}YYY}
\toprule
& \multicolumn{3}{c}{Unbalanced} & & \multicolumn{3}{c}{Balanced} \\
& Short &  Long & Diff. & & Short &  Long & Diff. \\
& (1) & (2) & (2)-(1) & & (3) & (4) & (4)-(3) \\
\midrule 
 Age              & 42.44  & 43.57  & 1.13*** &  & 43.04  & 43.01  & -0.03  \\
                  & (0.24) & (0.22) & (0.33)  &  & (0.24) & (0.22) & (0.33) \\
 Female           & 0.45   & 0.46   & 0.02    &  & 0.46   & 0.46   & -0.00  \\
                  & (0.01) & (0.01) & (0.02)  &  & (0.01) & (0.01) & (0.02) \\
 Married          & 0.59   & 0.63   & 0.04**  &  & 0.61   & 0.61   & -0.00  \\
                  & (0.01) & (0.01) & (0.02)  &  & (0.01) & (0.01) & (0.02) \\
 Black            & 0.10   & 0.09   & -0.01   &  & 0.10   & 0.09   & -0.00  \\
                  & (0.01) & (0.01) & (0.01)  &  & (0.01) & (0.01) & (0.01) \\
 College Degree   & 0.41   & 0.39   & -0.03*  &  & 0.40   & 0.40   & -0.00  \\
                  & (0.01) & (0.01) & (0.02)  &  & (0.01) & (0.01) & (0.02) \\
 Plant Closure    & 0.46   & 0.62   & 0.16*** &  & 0.54   & 0.54   & -0.00  \\
                  & (0.01) & (0.01) & (0.02)  &  & (0.01) & (0.01) & (0.02) \\
 Union Membership & 0.15   & 0.16   & 0.01    &  & 0.15   & 0.15   & 0.00   \\
                  & (0.01) & (0.01) & (0.01)  &  & (0.01) & (0.01) & (0.01) \\
 In Metro Area    & 0.84   & 0.82   & -0.01   &  & 0.83   & 0.83   & 0.00   \\
                  & (0.01) & (0.01) & (0.01)  &  & (0.01) & (0.01) & (0.01) \\
 Years of Tenure  & 7.12   & 9.18   & 2.06*** &  & 8.25   & 8.21   & -0.04  \\
                  & (0.15) & (0.16) & (0.22)  &  & (0.16) & (0.15) & (0.22) \\
 Log Earnings     & 6.54   & 6.56   & 0.03    &  & 6.54   & 6.55   & 0.00   \\
                  & (0.01) & (0.01) & (0.02)  &  & (0.01) & (0.01) & (0.02) \\
 Observations     & 1959   & 2216   &         &  & 1959   & 2216   &        \\
 \addlinespace[1ex]
\bottomrule
\end{tabularx}
\begin{tablenotes}
\item \textit{Note:} The sample consists of respondents from the Displaced Worker Supplement (DWS) for the years 1996-2020, who were between the ages of \agecutoff, who had worked full-time for at least six months at their previous job, received health insurance from their former employer, and did not expect to be recalled. Short notice refers to a notice period of 1-2 months, while long notice refers to a notice period exceeding two months. Columns (1) and (2) present raw averages for the sample, while columns (3) and (4) show weighted averages, where the weights correspond to the inverse of the estimated probabilities of receiving short or long notice.      
\end{tablenotes}
\end{threeparttable}
\end{table}

Throughout the remainder of the text, I use the term \textit{long notice} interchangeably to refer to a notice period of more than two months and \textit{short notice} to refer to a notice period of 1-2 months. Table \ref{tab_sum_stats} presents the summary statistics separately for workers with short and long notice in the sample. Columns (1) and (2) display the raw averages for the sample, revealing notable differences between the two groups. Workers with longer notice tend to be older and are more likely to be married. Additionally, workers laid off during plant closures are more likely to receive longer notice, potentially due to compliance with the WARN law. Workers with longer notice also tend to have longer job tenure. However, there are no notable differences in earnings for the two groups. 

To isolate the impact of notice from these additional correlates, which may affect the probability of exiting unemployment, I reweight the sample using inverse propensity score weighting. I use a logistic regression model to predict the likelihood of receiving a longer notice based on several covariates. These covariates consist of age, gender, marital status, race, education, location characteristics, the reason for displacement, year of displacement, industry and occupation of the lost job, as well as union status, tenure, and earnings at the lost job. I then utilize the propensity scores to assign weights to the observations. Specifically, individuals with the long notice are assigned a weight of $1/\hat{p}(X_i)$, where $\hat{p}(X_i)$ is the estimated probability of receiving the long notice from the regression model for an individual with covariates $X_i$. On the other hand, individuals who received the short notice are assigned a weight of $1/(1-\hat{p}(X_i))$. 

The summary statistics for the reweighted sample are presented in columns (3) and (4) of Table \ref{tab_sum_stats}. After reweighting, the observable differences between the two groups disappear, indicating that the weights effectively minimize the observed disparities. Section \ref{app_psw} in the  Appendix provides additional details on propensity score estimation. Figure \ref{fig_ps_bal} demonstrates a high degree of overlap in the estimated propensity score distributions between long and short-notice workers. Additionally, Figures \ref{fig_dyear_bal} and \ref{fig_occ_ind_bal} depict the balance of the weighted sample with respect to the displacement year and industrial and occupational composition, respectively.


\subsection{Distribution of Unemployment Duration}\label{subsec_dur_dist}

In this section, I explore how a longer notice impacts the exit rate over the spell of unemployment. Workers who receive a layoff notice may start searching for a job before separating from their previous employer. In this case, some of these workers may secure a new job during the notice period, thus avoiding any period of unemployment. In the data, \hOShort\% of the workers with the short layoff notice report no duration of unemployment. Since workers with longer notice periods have more time to search for a new job while still employed, we expect their chances of avoiding unemployment to be even greater.  

In Table \ref{tab_init_hazard}, panel A, I examine the relationship between receiving a long notice and reporting an unemployment duration of 0. Columns (1) and (2) present estimates from unweighted regressions, while columns (3) and (4) present weighted regression estimates using the weights described in the previous section. Additionally, columns (2) and (4) include a comprehensive set of controls identical to the ones used to generate the weights. The table shows that the impact of a lengthier notice on the exit probability is reduced after accounting for observable characteristics of the separation. The coefficient in column (2) indicates that individuals who receive a longer notice are 8 percentage points more likely to avoid unemployment. Similar estimates are observed in columns (3) and (4) as well. Notably, the inclusion of controls in column (4) does not lead to a change in the coefficient, indicating that the weighting has effectively achieved balance in terms of the covariates across the two groups. In panel B of Table \ref{tab_init_hazard}, I present a similar regression analysis, but this time using an indicator for exiting unemployment within the first 12 weeks. The results show that the exit rate out of unemployment is around \NoticeOnJFinit percentage points higher for long-notice workers compared to the short-notice group.

\begin{table}[t]
\begin{threeparttable}
\caption{Observed Exit Rate -- Early in the Spell}\label{tab_init_hazard}
\begin{tabularx}{\textwidth}{p{0.225\textwidth}YYYY}
\toprule
& (1) & (2) & (3) & (4) \\
\midrule \addlinespace[1ex]
\multicolumn{5}{c}{\underline{\sc{Panel A. $\I\{\text{Unemployment duration $=0$ weeks}\}$}}} \\ \addlinespace[2ex]
> 2 month notice & 0.094*** & 0.080*** & 0.077*** & 0.077***\\
  & (0.012) & (0.012) & (0.013) & (0.013)\\ \addlinespace[3ex]
\multicolumn{5}{c}{\underline{\sc{Panel B. $\I\{\text{Unemployment duration $\leq 12$ weeks}\}$}}} \\ \addlinespace[2ex]
> 2 month notice & 0.078*** & 0.074*** & 0.070*** & 0.070***\\
  & (0.015) & (0.016) & (0.016) & (0.016)\\ \addlinespace[2ex]
Controls   &  No & Yes  & No & Yes \\
Weights   & No  & No   & Yes & Yes \\
\midrule
Observations & 4175 & 4175 & 4175 & 4175\\
\bottomrule
\end{tabularx}
\begin{tablenotes}
\item \textit{Note:} The table presents estimates from linear regression models, where the main independent variable is an indicator variable that takes a value of 1 if the individual received a notice of more than 2 months and 0 if they received a notice of 1-2 months. The dependent variable is an indicator for reporting an unemployment duration of 0 weeks (Panel A) or less than 12 weeks (Panel B). The weights are generated using inverse probability weighting (IPW). Robust standard errors are reported in the parenthesis. 
\end{tablenotes}
\end{threeparttable}
\end{table}


To examine how the exit rate varies with the length of notice over the spell of unemployment, I bin unemployment duration into 12-week intervals.\footnote{Exit rate is defined as the ratio of individuals who found a job during a specific interval to those who were jobless at the beginning of the interval. See Figure \ref{fig_altbins} in the  Appendix for the presentation of data with alternative binning definitions.} Figure \ref{fig_dur_dist} presents the exit rate and the survival rate separately for the long- and short-notice workers over the spell of unemployment. Note that the rates are calculated using the weighted sample to ensure that the comparison is between similar groups of workers who received different lengths of notice.\footnote{ \ref{app_robust} presents the unweighted exit rates and corresponding estimates obtained from the Mixed Hazard model. There is little qualitative difference between unweighted and weighted quantities, suggesting that differences in other characteristics of the two groups are not driving the results.} Approximately \hOTwelveLong\% of individuals with a long notice exit within the first 12 weeks, while \hOTwelveShort\% of those with a short notice do the same. However, over the course of unemployment, individuals with shorter notice periods catch up, resulting in almost identical survival rates for both groups by the 48th week of unemployment. As shown in panel A of Figure \ref{fig_dur_dist}, for all durations beyond 12 weeks, individuals with shorter notice periods have a higher exit rate compared to those with longer notice periods. 

\begin{figure}[t]\caption{Exit and Survival Rate --- Later in the Spell}\label{fig_dur_dist}
\vspace{-0.5em}
\centering
\begin{subfigure}{.525\textwidth}
\centering
\includegraphics{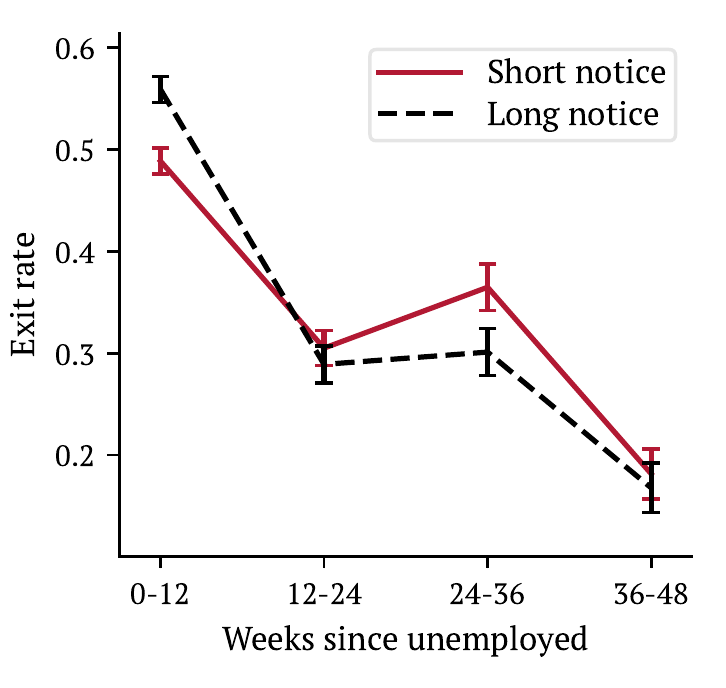}
\subcaption{Exit Rate}
\end{subfigure}
\begin{subfigure}{.45\textwidth}
\centering
\includegraphics{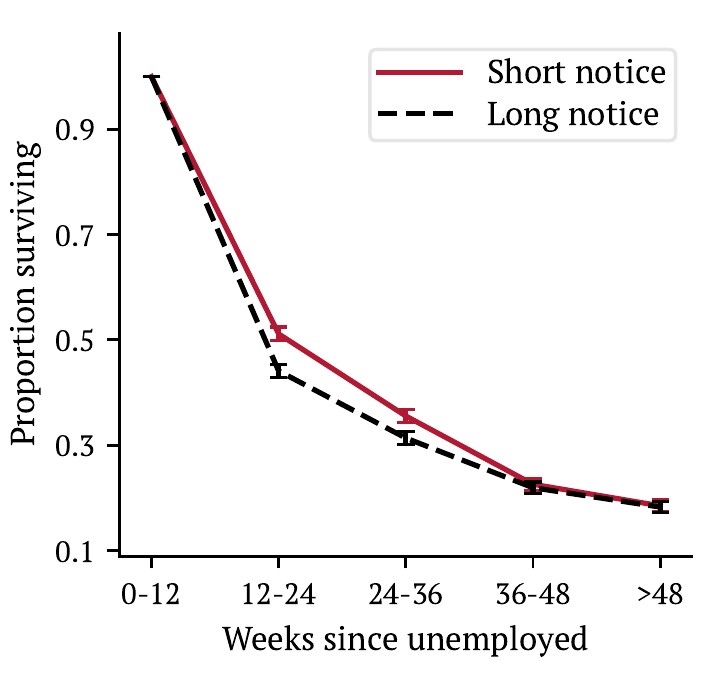}
\subcaption{Survival Rate}
\end{subfigure}
\vspace{-0.75em}
\floatfoot{\textit{Note:} Short notice refers to a notice of 1-2 months, and long notice refers to a notice of more than 2 months. Panel A presents the weighted proportion of individuals exiting unemployment in each interval amongst those who were still unemployed at the beginning of the interval. Panel B presents the weighted proportion of individuals who are unemployed at the beginning of each interval. Error bars represent 90\% confidence intervals.}
\end{figure}

I interpret the higher exit rate for short-notice workers beyond the initial 12 weeks as evidence for heterogeneity across workers. When workers are heterogeneous, those with better chances of exiting unemployment do so earlier. Given that a larger proportion of long-notice workers exit earlier, the long-notice group will have a lower proportion of individuals with higher exit probabilities, which is reflected in the (average) exit rate. It is important to note that this interpretation holds under the condition that longer notice does not directly reduce the probability of exiting unemployment at later durations. In the following section, I formally outline the assumptions necessary to identify heterogeneity and duration dependence in a Mixed Hazard model, and I also discuss the plausibility of these assumptions and potential violations.

\section{Econometric Framework}\label{sec_ef}

This section illustrates how variation in notice length can be used to identify structural duration dependence. Specifically, I set up a Mixed Hazard model in discrete time and specify the assumptions under which the key components of this model are identified. The model and assumptions are outlined in Section \ref{subsec_hazard_model}, while the main identification result is presented in Section \ref{subsec_identification}. The intuition behind identification is explained in Section \ref{subsec_intuition}. All proofs are presented in \ref{app_proofs}. 

\subsection{Mixed Hazard Model in Discrete Time}\label{subsec_hazard_model}


Prior to being laid off, workers are given a notice period of length \( L \), where \( L \) takes discrete values on some support $\mathcal{L}$. Workers are heterogeneous, each possessing an unobservable fixed type \( \nu \), characterized by the cumulative distribution \( F(.) \). Additionally, workers vary based on pre-notice observables, denoted by \( X \), with the distribution of \( X \) given by \( F_X(.) \). These characteristics may include details about the job, the worker, or the circumstances of the layoff. Using potential outcome notation \citep{Rubin1974}, let \(D_L\) denote the realized unemployment duration under notice \(L\).  \(D_L\) is a random variable that takes values in \(\{1, 2, 3, \ldots\}\). Let's begin by assuming that the econometrician observes $D = \sum_{\ell \in \mathcal{L}} D_{\ell} \I\{L=\ell\}$, which means that the potential duration under notice $\ell$ is observed among those who receive notice $\ell$.\footnote{Throughout the text, \(\I(.) \) is used to denote the indicator function.} Right-censored duration data will be incorporated later.

Finally, for a set of conditioning variables $\Upsilon$, define the hazard function \( h(.) \) as: 
\[h(d|\Upsilon)=\Pr(D=d|D \geq d, \Upsilon)\] 
Thus, \( h(d|\nu, L, X) \) represents the probability that an individual of type \(\nu\) with observed characteristics \(X\) will exit unemployment at duration \(d\), given that the individual has not yet exited and had received notice \(L\). 

We are interested in understanding how individual exit probabilities \( h(d|\nu, L, X) \) evolve over the duration of unemployment; this is referred to as structural duration dependence. However, these individual probabilities are not observed in the data. Instead, we can utilize the observed unemployment durations to determine the timing of each worker's exit from unemployment. Specifically, we can at most infer the exit rate \( h(d|L, X) \), which represents the proportion of individuals exiting at a specific duration relative to those who have remained up to that point in the group with notice \(L\) and observable characteristics \(X\).

It is worth noting that, according to the definition of the hazard function, we have:
\[ h(d|L, X) = \E[h(d|\nu, L, X) |D \geq d, L, X] \]
The above reformulation clarifies that the exit rate $h(d|L, X) $ captures the average hazard for individuals who have survived until duration $d$ rather than all individuals. This is what makes identifying structural duration dependence from observed exit rates challenging. To make the problem more manageable, the Mixed Hazard model introduces additional structure by assuming that the individual exit probabilities  $h(d|\nu, L, X) $  can be expressed as the product of the individual’s unobserved type  $\nu$  and a structural component  $\psi_L(d, X)$.\footnote{\cite{Lancaster1979} expanded the proportional hazard model \citep{Cox1972} to incorporate unobserved heterogeneity. His Mixed Proportional Hazard (MPH) model represented the hazard rate as a product of a regression function, a structural hazard that varies with duration, and the worker's unobserved type. The Mixed Hazard model formulated here is similar to Lancaster's MPH model but permits non-proportional effects of observable characteristics and distinguishes the length of notice from other observed variables.}

\begin{assump}{(Mixed Hazard)}\label{assump_mixed_hazard}
An individual's exit probability at duration $d$ is given by:
$$ h(d|\nu, L, X) =  \psi_L(d,X) \nu $$
where the structural hazard $\psi_L(d,X) \in (0,\infty)$ and worker's type $\nu \in (0, \bar{\nu}]$ with $\bar{\nu}=1/\max_{d,l,X}\{\psi_l(d,X)\}$.
\end{assump}

The structural hazard $\psi_L(d,X)$ is common to all individuals with observable characteristics $X$ and a notice period of length $L$, but it varies with the duration of unemployment. The restrictions on the structural hazard and the support of $\nu$ in Assumption \ref{assump_mixed_hazard} guarantee that individual exit probabilities lie between 0 and 1. Under Assumption \ref{assump_mixed_hazard}, workers with lower values of type \(\nu\) have lower exit probabilities at all durations. Consequently, low-type workers tend to remain unemployed for longer durations. Given that the observed exit rate reflects the average exit probability of surviving workers, it will decline more steeply than individual exit probabilities over the unemployment spell. The following proposition formally states this result. 

\begin{prop}\label{result_average_type_falls}
Under Assumption \ref{assump_mixed_hazard}, the exit rate $h(d|L,X)$ can be expressed as follows:
$$ h(d|L,X) = \psi_L(d,X) \E(\nu |D \geq d,L,X)$$
Moreover, the average type of workers who survive until $d$, $\E(\nu |D \geq d,L,X)$, decreases with the unemployment duration $d$.\end{prop}

This proposition clarifies why identifying the structural hazard $\psi_L(d, X)$ is challenging. The exit rate $h(d|L, X)$ is impacted by both the structural duration dependence $\psi_L(d, X)$ and the changing worker composition over the unemployment spell captured by $\E(\nu |D \geq d, L,X)$. If the observed exit rate $h(d|L, X)$ declines over the spell of unemployment, it is not possible to distinguish between the scenario where there is no structural duration dependence but significant worker heterogeneity causes the average type of workers and the observed exit rate to decline, and the scenario where there is no worker heterogeneity, but the structural hazard declines over the spell of unemployment. Both of these scenarios would be consistent with the observed decline in $h(d|L, X)$. Hence, structural duration dependence is not identified from exit rates in the model formulated so far.

I now introduce two additional assumptions under which variation in notice length leads to the identification of structural duration dependence. The first assumption is conditional independence, which states that the length of notice is independent of the worker's unobservable type given observable characteristics. In other words, for workers with similar observable characteristics, there is no systematic difference in the length of notice given to workers with different unobservable types.\footnote{In  \ref{app_gen}, I provide estimates from a model that permits the underlying type distribution to differ across various notice lengths instead of assuming conditional independence.}

\begin{assump}{(Conditional Independence)}\label{assump_independence}
The length of notice $L$ is independent of the worker's unobservable type $\nu$, given observable characteristics $X$, i.e., $L \perp \nu | X$.
\end{assump}

This assumption ensures that differences in the exit rates $h(d|L, X)$ for various notice lengths are only due to the direct effect of the notice on individual exit probabilities rather than due to differences in the types of workers who receive different lengths of notice. In other words, the assumption implies that for any notice $\ell$, \(h(d|\ell, X) = \E[h(d|\nu, \ell, X) |D_\ell \geq d, X] \). The second assumption, referred to as stationarity, states that the length of notice does not affect an individual's exit probability after the first period. 

\begin{assump}{(Stationarity)}\label{assump_stationarity}
For all $L, X,$ and $d>1$,
$$ \psi_L(d,X) = \psi(d,X) $$
\end{assump}

The rationale for Assumption \ref{assump_stationarity} is that workers with longer notice periods have more time to search for a new job before separating from their previous employer, potentially increasing their likelihood of finding a job at the beginning of their unemployment spell. However, suppose duration dependence in job-finding is caused by factors such as human capital depreciation due to prolonged unemployment or employers discriminating against long-term unemployed workers. In that case, a worker's exit probability later in the spell should only vary with the unemployment duration and not with the length of notice received at the onset of the spell. Given that I bin unemployment duration in 12-week intervals, this assumption translates to the length of notice only impacting the probability of exit within the first 12 weeks and not after that.

More generally, Assumption \ref{assump_stationarity} implies that individual exit probabilities vary only with the duration of unemployment and not with time elapsed since the start of the job search.\footnote{This assumption aligns with a large class of search models, including those that involve non-stationarity. For instance, the model proposed by \cite{LentzTranaes2005}, in which workers start searching harder over time as their savings run down, would be consistent with this assumption as savings only start depleting once unemployed.} This assumption would be violated if time spent searching increases or decreases an individual's likelihood of exiting unemployment. For instance, if workers learn while searching and become better at job search \citep{BurdettVishwanath1988, GonzalezShi2010}, then those with longer notice would have a higher hazard even beyond the initial period. On the other hand, time spent searching may decrease the exit probability if workers first apply to all jobs in stock but subsequently only apply to newly posted jobs \citep{ColesSmith1998}.\footnote{Another possibility for why individual exit probability may decline with time spent searching could be that workers get discouraged over time and stop trying. However, this seems unlikely since 85\% of individuals in the sample eventually find employment after displacement, and among those who do not, none report being out of the labor force.} In \ref{app_gen}, I investigate whether Assumption \ref{assump_stationarity} is violated by estimating a more general model where structural hazards are allowed to vary with the length of notice beyond the initial period, and I do not find evidence against it.


\subsection{Identification Results}\label{subsec_identification}

\begin{theorem}\label{result_main_theorem}
Under Assumptions \ref{assump_mixed_hazard}--\ref{assump_stationarity}, for any $\ell, \ell'$ with $\psi_{\ell}(1, X) \neq \psi_{\ell'}(1, X)$ and some integer $\bar{D}$, the structural hazards $\{\psi_{\ell}(1, X),\psi_{\ell'}(1, X), \{\psi(d, X) \}_{d=2}^{\bar{D}}\}$ and the conditional moments of the type distribution $\{\E(\nu^k|X)\}_{k=1}^{\bar{D}}$ are identified up to a scale from the conditional exit rates $\{h(d|\ell, X),h(d|\ell', X)\}_{d=1}^{\bar{D}}$.
\end{theorem}

The above theorem establishes that if the first-period hazard varies for two different notice lengths, we can identify the structural hazards up to $\bar{D}$ and the first $\bar{D}$ moments of $\nu$ conditional on $X$ using the conditional exit rates for both notice lengths up to $\bar{D}$. A direct implication of this result is that if $X$ does not enter the structural hazard and we assume independence instead of conditional independence, we can identify the model using exit rates unconditional on $X$. The following corollary presents this result formally.


\begin{corr}\label{result_uncond}
Assume the following conditions hold: \vspace{-1em}
\begin{itemize}[itemsep=-0.25em]
\item[(i)] $h(d|\nu, L) = \psi_L(d) \nu$, where \(\psi_L(d)\) and \(\nu\) are bounded to ensure \(h(d|\nu, L) \in (0,1)\)
\item[(ii)]  $L$ is independent of $\nu$
\item[(iii)]  $\psi_{L}(d) = \psi(d)$ for $d>1$
\end{itemize} \vspace{-1em}
Then for any $\ell$ and $\ell'$, with $\psi_{\ell}(1) \neq \psi_{\ell'}(1)$ and some integer $\bar{D}$, the structural hazards $\{\psi_{\ell}(1),\psi_{\ell'}(1), \{\psi(d) \}_{d=2}^{\bar{D}}\}$ and the moments of the type distribution $\{\E(\nu^k)\}_{k=1}^{\bar{D}}$ are identified up to a scale from the exit rates $\{h(d|\ell),\{h(d| \ell')\}_{d=1}^{\bar{D}}$
\end{corr}


Neither of the two results mentioned above is ideal for application to the data. The first result has a limitation in that $h(d|L,X)$ is only well-defined for discrete values of $X$, and even then, it may be imprecisely estimated if each bin size is not large enough. On the other hand, the second result allows us to use duration distributions that are unconditional on $X$, but it imposes a stronger restriction of unconditional independence, which may not hold in the data. 

To address these limitations, I present an additional result below, which allows controlling for observables more flexibly. Specifically, suppose observable characteristics enter the structural hazard proportionally, as in the MPH model. In that case, the model's parameters are identified under conditional independence using inverse propensity score weighted (IPW) exit rates, denoted by  $h^w(d|L)$.\footnote{Note that the length of notice still enters the structural hazard non-proportionally.} Specifically,  $h^w(d|L)$  is defined as follows:

$$h^w(d|L) =  \E\left[\frac{\Pr(D=d|\nu, L, X)}{p_L(X)} \middle| L\right] \ \big / \ \E\left[\frac{\Pr(D \geq d|\nu, L, X)}{p_L(X)} \middle| L \right] $$
where \( p_{\ell}(X) = \Pr(L=\ell|X) \). 

In words, the IPW exit rate \( h^w(d|L) \) is computed by using weighted averages of individual probabilities, where the weights are equal to the probability of the individual receiving a specific notice length based on their observable characteristics \( X \). As shown in the context of treatment effects \citep{Rosenbaum1987}, inverse propensity score weighting leads to the elimination of bias due to selection into treatment based on observable characteristics. This is because it gives less weight to individuals with characteristics that lead to a higher treatment probability, thereby readjusting the sample to reflect the general population more accurately. Therefore, the variation in $h^w(d|L)$ by notice length will capture the direct effect of notice on exit rates without being confounded with the effects of observable characteristics that correlate with notice length. 

Put another way, for any notice length $\ell$, while the unweighted exit rate \( h(d|\ell) \) reflects \( \E[h(d|\nu, \ell, X) \mid D_\ell \geq d, \ell] \), the IPW adjusted exit rate captures \( \E[h(d|\nu, \ell, X) \mid D_\ell \geq d] \). This fact is formally shown in the proof of the subsequent proposition, which presents the identification result using weighted exit rates. It is important to reiterate that, in this context, we must assume that $X$ enters the structural hazard proportionally, such that  \(\psi_L(d, X) = \psi_L(d)\phi(X)\). Consequently, it is useful to define \(\theta(X, \nu) = \phi(X)\nu\), which now represents the combined type of workers based on both observed and unobserved characteristics. 

\begin{prop}\label{result_prop}
Assume that Assumptions \ref{assump_mixed_hazard}--\ref{assump_stationarity} are satisfied, and also that the additional conditions \(\psi_L(d, X) = \psi_L(d)\phi(X)\) and  \( 0<p_L(X)<1\) for all $L, X$ are met. Then for any $\ell$ and $\ell'$, with $\psi_{\ell}(1) \neq \psi_{\ell'}(1)$ and some integer $\bar{D}$, the structural hazards $\{\psi_{\ell}(1),\psi_{\ell'}(1), \{\psi(d) \}_{d=2}^{\bar{D}}\}$ and the moments  $\{\E(\theta(X, \nu)^k)\}_{k=1}^{\bar{D}}$ are identified up to a scale from the weighted exit rates $\{h^w(d|\ell),\{h^w(d| \ell')\}_{d=1}^{\bar{D}}$
\end{prop}

The identification results discussed thus far rely on observing completed unemployment durations. However, as is typical in many datasets, some individuals are still unemployed at the time of the DWS survey. For the unemployed individuals, we observe how long they have been unemployed, but we do not know if and when they will find a job. I now extend Proposition \ref{result_prop} to incorporate right-censored unemployment durations.

Let \( D^c \) denote the censoring time, which is the time elapsed from when an individual becomes unemployed to the time of the survey. For individuals who have already exited unemployment at the time of the survey, we observe their completed unemployment duration $D$ in the data. However, we only observe the censoring time $D^c$ for currently unemployed individuals. Specifically, for each individual, we observe $\tilde{D} = \min\{D,D^c\}$ along with an indicator variable for whether the individual was censored or not, denoted by $C = \I\{D^c < D\}$.

Intuitively, the exit rate at duration \(d\) can be calculated from observed durations as the proportion of individuals who find a job at \(d\), indicated by \(\tilde{D}=d\) and \(C=0\), among those whose observed duration \(\tilde{D}\) is at least equal to \(d\). Since we are limiting ourselves to individuals whose observed duration is at least \(d\), it means that everyone in this group has remained unemployed up to duration \(d\). However, this approach excludes individuals who are censored before \(d\), among whom some may have also remained unemployed until \(d\). Therefore, the exit rate calculated in this manner captures the exit rate for those who are censored after \(d\). Mathematically,
$$ \tilde{h}(d| \Upsilon) = \Pr(\tilde{D}=d , C=0 |\tilde{D} \geq d, \Upsilon)  =  \Pr(D=d|D \geq d, D^c \geq d, \Upsilon) $$
The second equality in the above expression follows from the definition of $\tilde{D}$ and $C$. 

To incorporate inverse propensity score weighting while calculating exit rates based on observed durations and the censoring indicator, we can proceed as before and define the IPW exit rate using observed durations as follows:
\begin{equation*}\label{eqn_ipw_obs_h}
\tilde{h}^w(d| L) = \E\left[\frac{\Pr(\tilde{D}=d, C=0|\nu, L, X)}{p_{L}(X)} \middle | L\right] \ \big/ \ \E\left[\frac{\Pr(\tilde{D} \geq d|\nu, L, X)}{p_{L}(X)} \middle | L \right]
\end{equation*}
The proposition below states that, under the assumptions specified in Proposition \ref{result_prop}, the structural hazards are identified from \( \tilde{h}^w(d| L) \) with the additional assumption \( D^c \perp L \mid X \). Additionally, if we also assume that \( D^c \) is independent of both \( X \) and \( \nu \), moments of \( \theta(X, \nu) \) are also identified. In my application, these assumptions are generally innocuous, as the censoring time is determined by when an individual was surveyed in DWS, which should not correlate with individual outcomes.

\begin{prop}\label{result_prop_main}
Assume that the assumptions stated in Proposition \ref{result_prop} hold and additionally \(D^c \perp L |X\). Then for any $\ell$ and $\ell'$, with $\psi_{\ell}(1) \neq \psi_{\ell'}(1)$ and some integer $\bar{D}$, the structural hazards $\{\psi_{\ell}(1),\psi_{\ell'}(1), \{\psi(d) \}_{d=2}^{\bar{D}}\}$ are identified up to a scale from the weighted exit rates $\{\tilde{h}^w(d|\ell),\{\tilde{h}^w(d| \ell')\}_{d=1}^{\bar{D}}$. Furthermore, if \( D^c \perp \nu \) and \(D^c \perp X\), then the moments  $\{\E(\theta(X, \nu)^k)\}_{k=1}^{\bar{D}}$ are also identified. 
\end{prop}

Proposition \ref{result_prop_main} is the main proposition utilized in the application in the paper. Section \ref{sec_estimation} builds a GMM estimator based on this identification result. 

\subsection{Intuition for Identification}\label{subsec_intuition}

In this section, I elucidate the intuition behind the identification result. To simplify the explanation, I focus on the case without observable characteristics, as incorporating them does not provide any additional insights regarding identification.\footnote{A similar exposition can also be found in \cite{vandenBergvanOurs1996}, where the authors discuss identification utilizing calendar-time effects.} In this model, an individual worker's exit probability is given by $h(d|\ell,\nu) = \psi_{\ell}(d) \nu$ and $\nu$ is independent of $L$. Note that independence implies $f(\nu|L) = f(\nu)$. For brevity, let us denote the first and second moments of $\nu$ by $\mu_1 = \E(\nu)$ and $\mu_2 = \E(\nu^2)$, respectively. It is worth noting that the variance of $\nu$, given by $var(\nu) = \mu_2 - \mu_1^2$, captures the extent of heterogeneity across workers.

To see why the identification result holds, note that the exit rate in the first period is given by \( h(1|\ell) = \psi_{\ell}(1) \mu_1 \) and the exit rate at \( d=2 \) is given by:
$$ h(2|\ell) = \frac{Pr(D=2|\ell)}{Pr(D \geq 2|\ell)} = \psi(2) \left(\frac{\mu_{1} -\psi_{\ell}(1)\mu_{2}}{1-\psi_{\ell}(1) \mu_{1}} \right) = \psi(2)\mu_{1}  \left(\frac{1 -h(1|\ell)(\mu_2/\mu_1^2)}{1-h(1|\ell)} \right)  $$

The third equality in the above equation follows from $\psi_{\ell}(1)=h(1|\ell)/\mu_1$. In the presence of heterogeneity, the variance of $\nu$ is greater than zero, which means that $\mu_2/\mu_1^2 > 1$. Therefore, based on the expressions for $ h(1|\ell)$ and $ h(2|\ell)$, we can observe that $h(2|\ell)/h(1|\ell)$ will always be smaller than $\psi(2)/\psi_{\ell}(1)$. Furthermore, the greater the variance of $\nu$ (i.e., the more heterogeneity across workers), the larger $\mu_2/\mu_1^2$ will be, and the more distant $h(2|\ell)/h(1|\ell)$ will be from $\psi(2)/\psi_{\ell}(1)$. This occurs because greater heterogeneity across workers implies that the composition of workers from the first to the second period changes more drastically. For instance, in the absence of heterogeneity across workers where $\mu_2/\mu_1^2=1$, the composition across both periods is unchanged, and thus $h(2|\ell)/h(1|\ell)=\psi(2)/\psi_{\ell}(1)$.

If we knew the extent of heterogeneity across workers as captured by $\mu_2/\mu_1^2$, we could determine how the composition changes from the first to the second period and estimate the structural duration dependence $\psi(2)/\psi_{\ell}(1)$ from the observed duration dependence $h(2|\ell)/h(1|\ell)$. The variation in notice lengths allows us to learn about the underlying heterogeneity and estimate structural duration dependence. To understand why this is the case, note that for two lengths of notice $\ell$ and $\ell'$, the following expression holds:
$$ \frac{h(2|\ell)}{h(2|\ell')} = \left(\frac{1 -h(1|\ell)(\mu_2/\mu_1^2)}{1-h(1|\ell)} \right) \bigg / \left(\frac{1 -h(1|\ell')(\mu_2/\mu_1^2)}{1-h(1|\ell')} \right) $$

Assuming without loss of generality that $h(1|\ell')>h(1|\ell)$, we can see from the above expression that then $h(2|\ell)/h(2|\ell')\geq 1$. This is because more individuals with notice length $\ell'$ leave in the first period, leading to a worse composition for that group in the second period. Furthermore, when the variance across workers is higher, $h(2|\ell)$ will be further above $h(2|\ell')$. Thus, the difference in exit rates among workers with different notice lengths provides information about the degree of heterogeneity, and we can use the above expression to compute $\mu_2/\mu_1^2$.  Once we know $\mu_2/\mu_1^2$, we can plug that back into the expression for $h(2|\ell)/h(1|\ell)$ and estimate the structural duration dependence $\psi(2)/\psi_{\ell}(1)$. In summary, the difference in exit rates at duration $d=2$ across notice lengths reflects differences in the composition of remaining workers. Therefore, comparing exit rates of workers with different notice lengths can provide insights into the extent to which underlying heterogeneity impacts exit rates. A similar argument applies to identifying structural hazards beyond the second period.\footnote{To understand why higher moments determine the hazard at later durations, we can consider how the composition of workers changes from $d=2$ to $d=3$. This change depends on the level of heterogeneity across workers at the start of $d=2$. If the distribution of heterogeneity has a positive skew, the variance among individuals who survive to $d=2$ would be lower than that among individuals at the start of $d=1$. This is because the few individuals with a high likelihood of exiting unemployment would have already left, reducing the variance among surviving workers.}


\section{Estimation}\label{sec_estimation}

\textbf{Generalized Method of Moments (GMM).} 
Using the identification result presented in Proposition \ref{result_prop_main}, we can use the Generalized Method of Moments (GMM) to construct a consistent estimator for the structural hazards and moments of the heterogeneity distribution. Since the model is identified only up to scale, I normalize the first moment by setting $\E[\theta(X, \nu)]=1$. With $J$ possible notice lengths, the vector of unknown parameters is given by $\Theta = \{\{\psi_{\ell}(1)\}_{\ell=1}^J,\{\psi(d)\}_{d=2}^{\bar{D}},\{\E[\theta(X, \nu)^k]\}_{k=2}^{\bar{D}}\} $ and has a total of $2(\bar{D}-1)+J$ unknown parameters.

Now, for each individual $i$, define the following moment condition:
\[ m_i(\ell,d;\Theta) = \I\{L_i=\ell\} \cdot \left[\frac{\I\{\tilde{D}_i=d\}\I\{C_i=0\}}{p_{\ell}(X_i)}-\tilde{h}^w(d|\ell;\Theta) \cdot \frac{\I\{\tilde{D}_i \geq d\}}{p_{\ell}(X_i)} \right]  \]
Given the definition of \(\tilde{h}^w(d|\ell;\Theta)\), it follows that \(\E[m_i(\ell,d;\Theta)] = 0\), as reasoned in Section \ref{proof_moms_zero}. We can now stack moment conditions pertaining to different notice lengths and durations in one vector, denoted by $ \mathbf{m}_i(\Theta) = \{ \{ m_i(\ell,d,\Theta) \}_{d=1}^{\bar{D}}\}_{\ell=1}^{J} $. Observe that $\mathbf{m}_i(\Theta)$ contains $J \times \bar{D}$ moment conditions. As shown in Proposition \ref{result_prop_main}, our parameters of interest are identified from these moment conditions as long as $J>1$. This condition also ensures that the number of moments is equal to or greater than the number of parameters.

The corresponding sample average to \(\E[m_i(\ell,d;\Theta)]\) can be written as:
\[
  \hat{m}(\ell,d;\Theta) = \frac{1}{n}\sum_{i=1}^n m_i(\ell,d;\Theta) = \hat{h}^{\text{num}}(d|\ell) - \tilde{h}^w(d|\ell;\Theta) \cdot \hat{h}^{\text{den}}(d|\ell)
\]
Here, \(n\) is the sample size, and \( \hat{h}^{num}(d|\ell)\) and \( \hat{h}^{num}(d|\ell)\) are defined as follows:
\[
  \hat{h}^{num}(d|\ell) =  \frac{1}{n} \sum_{i=1}^n \frac{\I\{L_i=\ell\}\I\{\tilde{D}_i=d\}\I\{C_i=0\}}{p_{\ell}(X_i)}, \quad \hat{h}^{den}(d|\ell) =  \frac{1}{n} \sum_{i=1}^n \frac{\I\{L_i=\ell\}\I\{\tilde{D}_i\geq d\}}{p_{\ell}(X_i)}
\]
Note that \(\hat{h}(d|\ell) = \hat{h}^{num}(d|\ell)/\hat{h}^{den}(d|\ell)\) represents the sample counterpart to \(\tilde{h}^w(d|\ell;\Theta)\). As before, stack the sample moments in a single vector \( \hat{\mathbf{m}}(\Theta)\), such that \( \hat{\mathbf{m}}(\Theta) = \sum_{i=1}^n \mathbf{m}_i(\Theta)/n \).

The GMM estimator $\hat{\Theta}$ is then given by:
$ \hat{\Theta} = \arg \max_{\Theta}  \hat{\mathbf{m}}(\Theta)' \hat{W}  \hat{\mathbf{m}}(\Theta) $. When the model is just-identified, $\hat{W}$ is given by the identity matrix. In the case of over-identification, the efficient weighting matrix is given by $\hat{W} =\hat{\Omega}^{-1}$, where $\hat{\Omega}= \left[\frac{1}{n} \sum_{i=1}^n \mathbf{m}_i(\hat{\Theta}) \mathbf{m}_i(\hat{\Theta})' \right]$. Using the two-step estimation process, we can compute $\hat{\Theta}$. The asymptotic distribution of this estimator is given by $ \sqrt{n}(\hat{\Theta}-\Theta) \rightarrow N(0,(\hat{M}' \hat{\Omega}^{-1} \hat{M})^{-1})$, where $\hat{M} = \partial \hat{\mathbf{m}}(\hat{\Theta})/\partial \Theta$. 

\textbf{Functional Form for Structural Hazard.} Even though the model is identified non-parametrically, given small sample sizes, to minimize the number of estimated parameters, I assume that the structural hazard $\psi(d)$ for $d>1$ has a log-logistic form as follows 
\begin{equation}\label{eqn_functional_form}
\psi(d) = \frac{(\alpha_2/\alpha_1) (d/\alpha_1)^{\alpha_2-1}}{1+(d/\alpha_1)^{\alpha_2}} 
\end{equation}
where $\alpha_1>0,\alpha_2>0$. The hazard function in equation (\ref{eqn_functional_form}) is monotonically decreasing when $\alpha_2 \leq 1$ and is unimodal, initially increasing and subsequently decreasing when $\alpha_2>1$. The mode or the turning point is $\alpha_1 (\alpha_2-1)^{1/\alpha_2}$.\footnote{This provides a flexible parametrization for the structural hazard relative to other commonly used parametrization, such as Weibull or Gompertz, as it allows the structural hazard to be non-monotonic. However, I also present non-parametric estimates in  \ref{app_robust}.}


\section{Duration Dependence and Heterogeneity}\label{sec_estimates}

\subsection{Baseline Estimates}

Table \ref{tab_baseline_estimates} presents the main estimates from the Mixed Hazard model. Since I normalized the mean of worker type to equal 1, the estimated structural hazards corresponding to the first period for short and long-notice individuals coincide with their corresponding observed exit rates in the data. The last two lines in panel A of Table \ref{tab_baseline_estimates} show the estimated parameters for the log-logistic function specified in equation (\ref{eqn_functional_form}) used to model structural dependence. 

The structural hazards implied by these parameters are presented in panel B of Table \ref{tab_baseline_estimates} and panel A of Figure \ref{fig_baseline_estimates}. Additionally, panel A of Figure \ref{fig_baseline_estimates} shows the observed exit rate from the data, averaged across workers with short and long notice, alongside the estimated hazard. This figure shows that the estimated hazard consistently exceeds the observed hazard throughout the unemployment spell, indicating the role of underlying heterogeneity. While the observed hazard in the data declines by \ObsHazOnetoTwo\% over the first 24 weeks, the estimated structural hazard only decreases by \EstHazOnetoTwo\% during the same period. Hence, after accounting for heterogeneity, almost three-quarters of the observed decline in the first 24 weeks can be attributed to structural duration dependence.

\begin{table}[t]
\begin{threeparttable}
\caption{Estimation Results}\label{tab_baseline_estimates}
\begin{tabularx}{\linewidth}{Yp{0.5\textwidth}YY}
\toprule
Parameter & Explanation & Estimate & SE \\
\midrule  \addlinespace[1ex]
\multicolumn{4}{l}{\textit{Panel A: Estimated Parameters}}  \\ \addlinespace[1ex]
$\psi_S(1)$ & Structural hazard 0-12 weeks: Short notice & 0.49 & 0.01 \\ 
$\psi_L(1)$ & Structural hazard 0-12 weeks: Long notice & 0.56 & 0.01 \\ 
$\alpha_1$ & Scale parameter for $\psi(d)$ & 2.06 & 0.17 \\ 
$\alpha_2$ & Shape parameter for $\psi(d)$ & 2.54 & 0.27 \\ 
 \addlinespace[1ex]
\multicolumn{4}{l}{\textit{Panel B: Duration Dependence}}  \\ \addlinespace[1ex]
$\bar{\psi}(1)$ & Structural hazard: 0-12 weeks & 0.53 & 0.01 \\ 
$\psi(2)$ & Structural hazard: 12-24 weeks & 0.35 & 0.07 \\ 
$\psi(3)$ & Structural hazard: 24-36 weeks & 0.61 & 0.09 \\ 
$\psi(4)$ & Structural hazard: 36-48 weeks & 0.61 & 0.09 \\ 
   \\ 
\end{tabularx}
\begin{tabularx}{\linewidth}{p{1cm}XX}
\multicolumn{3}{l}{\textit{Hansen-Sargan Test}}  \\ \addlinespace[1ex]
 & Test statistic: 0.01 & Critical value, $df=1, \chi_{0.05}^2$: 3.84 \\
   
\bottomrule
\end{tabularx}
\begin{tablenotes}
\item \textit{Note:} The table presents estimates from the Mixed Hazard model. The first moment is normalized to one, and structural duration dependence is specified by equation (\ref{eqn_functional_form}). Panel A shows the estimated parameters from the model, and panel B presents structural hazards implied by the estimated parameters. The standard errors for the structural hazards are calculated using the delta method.
\end{tablenotes}
\end{threeparttable}
\end{table}

\begin{figure}[t]\caption{Baseline Estimates}\label{fig_baseline_estimates}
\centering
\begin{subfigure}{.49\linewidth}
\raggedleft
\includegraphics{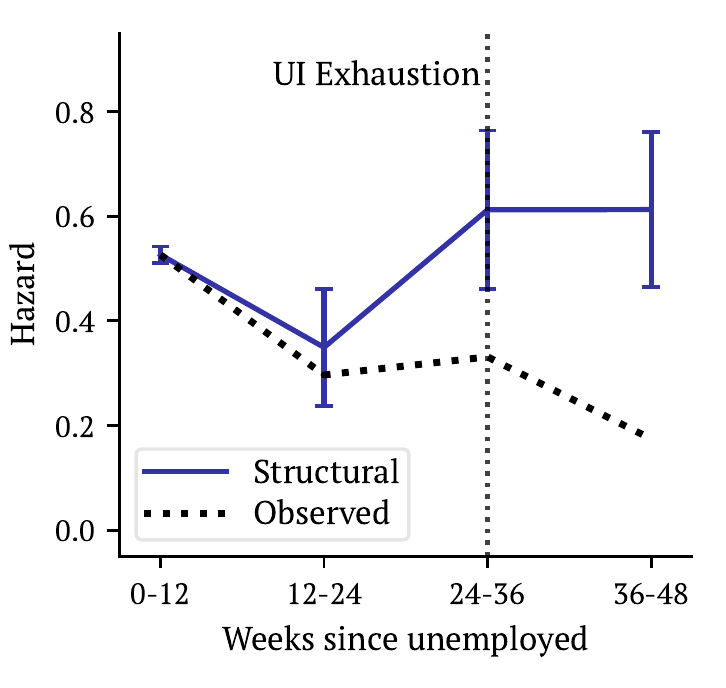}
\subcaption{Structural Hazard}
\end{subfigure} \hfill
\begin{subfigure}{.49\linewidth}
\includegraphics{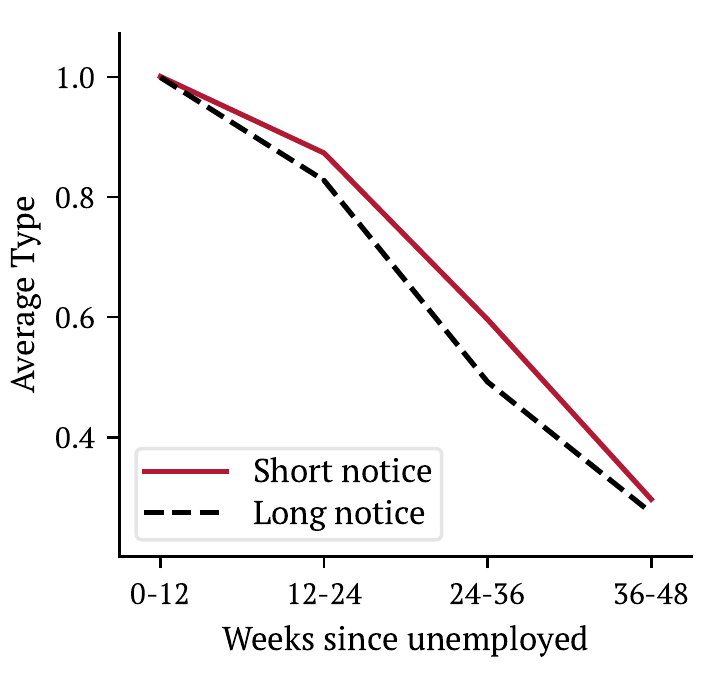}
\subcaption{Average Type}
\end{subfigure} 
\vspace{-0.75em}
\floatfoot{\textit{Note:} The solid blue line in panel A presents estimates for structural hazards as implied by the estimated parameters in panel A of Table \ref{tab_baseline_estimates}. The dotted black line in panel A presents the observed exit rate from the data, averaged across workers with short and long notice. Panel B presents the implied average type at each duration for those with short and long notice. Error bars represent 90\% confidence intervals.}
\end{figure}

However, the estimated structural hazard increases by \EstHazTwotoThree\% from 12-24 to 24-36 weeks, a much more pronounced increase than the observed hazard. This pattern possibly reflects individuals approaching the exhaustion of their unemployment insurance (UI) benefits.\footnote{Since the sample consists of displaced workers, a substantial portion of these individuals should be eligible for UI benefits. Table \ref{tab_uiben_recd} in the  Appendix shows that more than 80\% of individuals unemployed for longer than 12 weeks report receiving UI benefits in the sample.} As previously noted, there is variation across individuals in the eligible duration of UI receipt. However, a significant proportion of individuals are eligible for UI benefits that last for 26 weeks, which coincides with the third interval. Figure \ref{fig_uiex_break} in the  Appendix shows that the proportion of individuals reporting exhausting their UI benefits jumps up significantly in the third interval. This finding of increasing structural hazard leading up to benefit exhaustion is consistent with individuals intensifying their job search efforts or lowering their expectations to secure employment before depleting their benefits. Finally, while the observed hazard continues declining even after 36 weeks, the structural hazard remains constant.

Panel B of Figure \ref{fig_baseline_estimates} displays the average type implied by the model for individuals with short and long notice periods throughout the unemployment spell. For both groups, the average type deteriorates over the course of unemployment. However, for individuals with longer notice periods, the composition worsens more between 0-12 and 12-24 weeks, indicating a higher exit rate in the initial period for this group. By the end of 36 weeks, when a significant number of individuals have already left unemployment, there is little difference in the average type between the two groups.

Overall, the estimated pattern of structural duration dependence aligns with existing evidence from audit studies on call-back rates and with the predictions of search theory. I find that individual exit probabilities decline during the first five months, which can be attributed to duration-based employer discrimination. Further, I find that an individual's exit probability increases leading up to benefit exhaustion and remains constant after. This is consistent with search theory, which predicts that individuals increase their search effort or lower their reservation wages until they reach benefit exhaustion. After that point, if there are no further changes in workers' incentives, their probability of exiting unemployment should remain constant. Interestingly, in their audit study, \cite{KroftEtAl2013} find a decrease in callback rates only during the first six months of unemployment (refer to Figure 2 in their paper). In Section \ref{sec_search_model}, I formally illustrate that my findings are consistent with a search model incorporating heterogeneous workers and falling callback rates early in the unemployment spell.

In recent studies, researchers have introduced behavioral modifications to the standard search theory in order to explain the observed decline in exit rates after UI exhaustion, which deviates from the predictions of the standard search model. Most notably, \cite{DellaVignaEtAl2021} introduces reference dependence in utility to account for this decline. However, after adjusting for compositional effects, I do not find evidence of a decline in individual exit probabilities after UI exhaustion. Hence, I show that the data can be reconciled with the standard model by incorporating heterogeneous workers.


\subsection{Robustness Checks}\label{subsec_robust}

In this section, I examine the robustness of the main results by presenting non-parametric estimates, results using different moments for estimation, and findings from an extension of the MH model that relaxes the identifying assumptions.

\ref{app_robust} presents results from several robustness checks. Figure \ref{fig_robust_unwtd} displays the data and estimates using the unweighted sample. The results reveal that the exit rates and estimates for the unweighted sample are very similar to those obtained using the weighted sample, suggesting that the observable characteristics play a limited role in explaining the differences in exit rates between notice groups. Figure \ref{fig_robust_ff} compares the non-parametric estimate for the structural hazard with the baseline log-logistic estimate. The non-parametric estimates are practically equivalent to the baseline but have larger standard errors, particularly for the last data point where the standard error increases significantly.

I also test the robustness of my findings by including additional notice length categories to the original estimation sample. First, I add individuals with no notice as a third group, and then I repeat the process by including individuals with <1 month notice as the third group. Details for these analyses are provided in Section \ref{subsec_alt_notcats} in the Appendix, and estimates are presented in Figure \ref{fig_robust_notice_cat}. The results show that the estimated structural hazard is qualitatively similar to the baseline estimates in both cases. However, when adding the <1 month notice group, the increase in the structural hazard leading up to benefit exhaustion is less pronounced. Note that the model with three notice categories is overidentified with 4 degrees of freedom. The Sargan-Hansen J-statistic for testing overidentifying restrictions is reported in both cases and indicates no evidence against the null hypothesis of a correctly specified model.

For the main estimates, unemployment durations were grouped into 12-week intervals. I now assess the robustness of these findings to different bin sizes through two exercises. Firstly, I present estimates from the MH model using unemployment duration data binned into 9-week intervals. The estimated structural hazard, shown in Figure \ref{fig_robust_altbins}, mirrors the pattern reported in the main results. Specifically, it decreases from the first to the second interval, followed by an increase in the third interval spanning 18-27 weeks, corresponding to UI exhaustion at 26 weeks. 

Secondly, I conduct a simulation where I calculate exit rates using various bin sizes under a specified data-generating process. Using these differently binned exit rates, I estimate the MH model to examine how the estimates are affected by the choice of bin size. Additional details for this simulation are presented in Section \ref{subsec_binning} in the Appendix, and the results are illustrated in Figures \ref{fig_sim_binA} and \ref{fig_sim_binB}. The key takeaway is that the binned estimates reflect the cumulative structural hazard within an interval, representing the probability of an individual exiting unemployment at some point during that interval. Given this, while we may miss some intra-interval changes, the binned estimates still capture a meaningful quantity and are informative about the underlying duration dependence.

Finally, in \ref{app_gen}, I consider the unconditional model as in Corollary \ref{result_uncond} and show identification under more general conditions than independence and stationarity. Specifically, with \(h(d|\nu, L) = \psi_L(d) \nu\) and considering two notice lengths, \(\ell\) and \(\ell'\), I define the following two quantities:

\[ \kappa_d = E(\nu^d|\ell') - E(\nu^d|\ell), \quad \quad \gamma_d = \frac{\psi_{\ell'}(d)}{\psi_{\ell}(d)} \]
\(\kappa_d\) is the difference between the \(d^{th}\) moment of \(\nu\) for the two lengths of notice, while \(\gamma_d\) is the ratio of structural hazards at duration \(d\) for those lengths. In other words, these parameters determine how the distribution of heterogeneity and structural hazard vary by notice length. I then show that if \(\kappa_d\) is known for all \(d\) and \(\gamma_d\) for \(d >1\), the remaining parameters of the  model are identified. The result can be straightforwardly extended to the model in Proposition \ref{result_prop_main}. In fact, in the applications of this result described below, I utilize the IPW exit rates constructed using observed durations.


I utilize this result to test the independence and stationarity assumptions as follows: First, I relax the independence assumption by allowing the mean of the heterogeneity distribution to vary by notice length (\(\kappa_1 \neq 0\)) while assuming stationarity (\(\gamma_d = 1 \) for \(d>1\)).\footnote{In this exercise, I allow only the mean of the heterogeneity distribution to vary while keeping the shape of the distribution, as determined by the central moments, consistent across the two groups. Note that \(\kappa_d\) for \(d > 1\) will still be non-zero because non-central moments depend on scale, so I let them vary accordingly to ensure that the central moments, apart from the mean, are the same for both groups.} Specifically, I reestimate the model for varying values of \(\kappa_1\) to see where the residuals are minimized, which tells us which value of \(\kappa_1\) is most consistent with the data under the given assumptions. Second, while maintaining independence (setting all \(\kappa_d = 0\)), I relax stationarity by assuming \(\gamma_d = \gamma\) for all \(d > 1\). Once again, I reestimate the model, varying \(\gamma\) values to find the residual-minimizing point. The results for these exercises are presented in Figures \ref{fig_genHet} and \ref{fig_genStr}. In both cases, the residuals plotted as a function of the parameters resemble a convex function with a clear local minimum. Moreover, the residuals are minimized when \(\kappa_1 \approx 0\) and \(\gamma \approx 1\), supporting the identifying assumptions. Finally, as a last exercise, I relax both of these assumptions simultaneously by reestimating the model for a grid of \(\kappa_1\) and \(\gamma\) values. The results, presented in Figure \ref{fig_genHetStr}, still show the residual-minimizing values close to 0 and 1, respectively. Overall, this provides evidence in favor of the identifying assumptions employed in the paper.


\section{A Model of Job-Search}\label{sec_search_model}

The estimates obtained from the Mixed Hazard model suggest a decline in an individual worker's probability of exiting unemployment during the initial five months. Additionally, I find evidence that an individual's likelihood of exiting unemployment increases as they approach the exhaustion of unemployment insurance (UI) benefits and remains constant after that. The latter is in contrast to the observed exit rate, which continues to decline even after benefit exhaustion.  Researchers have tried to explain the decline in the observed rate after exhaustion using behavioral explanations such as storable offers \citep{BoonevanOurs2012} or reference dependence in utility \citep{DellaVignaEtAl2021}. In this section, I show that my findings align with standard search theory, incorporating heterogeneous workers, and are consistent with evidence from the audit study conducted by \cite{KroftEtAl2013}, which documents an initial decline in callback rates during the unemployment spell. 

In particular, I set up a search model with heterogeneous workers. Within this model, workers choose search effort to maximize their expected utility. The likelihood of finding a job depends on the offer arrival rate and a worker's search effort. Moreover, the offer arrival rate varies by the duration of unemployment and the type of worker. I calibrate the model to match the implied structural dependence to my estimate from the Mixed Hazard model and also match the exit rate implied by the model to the data. I then examine the trajectory of the offer arrival rate and search effort. This exercise also allows me to discern the impact on exit probabilities arising from two sources: the actions of optimizing agents in response to changing incentives and external factors directly influencing a worker's employment prospects. 

\subsection{Model Setup}
I consider a stylized model of job search where a worker's search environment is non-stationary \citep{Mortensen1986, VanDenBerg1990} and workers are heterogeneous. At every duration $d$, workers choose how much search effort $s$ to exert to maximize their discounted expected utility.\footnote{Alternatively, the model could feature a reservation wage choice, and all conclusions about search effort would instead be regarding reservation wages.} Costs of search effort are given by the function $c(s)$, which is increasing, convex, and twice continuously differentiable, with $c(0)=0$ and $c'(0)=0$. The probability that a worker finds a job $\lambda(d,\nu,s)$ depends on the time elapsed since unemployed $d$, their search effort $s$, and their type $\nu$ as follows: $\lambda(d,\nu,s) = \delta(d)\nu s $. Here, $\delta(d) \nu$ is the offer arrival rate, which varies over the duration of unemployment and across workers of different types. Once workers find a job, they remain employed forever. A worker receives unemployment insurance (UI) benefits $b(d)$ when unemployed and wages $w$ when employed. The function $u(.)$ gives the flow utility from consumption. Then the value function for a worker of type $\nu$ unemployed at duration $d$ is given by:
$$ V_u(d,\nu) = \max_{s} \; u(b(d)) -c(s)+ \beta \left[\lambda(s, d,\nu) V_e+ \left(1-\lambda(s, d,\nu) \right) V_u(d+1,\nu)  \right] $$

Here, $\beta$ is the discount rate, and $V_e$ is the value of employment given by $ V_e = u(w) + \beta V_e $. The UI benefits $b(d)$ are equal to $b$ for $d \leq D_B$ and equal to 0 otherwise. I also assume that after some time $D_T \geq D_B $ the job-finding function $\lambda(d,s,\nu)$ does not depend on the duration of unemployment $d$, such that for $d>D_T$, $\delta(d)=\delta_T$. This ensures that after $D_T$, jobseekers face a stationary environment, and hence, we can solve for the optimal search strategy of each worker in each period using backward induction. Finally, I consider the case of two types of workers: a high type $H$ and a low type $L$ with $\nu_H>\nu_L$, with $\pi$ denoting the share of workers with the higher arrival rate.

\subsection{Numerical Analysis}

I now calibrate the model specified in the previous section. Let $s(d,\nu)$ denote a worker's optimal search effort at duration $d$. The probability that this worker finds a job $h(d|\nu)$ is given by $\delta(d) s(d,\nu)\nu$. So, a worker's exit rate evolves over the spell of unemployment due to changes in the offer arrival rate $\delta(d)$ and the worker's search effort. However, just as before, the observed exit rate $h(d)=\E[h(d|\nu)|D\geq d]$ also changes due to changes in composition over the spell of unemployment. I use my estimate of the structural hazard from the Mixed Hazard model to target structural duration dependence $\E[h(d|\nu)]$ from the model.\footnote{Note that the search model does not correspond precisely to the econometric framework since it does not imply that $s(d,\nu)$ evolves in the same manner for each type of worker. However, in  \ref{app_sm_sim}, I simulate data from the search model with notice periods and show that my estimator does reasonably well in capturing movements in $\E[h(d|\nu)]$.} Additionally, I match the exit rate implied by the model $h(d)$ to the data. In order to compare the predictions from this model to a model with no heterogeneity, I also calibrate the model assuming just one type of worker. In this case, I match the structural duration dependence or the exit rate implied by the model to the exit rate in the data. Further details for the calibration are provided in \ref{app_search_model}. 

\begin{figure}[t]\caption{Calibration of the Search Model}\label{fig_calibration}
\centering
\begin{subfigure}{.475\textwidth}
\includegraphics{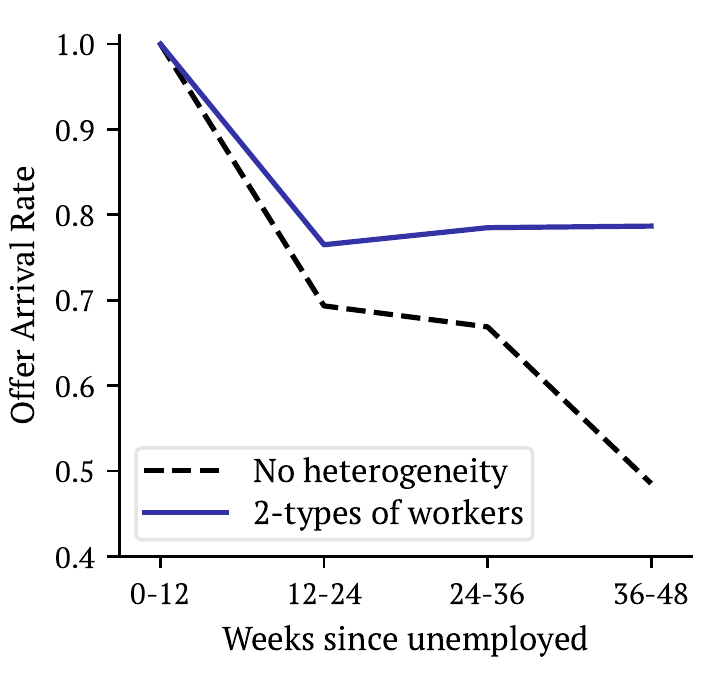}
\subcaption{Offer Arrival Rate}
\end{subfigure} 
\begin{subfigure}{.475\textwidth}
\includegraphics{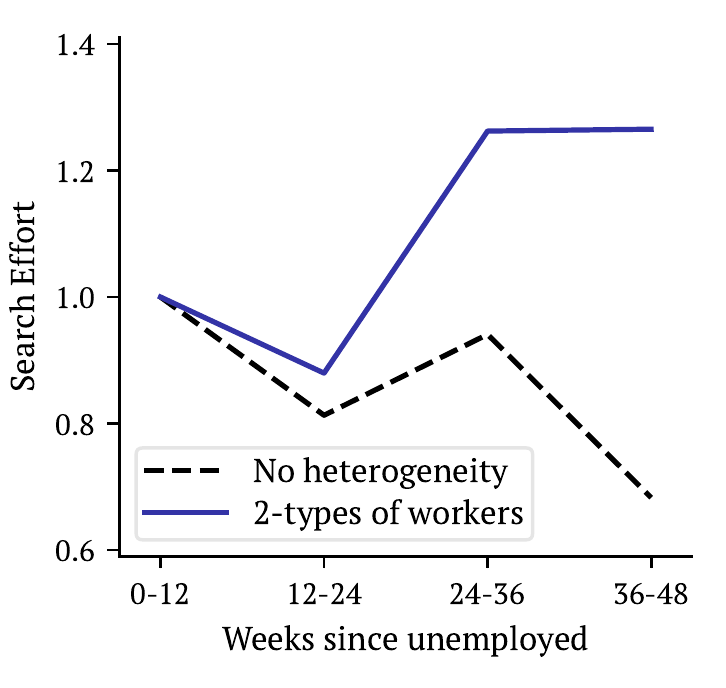}
\subcaption{Search Effort}
\end{subfigure}
\floatfoot{\textit{Notes}: The figure presents the search effort and the offer arrival rate from the calibration of the search model, assuming no heterogeneity (dashed black line) and assuming two types of workers (solid blue line). The search effort is averaged over two types of workers.}
\end{figure}

Figure \ref{fig_calib_fit} shows that both the model with and without heterogeneity fit the data almost perfectly. Figure \ref{fig_calibration} presents the search effort and the offer arrival rate implied by the two calibration exercises. The offer arrival rate implied by the model with heterogeneity declines during the first five months and is constant after that, which is consistent with evidence from \cite{KroftEtAl2013}.  Conversely, in the model with only one type of worker, the offer arrival rate continues to decline throughout the spell of unemployment. Finally, the model calibration implies that an individual's search effort decreases slightly during the first five months but then increases up to UI exhaustion and remains stable after that. In summary, the data and my findings can be rationalized with conventional search theory, without any behavioral adjustments, but by incorporating heterogeneity among workers and declining returns to search early in the unemployment spell.


\section{Conclusion}\label{sec_conclusion}

In this paper, I use a novel source of variation to disentangle the role of structural duration dependence from heterogeneity in the dynamics of the observed exit rate. I document that workers who receive a longer notice before being laid off are more likely to exit unemployment early in the spell. However, the observed exit rate is lower for long-notice workers at later durations. This points towards the presence of heterogeneity across workers. As a higher proportion of the more employable workers from the long-notice group exit early, the composition of surviving long-notice workers at later durations is worse. I utilize these reduced-form moments and estimate a Mixed Hazard model. 

The estimates from the hazard model uncover substantial heterogeneity in individual exit probabilities. The observed exit rate declines by about \ObsHazOnetoTwo\% over the first five months. In contrast, the estimated individual hazard only declines by \EstHazOnetoTwo\% over this period. Moreover, I find that after the first five months, none of the depreciation in the observed exit rate is due to structural duration dependence. Instead, an individual's exit probability increases up to UI exhaustion and remains constant after. The observed exit rate continues to decline after exhaustion as well, which has led researchers to suggest behavioral explanations for this pattern. I show that the alternative explanation of heterogeneity across individuals effectively accounts for this pattern. Specifically, as shown in the paper, my estimates can be rationalized within a standard search model with heterogeneous workers. These findings underscore the importance of incorporating heterogeneity when estimating and calibrating search models.


\clearpage
\setstretch{1.25} 
\bibliography{references.bbl} 

\begin{thebibliography}{}

\bibitem[\protect\citeauthoryear{Abbring}{Abbring}{2012}]{Abbring2012}
Abbring, J.~H. (2012).
\newblock Mixed {{Hitting-Time Models}}.
\newblock {\em Econometrica\/}~{\em 80\/}(2), 783--819.

\bibitem[\protect\citeauthoryear{Alvarez, Borovi{\v c}kov{\'a}, and
  Shimer}{Alvarez et~al.}{2016}]{AlvarezEtAl2016}
Alvarez, F.~E., K.~Borovi{\v c}kov{\'a}, and R.~Shimer (2016, April).
\newblock Decomposing {{Duration Dependence}} in a {{Stopping Time Model}}.
\newblock Working Paper w22188, {National Bureau of Economic Research}.

\bibitem[\protect\citeauthoryear{Alvarez, Borovi{\v c}kov{\'a}, and
  Shimer}{Alvarez et~al.}{2021}]{AlvarezEtAl2021}
Alvarez, F.~E., K.~Borovi{\v c}kov{\'a}, and R.~Shimer (2021, July).
\newblock Consistent evidence on duration dependence of price changes.
\newblock Working Paper 29112, National Bureau of Economic Research.

\bibitem[\protect\citeauthoryear{Boone and {van Ours}}{Boone and {van
  Ours}}{2012}]{BoonevanOurs2012}
Boone, J. and J.~C. {van Ours} (2012, December).
\newblock Why is {{There}} a {{Spike}} in the {{Job Finding Rate}} at {{Benefit
  Exhaustion}}?
\newblock {\em De Economist\/}~{\em 160\/}(4), 413--438.

\bibitem[\protect\citeauthoryear{Brinch}{Brinch}{2007}]{Brinch2007}
Brinch, C.~N. (2007, April).
\newblock Nonparametric {{Identification}} of the {{Mixed Hazards Model}} with
  {{Time-Varying Covariates}}.
\newblock {\em Econometric Theory\/}~{\em 23\/}(2), 349--354.

\bibitem[\protect\citeauthoryear{Burdett and Vishwanath}{Burdett and
  Vishwanath}{1988}]{BurdettVishwanath1988}
Burdett, K. and T.~Vishwanath (1988, October).
\newblock Declining {{Reservation Wages}} and {{Learning}}.
\newblock {\em The Review of Economic Studies\/}~{\em 55\/}(4), 655--665.

\bibitem[\protect\citeauthoryear{Coles and Smith}{Coles and
  Smith}{1998}]{ColesSmith1998}
Coles, M.~G. and E.~Smith (1998).
\newblock Marketplaces and {{Matching}}.
\newblock {\em International Economic Review\/}~{\em 39\/}(1), 239--254.

\bibitem[\protect\citeauthoryear{Cox}{Cox}{1972}]{Cox1972}
Cox, D. (1972).
\newblock Regression {{Models}} and {{Life}}-{{Tables}}.

\bibitem[\protect\citeauthoryear{DellaVigna, Heining, Schmieder, and
  Trenkle}{DellaVigna et~al.}{2021}]{DellaVignaEtAl2021}
DellaVigna, S., J.~Heining, J.~F. Schmieder, and S.~Trenkle (2021, 10).
\newblock {Evidence on Job Search Models from a Survey of Unemployed Workers in
  Germany*}.
\newblock {\em The Quarterly Journal of Economics\/}~{\em 137\/}(2),
  1181--1232.

\bibitem[\protect\citeauthoryear{DellaVigna, Lindner, Reizer, and
  Schmieder}{DellaVigna et~al.}{2017}]{DellaVignaEtAl2017}
DellaVigna, S., A.~Lindner, B.~Reizer, and J.~F. Schmieder (2017, November).
\newblock Reference-{{Dependent Job Search}}: {{Evidence}} from {{Hungary}}*.
\newblock {\em The Quarterly Journal of Economics\/}~{\em 132\/}(4),
  1969--2018.

\bibitem[\protect\citeauthoryear{Elbers and Ridder}{Elbers and
  Ridder}{1982}]{ElbersRidder1982}
Elbers, C. and G.~Ridder (1982, July).
\newblock True and {{Spurious Duration Dependence}}: {{The Identifiability}} of
  the {{Proportional Hazard Model}}.
\newblock {\em The Review of Economic Studies\/}~{\em 49\/}(3), 403--409.

\bibitem[\protect\citeauthoryear{Farber}{Farber}{2017}]{Farber2017}
Farber, H.~S. (2017).
\newblock Employment, {{Hours}}, and {{Earnings Consequences}} of {{Job Loss}}:
  {{US Evidence}} from the {{Displaced Workers Survey}}.
\newblock {\em Journal of Labor Economics\/}~{\em 35\/}(S1), 235--272.

\bibitem[\protect\citeauthoryear{Ganong and Noel}{Ganong and
  Noel}{2019}]{GanongNoel2019}
Ganong, P. and P.~Noel (2019, July).
\newblock Consumer {{Spending}} during {{Unemployment}}: {{Positive}} and
  {{Normative Implications}}.
\newblock {\em American Economic Review\/}~{\em 109\/}(7), 2383--2424.

\bibitem[\protect\citeauthoryear{Gonzalez and Shi}{Gonzalez and
  Shi}{2010}]{GonzalezShi2010}
Gonzalez, F.~M. and S.~Shi (2010).
\newblock An {{Equilibrium Theory}} of {{Learning}}, {{Search}}, and {{Wages}}.
\newblock {\em Econometrica\/}~{\em 78\/}(2), 509--537.

\bibitem[\protect\citeauthoryear{Heckman and Singer}{Heckman and
  Singer}{1984}]{HeckmanSinger1984}
Heckman, J. and B.~Singer (1984).
\newblock A {{Method}} for {{Minimizing}} the {{Impact}} of {{Distributional
  Assumptions}} in {{Econometric Models}} for {{Duration Data}}.
\newblock {\em Econometrica\/}~{\em 52\/}(2), 271--320.

\bibitem[\protect\citeauthoryear{Honor{\'e}}{Honor{\'e}}{1993}]{Honore1993}
Honor{\'e}, B.~E. (1993).
\newblock Identification {{Results}} for {{Duration Models}} with {{Multiple
  Spells}}.
\newblock {\em The Review of Economic Studies\/}~{\em 60\/}(1), 241--246.

\bibitem[\protect\citeauthoryear{Jarosch and Pilossoph}{Jarosch and
  Pilossoph}{2019}]{JaroschPilossoph2019}
Jarosch, G. and L.~Pilossoph (2019, July).
\newblock Statistical {{Discrimination}} and {{Duration Dependence}} in the
  {{Job Finding Rate}}.
\newblock {\em The Review of Economic Studies\/}~{\em 86\/}(4), 1631--1665.

\bibitem[\protect\citeauthoryear{Katz and Meyer}{Katz and
  Meyer}{1990}]{KatzMeyer1990}
Katz, L.~F. and B.~D. Meyer (1990, February).
\newblock The impact of the potential duration of unemployment benefits on the
  duration of unemployment.
\newblock {\em Journal of Public Economics\/}~{\em 41\/}(1), 45--72.

\bibitem[\protect\citeauthoryear{Kolsrud, Landais, Nilsson, and
  Spinnewijn}{Kolsrud et~al.}{2018}]{KolsrudEtAl2018}
Kolsrud, J., C.~Landais, P.~Nilsson, and J.~Spinnewijn (2018, April).
\newblock The {{Optimal Timing}} of {{Unemployment Benefits}}: {{Theory}} and
  {{Evidence}} from {{Sweden}}.
\newblock {\em American Economic Review\/}~{\em 108\/}(4-5), 985--1033.

\bibitem[\protect\citeauthoryear{Kroft, Lange, and Notowidigdo}{Kroft
  et~al.}{2013}]{KroftEtAl2013}
Kroft, K., F.~Lange, and M.~J. Notowidigdo (2013, August).
\newblock Duration {{Dependence}} and {{Labor Market Conditions}}: {{Evidence}}
  from a {{Field Experiment}}*.
\newblock {\em The Quarterly Journal of Economics\/}~{\em 128\/}(3),
  1123--1167.

\bibitem[\protect\citeauthoryear{Krueger and Mueller}{Krueger and
  Mueller}{2011}]{KruegerMueller2011}
Krueger, A.~B. and A.~Mueller (2011).
\newblock Job {{Search}}, {{Emotional Well-Being}}, and {{Job Finding}} in a
  {{Period}} of {{Mass Unemployment}}: {{Evidence}} from {{High Frequency
  Longitudinal Data}} [with {{Comments}} and {{Discussion}}].
\newblock {\em Brookings Papers on Economic Activity\/}, 1--81.

\bibitem[\protect\citeauthoryear{Lancaster}{Lancaster}{1979}]{Lancaster1979}
Lancaster, T. (1979).
\newblock Econometric {{Methods}} for the {{Duration}} of {{Unemployment}}.
\newblock {\em Econometrica\/}~{\em 47\/}(4), 939--956.

\bibitem[\protect\citeauthoryear{Lentz and Tran{\ae}s}{Lentz and
  Tran{\ae}s}{2005}]{LentzTranaes2005}
Lentz, R. and T.~Tran{\ae}s (2005, July).
\newblock Job {{Search}} and {{Savings}}: {{Wealth Effects}} and {{Duration
  Dependence}}.
\newblock {\em Journal of Labor Economics\/}~{\em 23\/}(3), 467--489.

\bibitem[\protect\citeauthoryear{Machin and Manning}{Machin and
  Manning}{1999}]{MachinManning1999}
Machin, S. and A.~Manning (1999, January).
\newblock Chapter 47 {{The}} causes and consequences of longterm unemployment
  in {{Europe}}.
\newblock In {\em Handbook of {{Labor Economics}}}, Volume~3, pp.\  3085--3139.
  {Elsevier}.

\bibitem[\protect\citeauthoryear{Marinescu and Skandalis}{Marinescu and
  Skandalis}{2021}]{MarinescuSkandalis2021}
Marinescu, I. and D.~Skandalis (2021, May).
\newblock Unemployment {{Insurance}} and {{Job Search Behavior}}*.
\newblock {\em The Quarterly Journal of Economics\/}~{\em 136\/}(2), 887--931.

\bibitem[\protect\citeauthoryear{Mortensen}{Mortensen}{1986}]{Mortensen1986}
Mortensen, D.~T. (1986, January).
\newblock Chapter 15 {{Job}} search and labor market analysis.
\newblock In {\em Handbook of {{Labor Economics}}}, Volume~2, pp.\  849--919.
  {Elsevier}.

\bibitem[\protect\citeauthoryear{Mueller, Spinnewijn, and Topa}{Mueller
  et~al.}{2021}]{MuellerEtAl2021}
Mueller, A.~I., J.~Spinnewijn, and G.~Topa (2021, January).
\newblock Job {{Seekers}}' {{Perceptions}} and {{Employment Prospects}}:
  {{Heterogeneity}}, {{Duration Dependence}}, and {{Bias}}.
\newblock {\em American Economic Review\/}~{\em 111\/}(1), 324--363.

\bibitem[\protect\citeauthoryear{Pavoni}{Pavoni}{2009}]{Pavoni2009}
Pavoni, N. (2009).
\newblock Optimal {{Unemployment Insurance}}, with {{Human Capital
  Depreciation}}, and {{Duration Dependence}}*.
\newblock {\em International Economic Review\/}~{\em 50\/}(2), 323--362.

\bibitem[\protect\citeauthoryear{Pavoni and Violante}{Pavoni and
  Violante}{2007}]{PavoniViolante2007}
Pavoni, N. and G.~L. Violante (2007, January).
\newblock Optimal {{Welfare-to-Work Programs}}.
\newblock {\em The Review of Economic Studies\/}~{\em 74\/}(1), 283--318.

\bibitem[\protect\citeauthoryear{Pissarides}{Pissarides}{1992}]{Pissarides1992}
Pissarides, C.~A. (1992, November).
\newblock Loss of {{Skill During Unemployment}} and the {{Persistence}} of
  {{Employment Shocks}}*.
\newblock {\em The Quarterly Journal of Economics\/}~{\em 107\/}(4),
  1371--1391.

\bibitem[\protect\citeauthoryear{Rosenbaum}{Rosenbaum}{1987}]{Rosenbaum1987}
Rosenbaum, P.~R. (1987, 2024/07/01/).
\newblock Model-based direct adjustment.
\newblock {\em Journal of the American Statistical Association\/}~{\em
  82\/}(398), 387--394.

\bibitem[\protect\citeauthoryear{Rosenbaum and Rubin}{Rosenbaum and
  Rubin}{1983}]{RosenbaumRubin1983}
Rosenbaum, P.~R. and D.~B. Rubin (1983, April).
\newblock The central role of the propensity score in observational studies for
  causal effects.
\newblock {\em Biometrika\/}~{\em 70\/}(1), 41--55.

\bibitem[\protect\citeauthoryear{Rubin}{Rubin}{1974}]{Rubin1974}
Rubin, D.~B. (1974).
\newblock Estimating causal effects of treatments in randomized and
  nonrandomized studies.
\newblock {\em Journal of educational Psychology\/}~{\em 66\/}(5), 688.

\bibitem[\protect\citeauthoryear{Schmieder, von Wachter, and Bender}{Schmieder
  et~al.}{2016}]{SchmiederEtAl2016}
Schmieder, J.~F., T.~von Wachter, and S.~Bender (2016, March).
\newblock The effect of unemployment benefits and nonemployment durations on
  wages.
\newblock {\em American Economic Review\/}~{\em 106\/}(3), 739--77.

\bibitem[\protect\citeauthoryear{Shimer and Werning}{Shimer and
  Werning}{2006}]{ShimerWerning2006}
Shimer, R. and I.~Werning (2006, May).
\newblock On the {{Optimal Timing}} of {{Benefits}} with {{Heterogeneous
  Workers}} and {{Human Capital Depreciation}}.
\newblock Technical Report w12230, {National Bureau of Economic Research}.

\bibitem[\protect\citeauthoryear{Van Den~Berg}{Van
  Den~Berg}{1990}]{VanDenBerg1990}
Van Den~Berg, G.~J. (1990).
\newblock Nonstationarity in {{Job Search Theory}}.
\newblock {\em The Review of Economic Studies\/}~{\em 57\/}(2), 255--277.

\bibitem[\protect\citeauthoryear{{van den Berg} and {van Ours}}{{van den Berg}
  and {van Ours}}{1996}]{vandenBergvanOurs1996}
{van den Berg}, G.~J. and J.~C. {van Ours} (1996).
\newblock Unemployment {{Dynamics}} and {{Duration Dependence}}.
\newblock {\em Journal of Labor Economics\/}~{\em 14\/}(1), 100--125.

\end{thebibliography}
\setstretch{1.5} 


\appendix
\setcounter{table}{0}
\setcounter{figure}{0}
\setcounter{section}{0}
\renewcommand{\thesection}{Appendix \Alph{section}}
\renewcommand{\thesubsection}{\Alph{section}.\arabic{subsection}}
\renewcommand{\theequation}{\arabic{equation}}
\renewcommand{\thetable}{\Alph{section}\arabic{table}}
\renewcommand{\thefigure}{\Alph{section}\arabic{figure}}

\setlength{\abovedisplayshortskip}{-0.15cm}
\setlength{\belowdisplayshortskip}{0.35cm}
\setlength{\abovedisplayskip}{0.35cm}
\setlength{\belowdisplayskip}{0.35cm}


\clearpage
\section{Proofs and Derivations}\label{app_proofs}

This section provides proofs of the results discussed in the main text. Before proceeding, let us introduce some additional notation. Denote the probability distribution function of observed duration by $g(.)$. Additionally, define the corresponding survival function $S(.)$ to represent the probability that unemployment duration exceeds a specific value.  

At this point, it will be useful to examine the expressions for the quantities defined above under Assumption \ref{assump_mixed_hazard}, which specifies that  \( h(d|\nu, L, X) = \psi_L(d, X) \nu \). First, note that we can express the survival function $S(d|\nu, L,X)$ as follows:
\[ S(d|\nu, L,X) = \Pr(D>d|\nu, L,X) = \prod_{s=1}^d \left[1-\psi_L(s,X)\nu  \right] \]
Given the above expression, the conditional distribution of observed duration is given by:
\begin{align}\label{eq_density_surv}
g(d|\nu, L,X) 
& = \Pr(D=d|\nu, L, X)  \notag \\
& =  S(d-1|\nu, L,X)-S(d|\nu, L,X) \notag \\
&= \psi_L(d,X) \nu S(d-1|\nu, L,X)  
\end{align} 
with $S(0|\nu, L,X)=1$.

Finally, note that according to the definition of the hazard function in the text, for a set of conditioning variables $\Upsilon$, we have:
$$ h(d|\Upsilon) = \Pr(D=d|D \geq d, \Upsilon) = \frac{\Pr(D=d|\Upsilon)}{\Pr(D \geq d|\Upsilon)} = \frac{g(d|\Upsilon)}{S(d-1|\Upsilon)} $$
The above expression illustrates that, given \( h(d|\Upsilon) \) for \( d=1,2,...,\bar{D} \), we can compute \( g(d|\Upsilon) \) and \( S(d|\Upsilon) \) for the same durations, and vice versa. This means that knowing any one of these three quantities enables the calculation of the other two. Therefore, although the statement of results might specify a set of parameters being identified from  $h(d|\Upsilon)$, it does not matter if the proof utilizes derived quantities $g(d|\Upsilon)$  or  $S(d|\Upsilon)$.


\subsection{Proof of Proposition \ref{result_average_type_falls}}\label{proof_average_type_falls}

\begin{proof}
Given equation (\ref{eq_density_surv}), we can write $g(d|L,X) = \psi_L(d, X) \E[ \nu S(d-1|\nu, L,X) |L,X] $. Plugging this in the definition of $h(d|L,X)$, we get:
\[ h(d|L,X) = \psi_L(d, X) \E\left(\nu \cdot \frac{S(d-1|\nu,L,X)}{S(d-1|L,X)} \middle | L,X \right) \]

To see that the second term in the above expression is the average type $\E(\nu| D \geq d, L,X)$ amongst surviving workers at the beginning of $d$, note that
\[f(\nu |D \geq d,L,X) = \frac{\Pr(D>d-1|\nu,L,X)f(\nu |L,X)}{\Pr(D>d-1 | L,X)}=\frac{S(d-1|\nu,L,X)f(\nu|L,X)}{S(d-1|L,X)}\]
where the first inequality follows from the Bayes rule.\\
To see that the average type declines with duration, note that for any $d$ and $\nu_H>\nu_L$, \[ \frac{S(d|\nu_H,L,X)}{S(d-1|\nu_H,L,X)}<\frac{S(d|\nu_L,L,X)}{S(d-1|\nu_L,L,X)} \]
The above equation implies that, 
\[ \frac{f(\nu_H |D \geq d+1,L,X)}{f(\nu_L |D \geq d+1,L,X)}< \frac{f(\nu_H |D \geq d,L,X)}{f(\nu_L |D \geq d,L,X)} \]
In which case, $f(\nu|D \geq d,L,X)$ first-order stochastically dominates $f(\nu|D \geq d+1,L,X)$ which implies that $\E(\nu| D \geq d, L,X) \geq \E(\nu| D \geq d+1, L,X)$.
\end{proof}


\subsection{Statement and Proof of Lemma \ref{result_mom_id}}\label{proof_lemma}

The following lemma states that the identification of structural hazards implies the identification of higher moments of the unobserved type distribution.

\begin{lemma}\label{result_mom_id}
Under Assumption \ref{assump_mixed_hazard}, if $\psi_L(d, X)$ is known for $d=1,..,\bar{D}$, then we can identify the first $\bar{D}$ conditional moments of $\nu$, given by $\{\E(\nu^k|L,X) \}_{k=1}^{\bar{D}}$, from the conditional unemployment distribution $g(d|L,X) $ for $d=1,..,\bar{D}$.
\end{lemma}	
\begin{proof}
Expanding equation (\ref{eq_density_surv}) for $d=1,2,3,...$, we can write: 
\begin{align*}
g(1|\nu, L,X) &= \psi_L(1, X) \nu \\
g(2|\nu, L,X) &= \psi_L(2, X) \big[\nu -\psi_L(1, X)\nu^2 \big]   \\
g(3|\nu, L,X) &= \psi_L(3, X) \big[\nu-[\psi_L(1, X)+\psi_L(2, X)]\nu^2+\psi_L(1, X)\psi_L(2, X)\nu^3\big]   \\
\vdots
\end{align*}
Or, more compactly, 
\begin{align}\label{eqn_compact_moments}
g(d|\nu, L,X) = \psi_L(d, X) \sum_{k=1}^{\bar{D}} c_k(d,\bm{\psi_{L,X}}) \nu^k
\end{align}
where $\bm{\psi_{L,X}} = \{\psi_L(d, X)\}_{d=1}^{\bar{D}}$ and 
$$ c_k(d,\bm{\psi_{L,X}}) =  \begin{cases}
1 & \text{for }  k=1 \\
c_k(d-1, \bm{\psi_{L,X}})-\psi_L(d-1, X) c_{k-1}(d-1, \bm{\psi_{L,X}}) & \text{for }   1 \leq k \leq d \\
0 & \text{for }  k>d \\
\end{cases} $$
Taking the expectation conditional on $L$ and $X$ of the above expression, we can write: 
$$ g(d|L,X) = \psi_L(d, X) \sum_{k=1}^{\bar{D}} c_k(d,\bm{\psi_{L,X}}) \E(\nu^k |L, X) $$
Let us denote $\bm{g_{L,X}} = \{g(d|L,X)\}_{d=1}^{\bar{D}}$ and $\bm{\mu_{L,X}} = \{\E(\nu^k|L,X)\}_{k=1}^{\bar{D}}$.  Then we can write $\bm{g_{L,X}} = C(\bm{\psi_{L,X}}) \bm{\mu_{L,X}}$ where $C(\bm{\psi_{L,X}})$ is the $\bar{D} \times \bar{D}$ upper triangular matrix with $C_{s,k}(\bm{\psi_{L,X}}) = \psi_{L}(s,X) c_k(s,\bm{\psi_{L,X}})$. In addition, the diagonal elements of  $C(\bm{\psi_{L,X}})$ are non-zero. To see this note that, $C_{d,d}(\bm{\psi_{L,X}}) = (-1)^{d-1} \prod_{s=1}^{d} \psi_{L}(s,X) $ and each $\psi_{L}(s,X)>0$. Hence, $C(\bm{\psi_{L,X}})$ is invertible and we can plug in $\bm{\psi_{L,X}}$ in $\bm{g_{L,X}} = C(\bm{\psi_{L,X}}) \bm{\mu_{L,X}}$ to solve for $\bm{\mu_{L,X}}$.
\end{proof}


\subsection{Proof of Theorem \ref{result_main_theorem}}\label{proof_main_theorem}

\begin{proof}
Assumption \ref{assump_stationarity} states that $\psi_L(d, X)=\psi(d, X)$ for $d>1$. For this reason, define $\check{S}(d|\nu,X)$ as: 
$$ \check{S}(d|\nu,X) = \prod_{s=2}^{d} \big[1-\psi(s, X) \nu\big]    $$
Then by Assumption \ref{assump_stationarity}, for $d>1$, we can write:
\begin{equation}\label{eqn_isurv}
S(d|\nu,L,X) = [1-\psi_L(1, X)\nu] \check{S}(d|\nu,X) 
\end{equation}
Next, by substituting the expression for \( S(d-1|\nu, L, X) \) implied by equation (\ref{eqn_isurv}) into the expression for \( g(d|\nu, L, X) \) given in equation (\ref{eq_density_surv}), we obtain:
\begin{equation}\label{eqn_idens}
 g(d|\nu,L,X) = \psi_L(d,X) \left[\nu \check{S}(d-1|\nu,X)-\psi_L(1, X)\nu^2 \check{S}(d-1|\nu,X) \right]  
\end{equation}
Now note that since Assumption \ref{assump_independence} states that \( \nu \) is independent of \( L \) given \( X \), it follows that for any $L=\ell$: 
\[ q(d|\ell ,X) = \E[q(d|\nu, \ell, X) | \ell, X] = \E[q(d|\nu, \ell, X) | X] \quad \text{for } q = S, g\]
Therefore, by taking the conditional expectation of the expressions implied by equations (\ref{eqn_isurv}) and (\ref{eqn_idens}) under Assumptions \ref{assump_independence} and \ref{assump_stationarity}, we get:
\begin{equation}\label{eqn_surv}
S(d-1|L,X) = \E[\check{S}(d-1|\nu,X)|X]-\psi_L(1, X)\E[\nu \check{S}(d-1|\nu,X)|X]
\end{equation}
\begin{equation}\label{eqn_dens}
g(d|L,X) = \psi(d, X) [\E[\nu \check{S}(d-1|\nu,X)|X]-\psi_L(1, X)\E[\nu^2 \check{S}(d-1|\nu,X)|X]]
\end{equation}
Next, consider two notice lengths $\ell$ and $\ell'$. Substituting $L=\ell, \ell'$ in equations (\ref{eqn_surv}) and (\ref{eqn_dens}) and taking the difference between expressions for $\ell'$ and $\ell$, we get:
\begin{align}
\label{eqn_fm} S(d-1|\ell',X)-S(d-1|\ell,X)  = [\psi_{\ell}(1, X)-\psi_{\ell'}(1, X)] \E[\nu \check{S}(d-1|\nu,X)|X] \\
\label{eqn_sm} g(d|\ell',X)-g(d|\ell,X)  =   \psi(d, X) [\psi_{\ell}(1, X)-\psi_{\ell'}(1, X)] \E[\nu^2 \check{S}(d-1|\nu,X)|X] 
\end{align}
From equation (\ref{eqn_fm}), we can write:
\begin{equation}\label{eqn_fexp}
\E[\nu \check{S}(d-1|\nu,X)|X] = \frac{S(d-1|\ell',X)-S(d-1|\ell,X)}{\psi_{\ell}(1, X)-\psi_{\ell'}(1, X)} 
\end{equation}
Similarly, from equation (\ref{eqn_sm}), we obtain:
\begin{equation}\label{eqn_sexp}
\E[\nu^2 \check{S}(d-1|\nu,X)|X]  =  \frac{g(d|\ell',X)-g(d|\ell,X)}{\psi(d, X)(\psi_{\ell}(1, X)-\psi_{\ell'}(1, X))} 
\end{equation}
Plugging terms from equations (\ref{eqn_fexp}) and (\ref{eqn_sexp}) in the expression for $g(d|\ell,X)$ given by equation (\ref{eqn_dens}), implies that for $d>1$:
\begin{equation*}
 \psi(d, X) =   \frac{g(d|\ell',X)\psi_{\ell}(1, X)-g(d|\ell,X)\psi_{\ell'}(1, X)}{S(d-1|\ell',X)-S(d-1|\ell,X)}
\end{equation*}
Here, the denominator is not equal to zero as we assumed $\psi_{\ell'}(1, X)\neq \psi_{\ell}(1, X)$. 

Now note that for $d=1$, $g(1|\ell,X) = \psi_{\ell}(1, X) \E[\nu | X]  $.
So plugging in $\psi_{\ell}(1, X)=g(1|\ell,X)/\E[\nu| X]$ in the expression for $\psi(d, X)$ above, we can write:
\begin{equation*}
\psi(d, X)\E(\nu|X) =   \frac{g(d|\ell',X)g(1|\ell,X)-g(d|\ell,X)g(1|\ell',X)}{S(d-1|\ell',X)-S(d-1|\ell,X)}  
\end{equation*}
Since all the quantities on the right-hand side of the above equation can be derived from the conditional exit rates, it follows that \(\{\psi_{\ell}(1,X), \psi_{\ell'}(1,X), \{\psi(d, X)\}_{d=2}^{\bar{D}}\}\) are identified up to a multiplicative constant from \(\{h(d|\ell,X), h(d|\ell',X)\}_{d=1}^{\bar{D}}\). Identification of conditional moments follows from Lemma \ref{result_mom_id} by noting that $\E(\nu^k| L, X) = \E(\nu^k| L)$ under conditional independence.
\end{proof}

\subsection{Proof of Corollary \ref{result_uncond}}\label{proof_uncond}
\begin{proof}
Given that \( h(d|\nu, L) = \psi_L(d)\nu \), we can write:
\[ S(d|\nu, L) = \prod_{s=1}^d \left[1-\psi_{L}(s)\nu  \right], \quad g(d|\nu, L) = \psi_{L}(d) \nu S(d-1|\nu, L) \]
Moreover, \( L \perp \nu \) implies that for any \(L=\ell\), \( q(d|\ell) = \E[q(d|\nu, \ell) \mid \ell] = \E[q(d|\nu, \ell)] \) for \( q = S, g \). 
Since we assumed that \( \psi_L(d) = \psi(d) \) for \( d > 1 \), let us define as before: \(\check{S}(d|\nu) = \prod_{s=2}^{d} \left[1-\psi(s) \nu\right] \) for \(d>1\).
We then obtain the following expressions:
\[
\begin{aligned}
S(d-1|L) = & \E[\check{S}(d-1|\nu)] - \psi_{L}(1)\E[\nu \check{S}(d-1|\nu)] \\
g(d|L) = \psi(d) & \left[\E[\nu \check{S}(d-1|\nu)] - \psi_{L}(1)\E[\nu^2 \check{S}(d-1|\nu)]\right]
\end{aligned}
\]
Note that the above equations take a similar form to equations (\ref{eqn_surv}) and (\ref{eqn_dens}), with the only difference being that the expressions are free of \(X\). Thus, by considering \( L = \ell, \ell' \) and replicating the steps following equation (\ref{eqn_dens}) in the proof of Theorem \ref{result_main_theorem}, we can derive that \(\psi_L(1) = g(1|L)/\E(\nu)\) for \(L \in \{\ell, \ell'\}\), and
\[
\psi(d)\E(\nu) = \frac{g(d|\ell')g(1|\ell) - g(d|\ell)g(1|\ell')}{S(d-1|\ell') - S(d-1|\ell)} \text{ for } d > 1.
\]
This concludes the proof for identifying structural hazards up to a scale. For the identification of the unconditional moments of the type distribution, we can follow the steps from Lemma \ref{result_mom_id}. In particular, \( \bm{g_{L}} = [g(d|L)]_{d=1}^{\bar{D}} \) can be expressed using as a system of \( \bar{D} \) equations, in terms of the structural hazards \( \bm{\psi_{L}} = [\psi_{L}(d)]_{d=1}^{\bar{D}} \) and the moments of the type distribution \( \bm{\mu} = [\E(\nu^k)]_{k=1}^{\bar{D}} \). Specifically, \( \bm{g_{L}} = C(\bm{\psi_{L}}) \bm{\mu} \), where \( C(\bm{\psi_{L}}) \) is an invertible matrix, as outlined in Lemma \ref{result_mom_id}. Hence, we can plug in the identified structural hazards to derive the moments 
\end{proof}

\subsection{Proof of Proposition \ref{result_prop}}\label{proof_prop}

\begin{proof}
Define \(g^w(.)\) as the weighted distribution and \(S^w(.)\) as the weighted survival function. Specifically, \(q^w(d|L) = \E[q(d|\nu, L,X)/p_L(X)|L]\) for \(q=g, S\). According to the definition in the main text, \(h^{w}(d|L) = g^w(d|L)/S^w(d-1|L)\). Therefore,  given \( h^{w}(d|L) \) for \( d=1,2,...,\bar{D} \), we can compute \( g^w(d|L) \) and \( S^w(d|L) \) for the same durations as well.

We can show  that for any $L=\ell$, \(q^w(d|\ell) = \E[q(d|\nu, \ell, X)]/\pi_{\ell}\) for \(q=g, S\), under conditional independence, by following the steps outlined below:
\begin{align*}
 \E\left[\frac{q(d|\nu, \ell,X)}{p_{\ell}(X)} \middle|\ell\right] & =  \E\left[\E\left[\frac{q(d|\nu, \ell,X)}{p_{\ell}(X)}\middle|\ell, X\right]  \middle|\ell\right] & (i)  \\
& = \E\left[\E\left[\frac{q(d|\nu, \ell,X)}{p_{\ell}(X)}\middle| X\right]  \middle|\ell\right] & (ii) \\
&= \int \frac{\E[q(d|\nu, \ell,x)|x]}{p_{\ell}(x)}\cdot f_{X|\ell}(x) \ \partial x & (iii) \\
& = \frac{1}{\pi_{\ell}} \cdot \E[q(d|\nu, \ell,x)] & (iv)
\end{align*}
Here, \(\pi_{\ell} = \Pr(L=\ell)\) and \(f_{X|L}(.)\) represents the distribution of \(X\) given \(L\). Step (i) utilizes the law of iterated expectations, step (ii) follows from \(L \perp \nu | X\), step (iii) uses the definition of conditional expectation, and step (iv) follows from the Bayes' rule, which implies that \(f_{X|\ell}(x) = p_{\ell}(x) f_X(x)/\pi_{\ell}\).

As an aside, note that this implies that:
\begin{align*}
h^{w}(d|\ell) &= \frac{\E[g(d|\nu, \ell, X)]}{\E[S(d-1|\nu, \ell, X)]} = \frac{\E[\Pr(D_{\ell}=d |\nu, X)]}{\E[\Pr(D_{\ell} \geq d |\nu, X)]} = \frac{\Pr(D_{\ell}=d)}{\Pr(D_{\ell} \geq d )} =  \Pr(D_{\ell}=d|D_{\ell} \geq d)  
\end{align*}

Having shown that \(\pi_{\ell} q^w(d|{\ell}) = \E[q(d|\nu, \ell, X)]\) for \(q=g, S\), if we now take the expectation of the expressions implied by equations (\ref{eqn_isurv}) and (\ref{eqn_idens}) and substitute \(\psi_L(d, X) = \psi_L(d) \phi(X)\) with \(\psi_L(d) = \psi(d)\) for \(d>1\), we obtain:
\[
\begin{aligned}
\pi_L S^w(d-1|L) &=  \E[\check{S}(d-1|\nu, X)] - \psi_{L}(1)\E[\theta(X, \nu) \check{S}(d-1|\nu, X)] \\
\pi_L g^w(d|L) &= \psi(d)  \left[\E[\theta(X, \nu) \check{S}(d-1|\nu, X)] - \psi_{L}(1)\E[\theta(X, \nu)^2 \check{S}(d-1|\nu, X)]\right]
\end{aligned}
\]
where $\theta(X, \nu) = \phi(X) \nu$.

Note that the above equations take a similar form to equations (\ref{eqn_surv}) and (\ref{eqn_dens}). Thus, by considering \( L = \ell, \ell' \) and replicating the steps following equation (\ref{eqn_dens}) in the proof of Theorem \ref{result_main_theorem}, we can derive that \(\psi_L(1) = \pi_L g^w(1|L)/\E[\theta(X, \nu)]\) for \(L \in \{\ell, \ell'\}\), and
\[
\psi(d) \E[\theta(X, \nu)] = \frac{g^w(d|\ell')g^w(1|\ell) - g^w(d|\ell)g^w(1|\ell')}{S^w(d-1|\ell')/\pi_{\ell} - S^w(d-1|\ell)/\pi_{\ell'}} \quad \text{for } d > 1
\]
To see that the moments are identified as well, first note that by plugging in $\psi_L(d, X) = \psi_L(d) \phi(X) $, we can rewrite equation (\ref{eqn_compact_moments}) as:

\begin{equation}\label{eqn_compact_mph}
g(d|\nu, L,X) = \psi_{L}(d) \sum_{k=1}^{\bar{D}} c_k(d,\bm{\psi_{L}}) \phi(X)^k \nu^k
\end{equation}
where $\bm{\psi_{L}} = \{\psi_{L}(d)\}_{d=1}^{\bar{D}}$ and $c_k(d,\bm{\psi_{L}}) $ is defined as in Lemma \ref{result_mom_id}.
Since we established that $\pi_{\ell} g^w(d|\ell) = \E[g(d|\nu, \ell, X)]$, we can write:
$$
\pi_L g^w(d|L) = \psi_{L}(d) \sum_{k=1}^{\bar{D}} c_k(d,\bm{\psi_{L}}) \E(\phi(X)^k \nu^k)  $$
The rest of the proof follows as in the proof for Lemma \ref{result_mom_id}. Denote $\bm{g^{w}_{L}} = \{\pi_L g^{w}(d|L)\}_{d=1}^{\bar{D}}$ and $\bm{\mu^{w}} = \{\E[\theta(X, \nu)^k]\}_{k=1}^{\bar{D}}$.  Then we can write $\bm{g^{w}_{L}} = C(\bm{\psi_{L}}) \bm{\mu^{w}}$ where $C(\bm{\psi_{L}})$ is an invertible matrix, as outlined in Lemma \ref{result_mom_id}. Hence, we can plug in the identified structural hazards to find $\bm{\mu^{w}}$. 
\end{proof}

\subsection{Proof of Proposition \ref{result_prop_main}}\label{proof_prop_main}
\begin{proof}
First note that, by the definition of $\tilde{D}$ and $C$, we can write $\tilde{h}^w(d| L)$ as follows:
$$ \tilde{h}^w(d| L) = \E\left[\frac{\Pr(D=d, D^c \geq d|\nu, L, X)}{p_{L}(X)} \middle | L\right] \ \big/ \ \E\left[\frac{\Pr(D \geq d, D^c \geq d |\nu, L, X)}{p_{L}(X)} \middle | L \right] $$
Now note that, 
\begin{align*}
\E\left[\frac{\Pr(D \in \Delta, D^c \geq d|\nu, L, X)}{p_L(X)} \middle | L \right] &=  \E\left[ \E\left[\frac{\Pr(D \in \Delta, D^c \geq d|\nu, L, X)}{p_L(X)}\middle | L, X \right] \middle | L \right] & (i) \\
& = \E\left[ \E\left[\frac{\Pr(D \in \Delta, D^c \geq d|\nu, L, X)}{p_L(X)}\middle | X \right] \middle | L \right] & (ii) \\
& = \int \frac{\E[\Pr(D_L \in \Delta, D^c \geq d|\nu, x) | x ]}{p_L(x)} \cdot f_{X|L}(x) & (iii) \\
& = \E[\Pr(D_L \in \Delta, D^c \geq d|\nu, X)] & (iv) \\
& = \Pr(D_L \in \Delta, D^c \geq d) 
\end{align*}
Here, \(\pi_{\ell} = \Pr(L=\ell)\) and \(f_{X|L}(.)\) represents the distribution of \(X\) given \(L\). Step (i) utilizes the law of iterated expectations, step (ii) follows from \(L \perp \nu | X\) and \(D^c \perp \nu |X\), step (iii) uses the definition of conditional expectation and swaps $D$ with potential duration $D_L$, and step (iv) follows from the Bayes' rule, which implies that \(f_{X|L}(x) = p_{L}(x) f_X(x)/\pi_{L}\).

Note that this implies that:
$$ \tilde{h}^w(d| L) = \frac{\Pr(D_L = d, D^c \geq d) }{\Pr(D_L \geq d, D^c \geq d)} = \frac{\Pr(D_L = d| D^c \geq d) }{\Pr(D_L \geq d| D^c \geq d)} = \Pr(D_L = d| D_L \geq d, D^c \geq d) $$

As before, if we are given the hazard rate, we can also calculate the corresponding density and survival rate. In what follows, the identification of parameters is shown using the density and survival rates \(\Pr(D_L=d|D^c \geq d)\) and \(\Pr(D_L \geq d|D^c \geq d)\). In particular, note that \(\Pr(D_L = d|D^c \geq d) = \E[g(d|\nu, L, X)|D^c \geq d]\) and \(\Pr(D_L \geq d|D^c \geq d) = \E[S(d|\nu, L, X)|D^c \geq d]\). Therefore, taking the conditional expectation of the expressions implied by equations (\ref{eqn_isurv}) and (\ref{eqn_idens}) and substituting \(\psi_L(d, X) = \psi_L(d) \phi(X)\) with \(\psi_L(d) = \psi(d)\) for \(d>1\), we get:
\[
\begin{aligned}
 & \Pr(D_L \geq  d  |D^c \geq d) = \E_d[\check{S}(d-1|\nu, X)] - \psi_{L}(1)\E_d[\theta(X, \nu) \check{S}(d-1|\nu, X)] \\
 & \Pr(D_L = d  |D^c \geq d) = \psi(d) \left[\E_d[\theta(X, \nu) \check{S}(d-1|\nu, X)] - \psi_{L}(1)\E_d[\theta(X, \nu)^2 \check{S}(d-1|\nu, X)]\right]
\end{aligned}
\]

Here, $\theta(X, \nu) = \phi(X) \nu$, and I use $\E_d$ to denote expectations conditional on $D^c \geq d$ for brevity; specifically, \(\E_d(.) = \E(.|D^c \geq d)\). Except for the expectations being conditional on \(D^c \geq d\), the expressions on the right-hand side in the equations are identical to those in Proposition \ref{result_prop}. Therefore, one can follow the same steps as in Proposition \ref{result_prop} to show the identification of structural hazards. However, it should be noted that the moments are not yet identified. This is because taking the expectation of equation (\ref{eqn_compact_mph}) conditional on \(D^c \geq d\) results in different moments entering different period densities. For instance, we will have $\E[\theta(X, \nu)|D^c \geq 1]$ in the first-period density, but the second-period density will have $\E[\theta(X, \nu)|D^c \geq 2]$. However, if we further assume that $D^c \perp X$ and $D^c \perp \nu$, then $\E[\theta(X, \nu)^k|D^c \geq d]=\E[\theta(X, \nu)^k]$, and identification of moments follows from Proposition \ref{result_prop}.
\end{proof}

\subsection{Proof for $\E[m_i(\ell,d; \Theta)]=0$}\label{proof_moms_zero}
For brevity, I omit the $i$ subscript from the terms in the following exposition. Additionally, let us denote the indicators for exit and survival used in the individual moments as follows: \(I^E(d) = \I\{\tilde{D}=d, C=0\}\) and \(I^S(d) = \I\{\tilde{D} \geq d\}\). In this case, we can rewrite the expression for individual moments as follows:
\[
  m(\ell,d;\Theta) = \I\{L=\ell\} \left[  \frac{I^E(d) }{p_\ell(X)}-\tilde{h}^w(d|\ell;\Theta) \cdot \frac{I^S(d)}{p_\ell(X)} \right]
\]

\begin{proof} By taking the expectation of the \(m(\ell,d;\Theta)\) expression above and applying the law of iterated expectations, we get:
\[
 \E[ m(\ell,d;\Theta) ]= \pi_\ell \left[ \E\left[ \frac{I^E(d) }{p_\ell(X)} \middle | L = \ell \right]-\tilde{h}^w(d|\ell;\Theta) \cdot \E\left[ \frac{I^S(d)}{p_\ell(X)} \middle | L = \ell \right] \right]
\]
Applying the law of iterated expectations again, we can write: 
\[
 \frac{\E[ m(\ell,d;\Theta) ]}{\pi_\ell} = \E \left[\frac{\E \left[I^E(d)\middle |L = \ell, X \right]}{p_\ell(X)}   \middle |L = \ell \right]-\tilde{h}^w(d|\ell;\Theta) \cdot \E \left[\frac{\E \left[I^S(d)\middle |L = \ell, X \right]}{p_\ell(X)}   \middle |L = \ell \right] 
\]
Finally, by substituting \(\E[I^E(d)|L, X] = \Pr(\tilde{D}=d, C=0|\nu, L, X)\) and \(\E[I^S(d)|L, X] = \Pr(\tilde{D} \geq d|\nu, L, X)\) into the above expression, and applying the definition of \(\tilde{h}^w(d|\ell;\Theta)\), it becomes clear that the right-hand side equals zero.
\end{proof}

\clearpage
\setcounter{table}{0}
\setcounter{figure}{0}
\section{Search Model Calibration}\label{app_search_model}

I calibrate the model under standard values for model parameters. To maintain consistency with the econometric model, each period is assumed to be 12 weeks long. Corresponding to a 5 percent annual interest rate, the discount factor $\beta$ is set equal to 0.985. I normalize the wage to 1 and set the replacement rate for unemployment benefits at 0.5. In addition, I assume individuals receive an annuity payment of 0.1 times their wages in each period, regardless of their employment status. This can be interpreted as the income of a secondary earner. Utility from consumption is given by the constant relative risk aversion (CRRA) utility function, $u(c) = c^{1-\sigma}/(1-\sigma)$ with $\sigma=1.75$. I follow \cite{DellaVignaEtAl2017} and \cite{MarinescuSkandalis2021} in assuming that costs of job search are given by $ c(s) = \theta s^{1+\rho}/(1+\rho)$. I set $\rho=1$ and $\theta=50$.\footnote{Different parameters for the cost function do not change qualitative predictions of my exercise but do lead to changes in the scale of the search effort.} Table \ref{tab_calib_details} summarizes the calibration parameters and Figure \ref{fig_calib_fit} displays the fit of the calibrated model. 

\begin{table}[h]
\begin{threeparttable}
\caption{Calibration Parameters for the Search Model}\label{tab_calib_details}
\begin{tabularx}{\textwidth}{p{0.65\textwidth}c}
\toprule
Parameter & Value \\
\midrule
Length of each period & 12 Weeks \\
Discount factor $\beta$ & 0.985 \\
Relative risk aversion $\sigma$ & 1.75 \\
Per period wages $w$ & 1 \\
Annuity Payments & 0.1 \\
Unemployment benefits & 0.5 \\
Benefit exhaustion $D_B$ & 3 \\
Search cost parameter $\rho$ & 1 \\
Search cost parameter $\theta$ & 50 \\
First period arrival rate $\delta(1)$ & 1 \\
\bottomrule
\end{tabularx}
\begin{tablenotes}
\item \textit{Note:} The table presents the parameters used for calibrating the search model in Section \ref{sec_search_model}.
\end{tablenotes}
\end{threeparttable}
\end{table}

\begin{figure}[t]\caption{Search Model Calibration: Fit}\label{fig_calib_fit}
\centering
\begin{subfigure}{.49\linewidth}
\raggedleft
\includegraphics{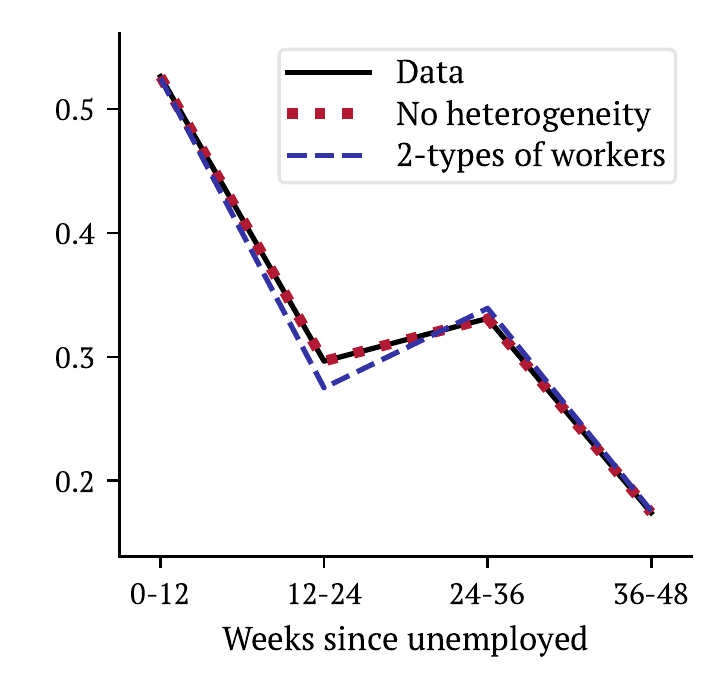}
\subcaption{Observed Hazard}
\end{subfigure} \hfill
\begin{subfigure}{.49\linewidth}
\includegraphics{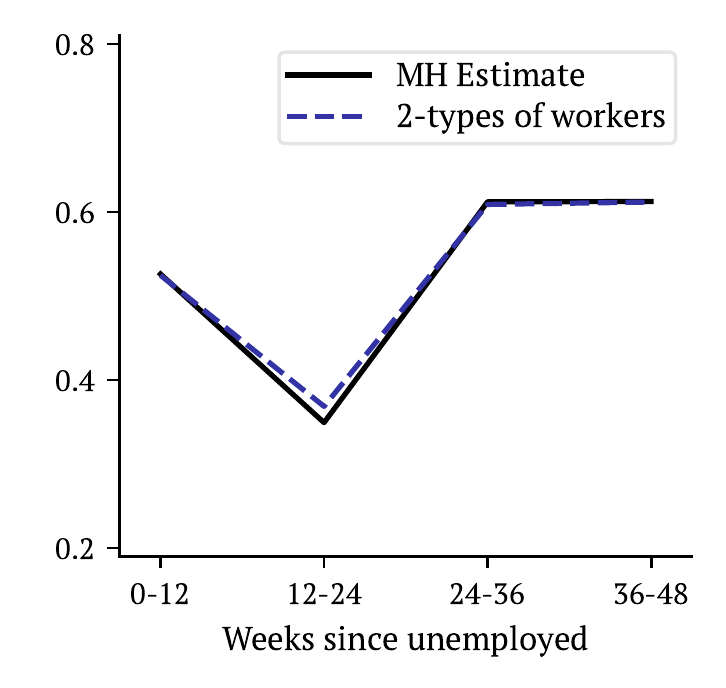}
\subcaption{Structural Hazard}
\end{subfigure} 
\vspace{-0.75em}
\floatfoot{\textit{Note:} The figure displays the fit of the search model for the two calibration exercises described in the text. Panel A shows the observed exit rate in the data (solid black line) alongside the corresponding fitted values obtained from calibrating the search model without heterogeneity (dotted red line) and with two types of workers (dashed blue line). Panel B displays the estimated structural hazard from the Mixed Hazard (MH) model (solid black line) and the fitted structural hazard from calibrating the search model with two types of workers (dashed blue line).}
\end{figure}


\clearpage
\setcounter{table}{0}
\setcounter{figure}{0}
\section{Data}\label{app_data}
This section provides more details about data construction and sample selection. It also elaborates on the estimation of propensity scores and checks the balance of the weighted sample on additional measures. Supplementary data descriptives are also presented here.

\subsection{Data Construction and Sample Selection}\label{app_data_cnstr}

The Displaced Worker Supplement (DWS) was introduced in 1984, but the variable on the length of notice was not included in the first two samples. Furthermore, the definition of displaced workers has changed over time.\footnote{The recall window was 5 years instead of 3 before 1994.} Before 1998, self-employed individuals or those who expected to be recalled to their lost job within six months were also included in the survey. However, the information on whether a worker expected to be recalled is only available for the years 1994 and 1996. In addition, the data on the length of time individuals took to find their next job is miscoded and largely missing for the year 1994. For these reasons, my analysis begins from 1996. Moreover, to maintain consistency across years, I exclude self-employed individuals or those who expected to be recalled from the 1996 sample.

\begin{table}[p]
\begin{threeparttable}
\caption{Sample Selection}\label{tab_sample}
\begin{tabularx}{\textwidth}{p{9cm}Y}
\toprule
Sample condition & Observations  \\ \midrule
 DWS 1996-2020, 21-64 year old respondents  & 44707 \\
 No recall expectation                      & 44537 \\
 Lost job was not self-employment           & 43471 \\
 Non-missing values for variables used      & 29443 \\
 Worked full-time at lost job               & 25293 \\
 Employed for at least 6 months at lost job & 22991 \\
 Had health insurance at lost job           & 14825 \\
 Held less than 3 jobs since lost job       & 13784 \\
 Got a notice before job loss               &  5898 \\
 Got a notice of \ensuremath{>}1 month                   &  4175 \\

\bottomrule
\end{tabularx}
\begin{tablenotes}
\item  \textit{Note:} The table shows the number of observations remaining at each step of sample selection. 
\end{tablenotes}
\end{threeparttable}
\end{table}

\begin{table}[p]
\begin{threeparttable}
\centering
\caption{Comparison of the analytical sample to all individuals in the Displaced Worker Supplement (DWS) and the Current Population Survey (CPS)}\label{tab_cps_comparison}
\begin{tabularx}{\textwidth}{p{5cm}YYY}
\toprule
& Sample & DWS & CPS \\
& (1) & (2) & (3) \\
\midrule
 Age            & 43.04 & 40.61 & 42.22  \\
 Female         & 0.45  & 0.44  & 0.52   \\
 Married        & 0.61  & 0.54  & 0.60   \\
 Black          & 0.10  & 0.11  & 0.10   \\
 High School    & 0.30  & 0.42  & 0.41   \\
 Some College   & 0.30  & 0.32  & 0.29   \\
 College Degree & 0.40  & 0.26  & 0.30   \\
 Employed       & 0.80  & 0.67  & 0.75   \\
 Unemployed     & 0.18  & 0.21  & 0.04   \\
 NILF           & 0.02  & 0.12  & 0.21   \\
\midrule
 Observations   & 4175  & 44707 & 964225 \\
 
\bottomrule
\end{tabularx}
\begin{tablenotes}
\item  \textit{Note:}  All samples are restricted to individuals between the ages of \agecutoff and pertain to years 1996-2020. Column (1) includes individuals from the DWS who worked full-time for at least six months and were provided health insurance at their lost job, did not expect to be recalled, and received a layoff notice of 1-2 months or greater than 2 months. Columns (2) and (3) include all individuals in the DWS and the monthly CPS, respectively.
\end{tablenotes}
\end{threeparttable}
\end{table}

The final sample only includes individuals with non-missing information on all the variables utilized in the analyses. Additionally, to capture individuals who had lost a job representing stable, non-temporary employment, the sample is limited to those who worked full-time for at least six months at their previous job and had health insurance provided through that employment. Finally, to minimize retrospective bias, I exclude individuals who report switching more than two jobs since losing their previous job. The sample selection procedure is summarized in Table \ref{tab_sample}.

The duration of unemployment for individuals who have secured a job by the time of the survey is given by the \textit{dwwksun} variable, which measures the number of weeks the person was unemployed between leaving or losing one job and starting another. For those who report not holding another job since their lost job, censored duration is obtained using the \textit{durunemp} variable from the CPS. Since 2012, tenure at the lost job was top-coded at 24 years. To maintain consistency across samples, I also implement a top code of 24 years for all years prior to 2012. Earnings are reported in 1999 dollars.  Table \ref{tab_cps_comparison} presents the descriptive statistics for my analytical sample relative to all individuals aged \agecutoff in the DWS and the CPS over the sample period. 


\subsection{Propensity Score Weighting}\label{app_psw}

\begin{figure}[t]\caption{Assessing Overlap of Propensity Score Distributions}\label{fig_ps_bal}
\includegraphics{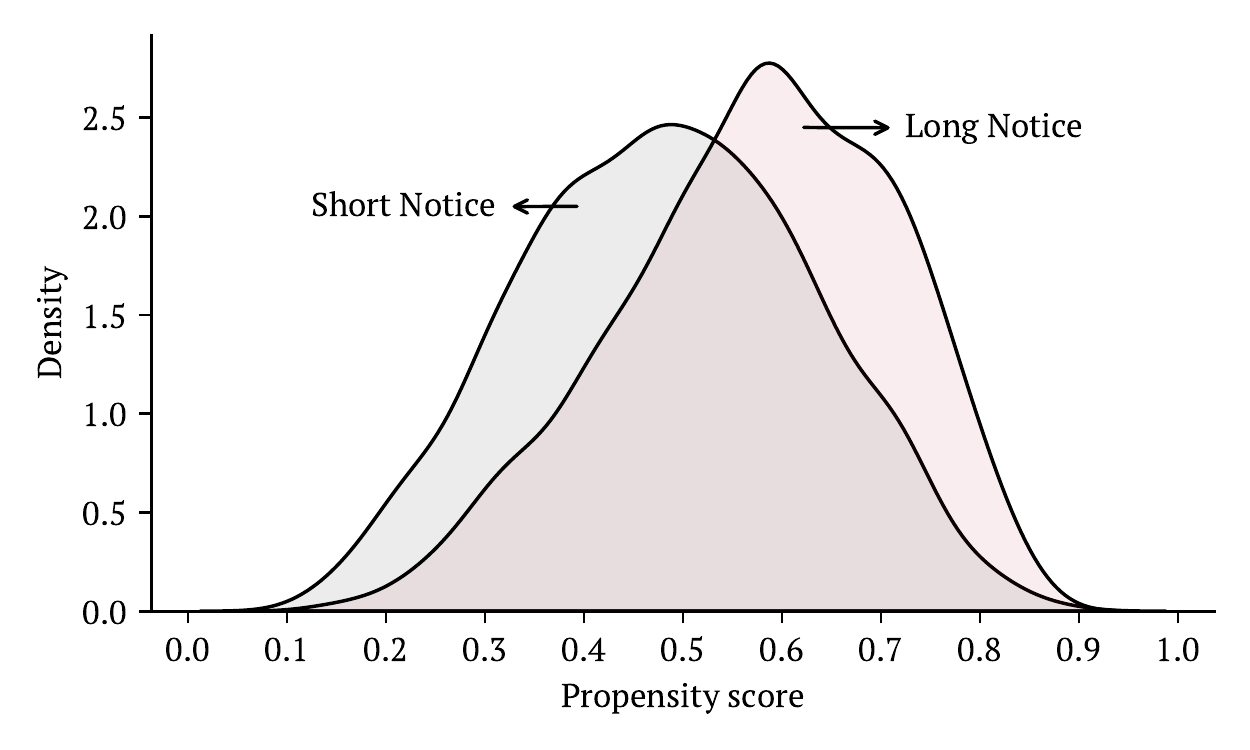}
\floatfoot{\textit{Note:} The figure presents the density of estimated propensity scores for individuals with short and long notice separately.}
\end{figure}

\begin{figure}[t]\caption{Length of Notice over Time}\label{fig_dyear_bal}
\centering
\includegraphics{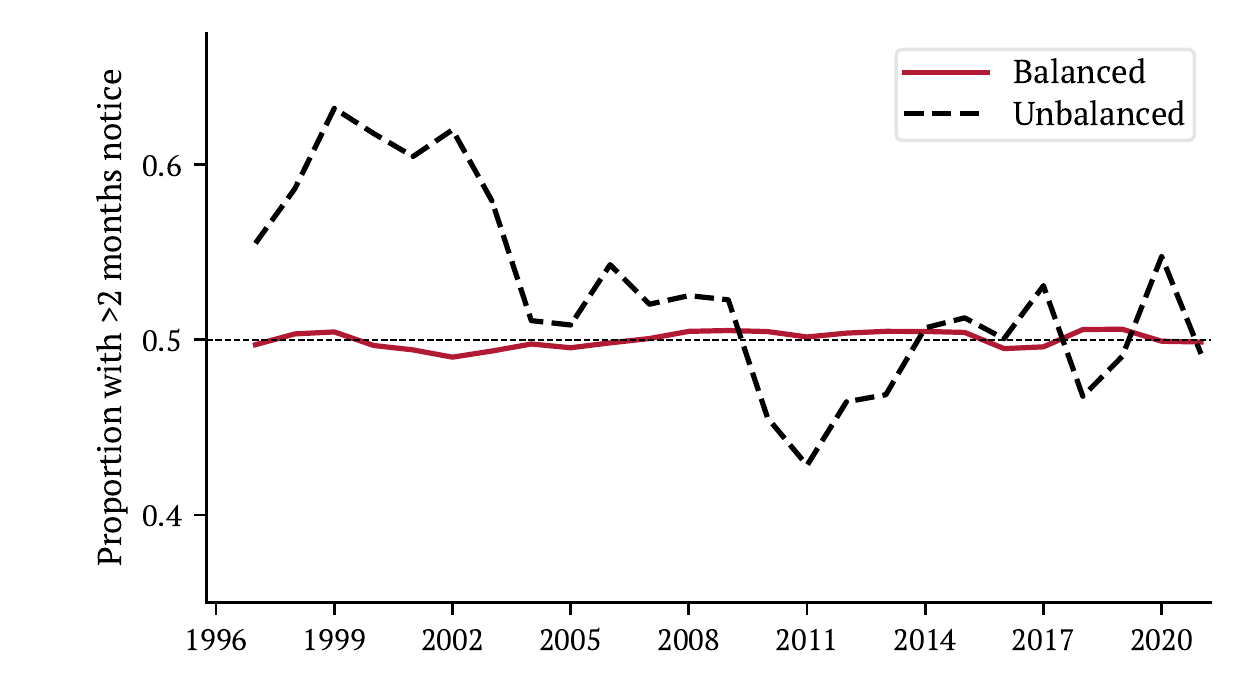}
\floatfoot{\textit{Note:} The figure plots a 3-year moving average of the proportion of individuals who received a notice of more than
2 months amongst all individuals in the sample who were displaced in a given year.}
\end{figure}

\begin{figure}[p]\caption{Industry and Occupation of the Lost Job}\label{fig_occ_ind_bal}
\begin{subfigure}{0.51\textwidth}
\includegraphics[trim = 0.5in 0in 0in 0in]{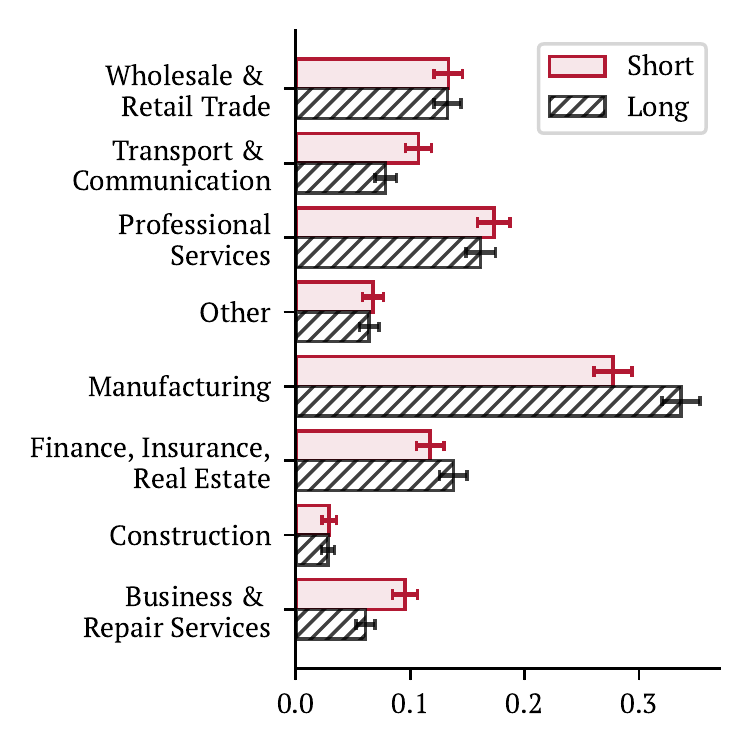}
\caption{Industry, Unbalanced}
\end{subfigure} 
\begin{subfigure}{0.475\textwidth}
\includegraphics[trim = 0.5in 0in 0in 0in]{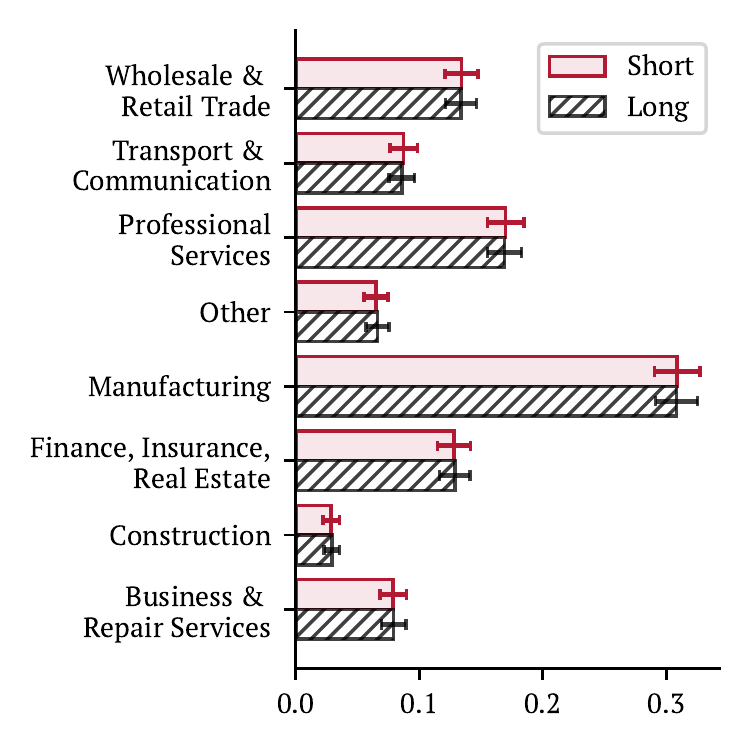}
\caption{Industry, Balanced}
\end{subfigure}
\begin{subfigure}{0.51\textwidth}
\includegraphics[trim = 0.5in 0in 0in 0in]{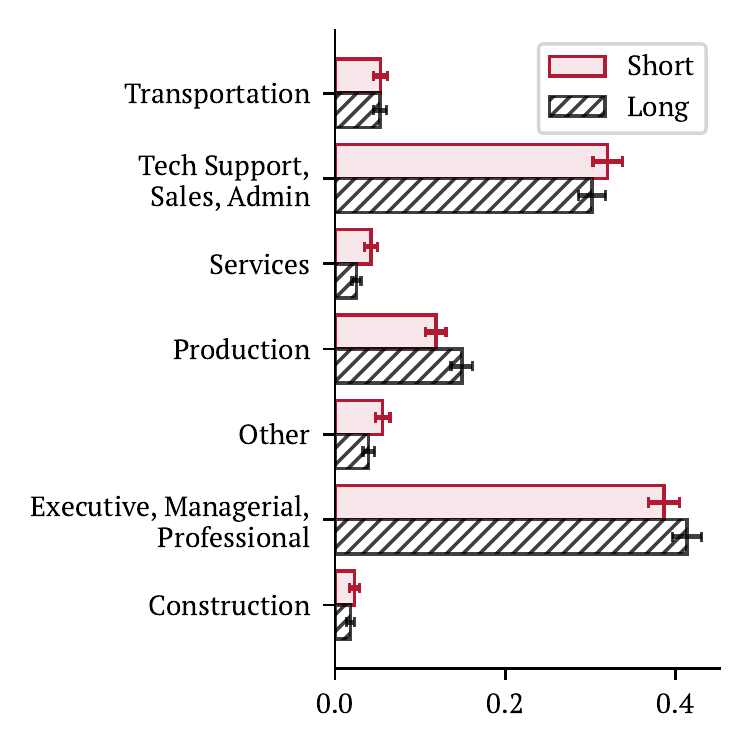}
\caption{Occupation, Unbalanced}
\end{subfigure} 
\begin{subfigure}{0.475\textwidth}
\includegraphics[trim = 0.5in 0in 0in 0in]{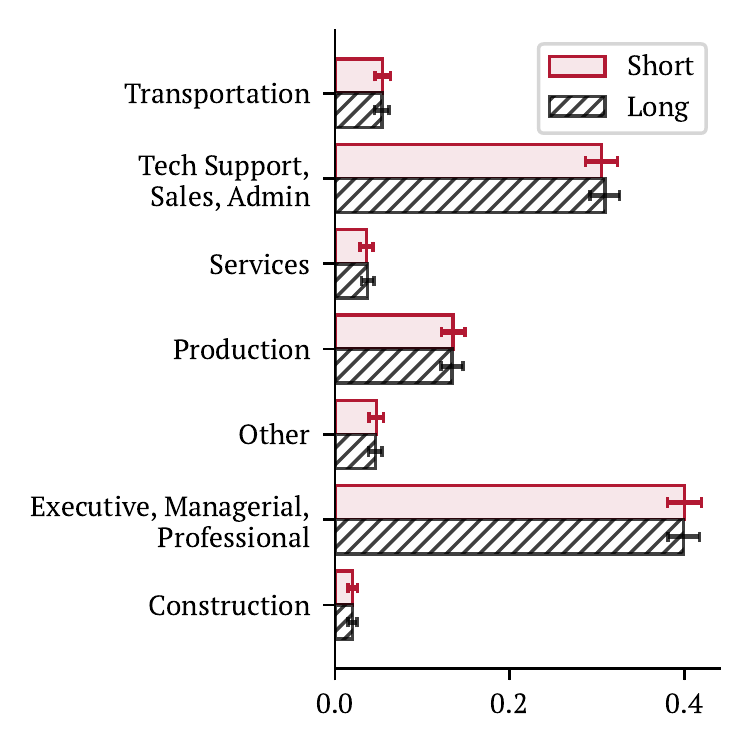}
\caption{Occupation, Balanced}
\end{subfigure}
\floatfoot{\textit{Note:} The figure presents the proportions of individuals whose displaced jobs were in specific industries (panels A and B) and occupations (panels C and D) among long-notice and short-notice workers in both the unbalanced and balanced samples. The error bars represent the 90\% confidence intervals.}
\end{figure}

To ensure individuals with long and short notice are comparable, I reweight the sample using inverse propensity score weighting. The weight for each individual is calculated as the inverse of the likelihood of receiving the reported notice length. To estimate the propensity scores, I utilize a logistic regression where the odds of receiving a longer notice are modeled as a function of several variables. These variables include age, gender, marital status, race (specifically, an indicator for Black), education level, reason for layoff, union membership, residence in a metropolitan area, tenure, and earnings at the lost job. Additionally, fixed effects for occupation and industry of the lost job, state of residence, and year of displacement are included in the model. The density of estimated propensity scores for short and long-notice individuals is displayed in Figure \ref{fig_ps_bal}. The figure indicates a significant overlap between the two distributions, with all estimated propensity scores falling in the 0.1 to 0.9 range, thereby rendering further data trimming unnecessary.

Table 1 in the main text provides evidence that the reweighting achieves balance across certain observable variables. Figure \ref{fig_dyear_bal} demonstrates that reweighting leads to balance with respect to the year of displacement. In addition, Figure \ref{fig_occ_ind_bal} presents occupation and industry distributions for short and long-notice workers in both the balanced and unbalanced samples. Notably, the weighted sample exhibits more similarity in the industrial and occupational composition of short and long-notice workers.

\subsection{Additional Descriptives}

This section provides additional descriptive statistics. Table \ref{tab_reg_wages} presents the relationship between longer notice and earnings at the subsequent job. The table indicates that workers with longer notice tend to have higher earnings in their subsequent jobs. However, we cannot interpret this as a direct impact of longer notice because extended periods of unemployment can have a negative impact on wages \citep{SchmiederEtAl2016}, and as shown in this paper, a longer notice leads to shorter unemployment spells. Table \ref{tab_uiben_recd} describes the incidence of UI take-up in the sample, and Figure \ref{fig_uiex_break} describes the timing of benefit exhaustion amongst UI takers. Figure \ref{fig_altbins} presents the data with unemployment duration binned in 4 and 9-week intervals. 

\begin{table}[t]
\begin{threeparttable}
\caption{Earnings at the Subsequent Job}\label{tab_reg_wages}
\begin{tabularx}{\textwidth}{p{0.225\textwidth}YYYY}
\toprule
& \multicolumn{4}{c}{Weekly Log Earnings} \\
& (1) & (2) & (3) & (4) \\
\midrule
> 2 month notice & 0.091** & 0.066** & 0.093** & 0.063*\\
  & (0.037) & (0.032) & (0.038) & (0.032)\\ \addlinespace[2ex]
Controls   &  No & Yes  & No & Yes \\
Weights   & No  & No   & Yes & Yes \\
\midrule
Observations & 2657 & 2657 & 2657 & 2657\\
\bottomrule
\end{tabularx}
\begin{tablenotes}
\item \textit{Note:} The table shows results from linear regressions of log weekly wages at the subsequent job on an indicator for receiving a notice of more than 2 months. The sample used is similar to the main analytical sample, but it excludes individuals who had not yet found employment at the time of the survey, had multiple jobs between their previous and current jobs, or had incomplete earnings information for other reasons. Robust standard errors are reported in the parenthesis.
\end{tablenotes}
\end{threeparttable}
\end{table}

\begin{table}[p]\caption{Unemployment Insurance Take-up}\label{tab_uiben_recd}
\begin{threeparttable}
\begin{tabularx}{\textwidth}{XYY}
\toprule
Duration & Observations & Received UI Benefits \\
\midrule
 0 weeks    & 820  & 0.06 \\
 0-4 weeks  & 959  & 0.30 \\
 4-8 weeks  & 457  & 0.61 \\
 8-12 weeks & 363  & 0.70 \\
 \ensuremath{>} 12 weeks & 1562 & 0.82 \\
 
\bottomrule
\end{tabularx}
\begin{tablenotes}
\item \textit{Notes}: This table reports the percentage of individuals in the baseline sample who reported receiving UI benefits by the duration of unemployment.
\end{tablenotes}
\end{threeparttable}
\end{table}

\begin{figure}[p]\caption{Timing of Benefit Exhaustion}\label{fig_uiex_break}
\includegraphics{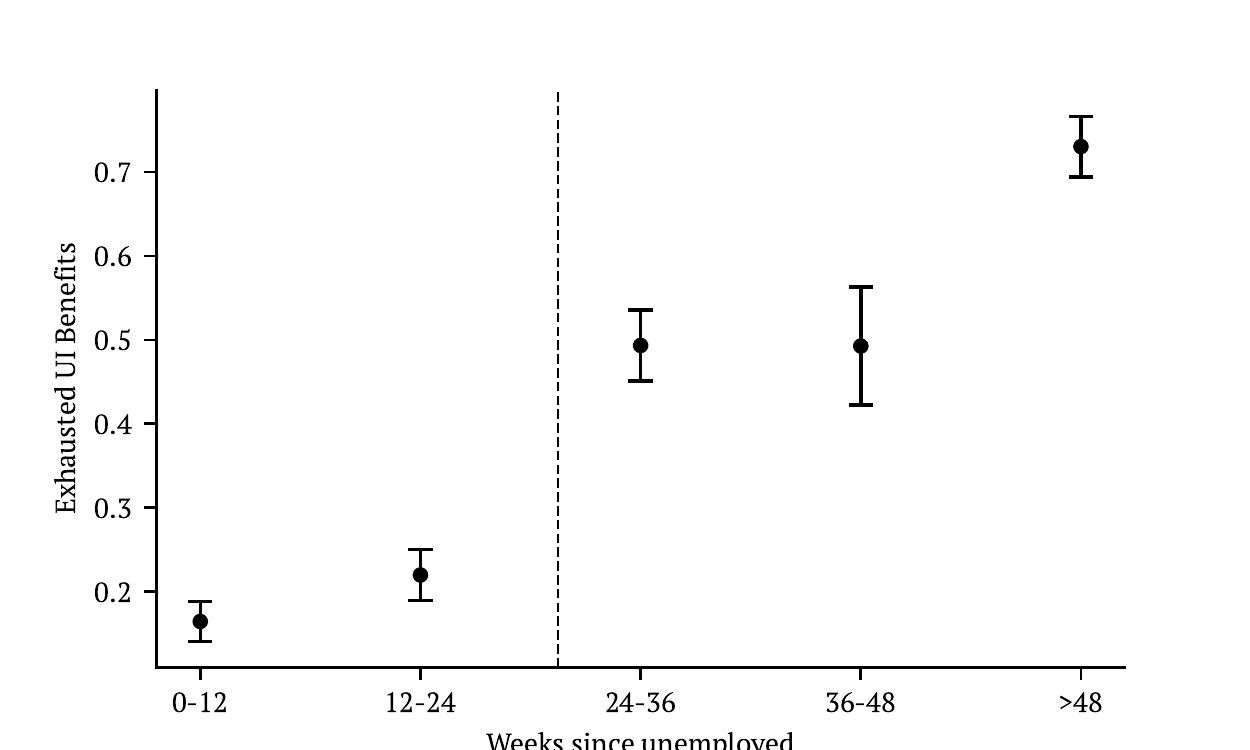}
\floatfoot{\textit{Note:} The figure presents the proportion of individuals who report having exhausted their UI benefits by the duration of unemployment. The sample is restricted to individuals in the main analytical sample who reported receiving UI benefits, and duration is binned in 12-week intervals.} 
\end{figure}


\begin{figure}[p]\caption{Survival and Exit Rates with Alternative Bins}\label{fig_altbins}
\begin{subfigure}{.475\textwidth}
\centering
\includegraphics{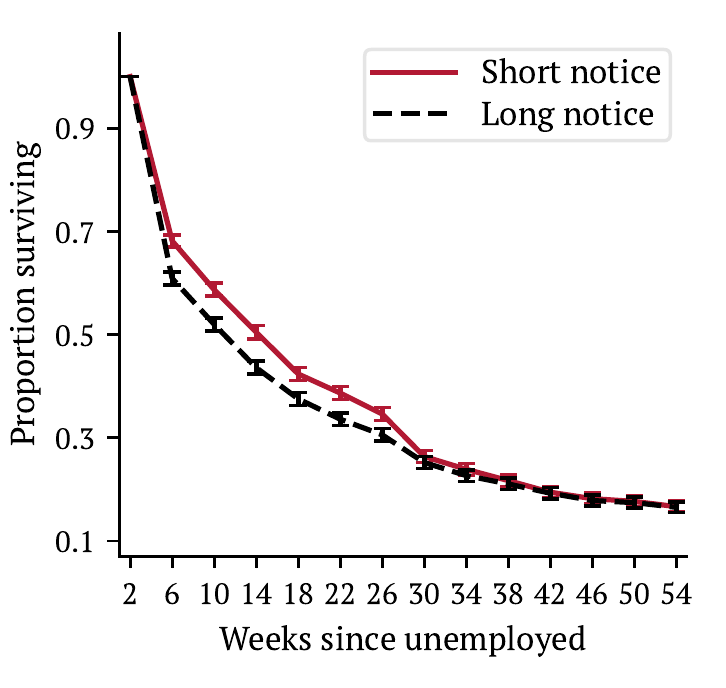}
\subcaption{Survival Rate}
\end{subfigure}
\begin{subfigure}{.475\textwidth}
\centering
\includegraphics{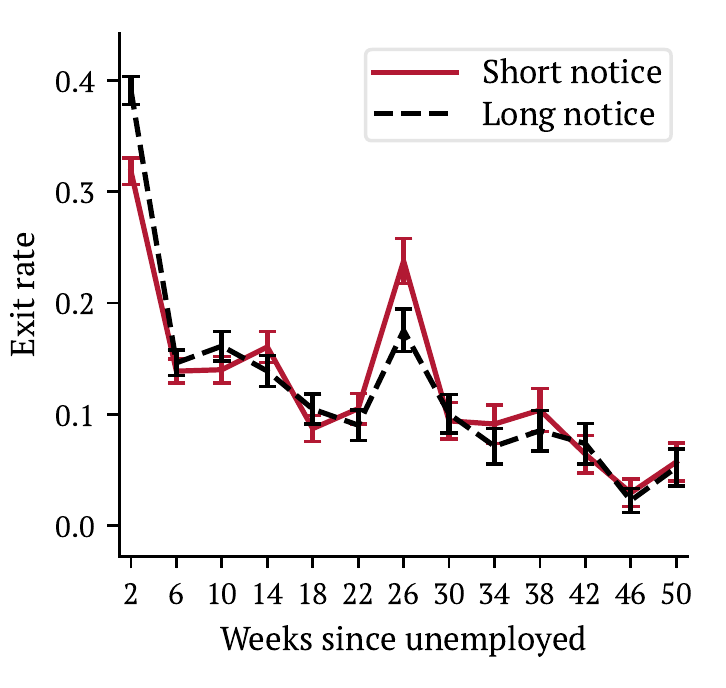}
\subcaption{Exit Rate}
\end{subfigure}
\begin{subfigure}{.475\textwidth}
\includegraphics{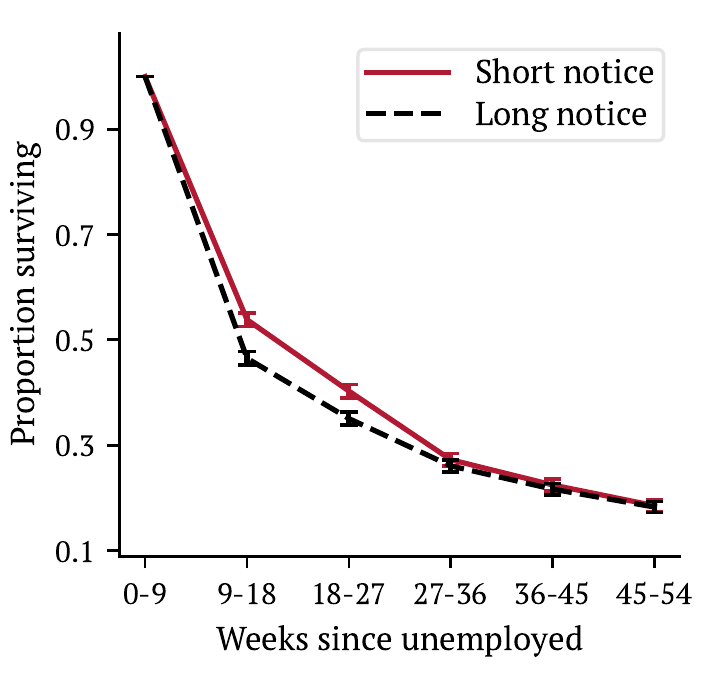}
\subcaption{Survival Rate}
\end{subfigure}
\begin{subfigure}{.475\textwidth}
\includegraphics{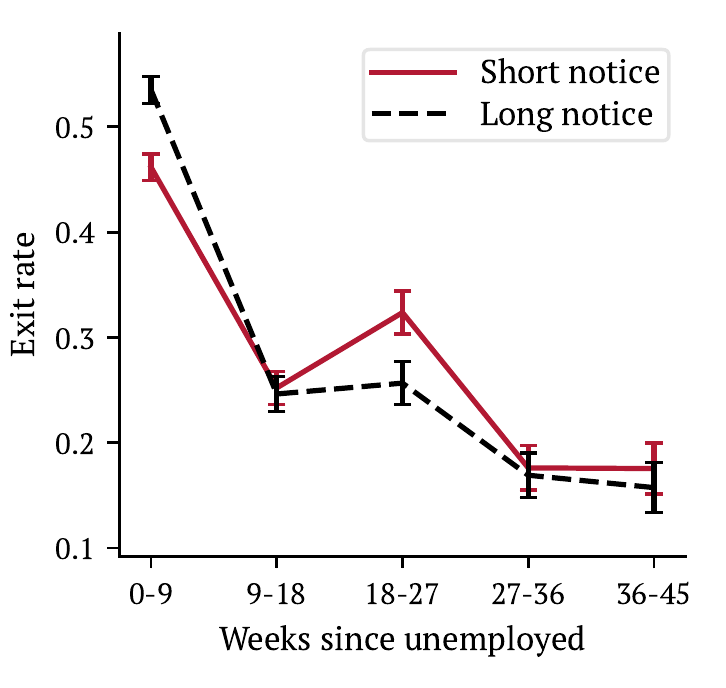}
\subcaption{Exit Rate}
\end{subfigure}
\floatfoot{\textit{Note:} Unemployment duration is binned in 4-week intervals for panels A and B, while it is binned in 9-week intervals for panels C and D. Panel A and C present the proportion of individuals who are unemployed at the beginning of each interval. Panel B and D present the proportion of individuals exiting unemployment in each interval amongst those who were still unemployed at the beginning of the interval. Error bars represent 90\% confidence intervals.}
\end{figure}



\clearpage
\setcounter{table}{0}
\setcounter{figure}{0}
\section{Robustness Checks}\label{app_robust}

In this subsection, I present a series of robustness checks. 

\subsection{Non-Parametric Estimates}

As shown in the paper, the specified MH model is non-parametrically identified. Figure \ref{fig_robust_ff} presents the non-parametric estimate for the structural hazard alongside the baseline estimate, which assumes a log-logistic functional form for the hazard. This figure shows that the non-parametric point estimates are roughly equivalent to the baseline estimates. However, the standard errors are larger, specifically the standard error associated with the last data point, which blows up dramatically.

\begin{figure}[p]\caption{Non-Parametric Estimates}\label{fig_robust_ff}
\centering 
\includegraphics{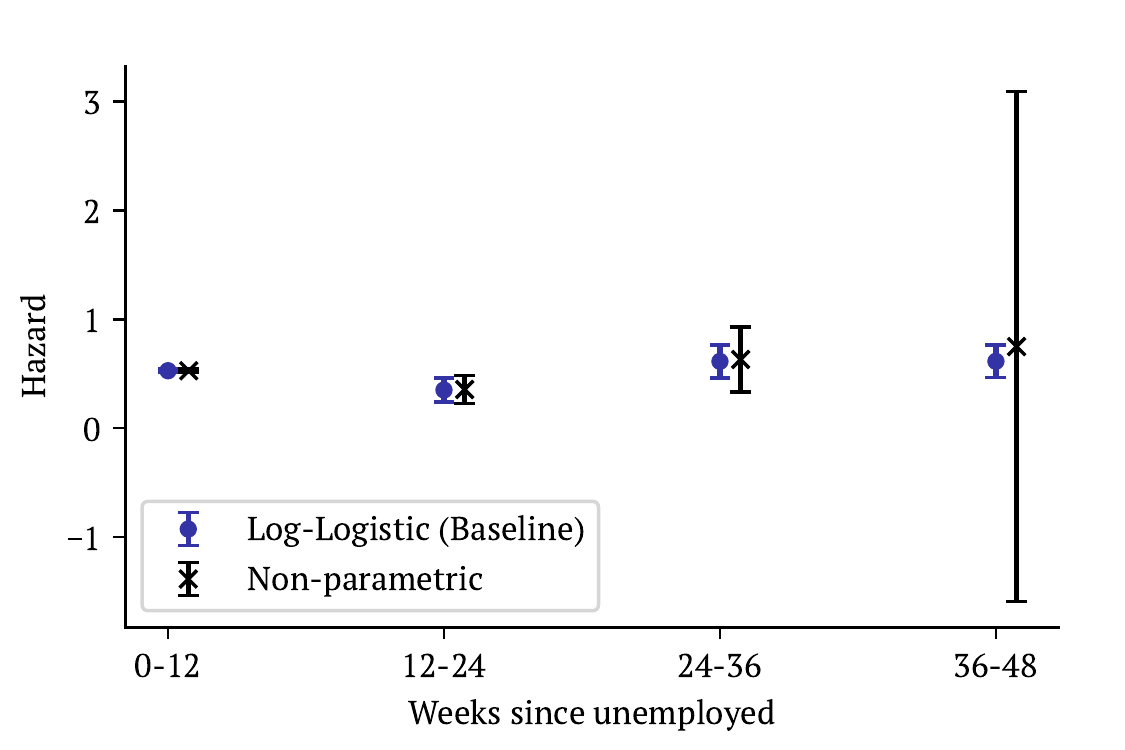}
\floatfoot{\textit{Note:} The figure compares non-parametric estimates of the structural hazard from the Mixed Hazard model with baseline estimates that assume a log-logistic functional form for the structural hazard. Error bars represent 90\% confidence intervals.}
\end{figure}

\subsection{Unweighted Sample}

Figure \ref{fig_robust_unwtd} displays the data and estimates using the unweighted sample. Panel A shows that the pattern of exit rates and how they vary by notice length is similar to that observed in the weighted sample in Figure \ref{fig_dur_dist}. Similarly, the estimates based on the unweighted sample, presented in panel B, closely resemble the baseline estimates shown in Figure \ref{fig_baseline_estimates}. This suggests that observable characteristics have a limited role in explaining the raw differences in exit rates between the two notice groups.

\begin{figure}[p]\caption{Data and Estimates using the Unweighted Sample}\label{fig_robust_unwtd}
\begin{subfigure}{.475\textwidth}
\centering
\includegraphics{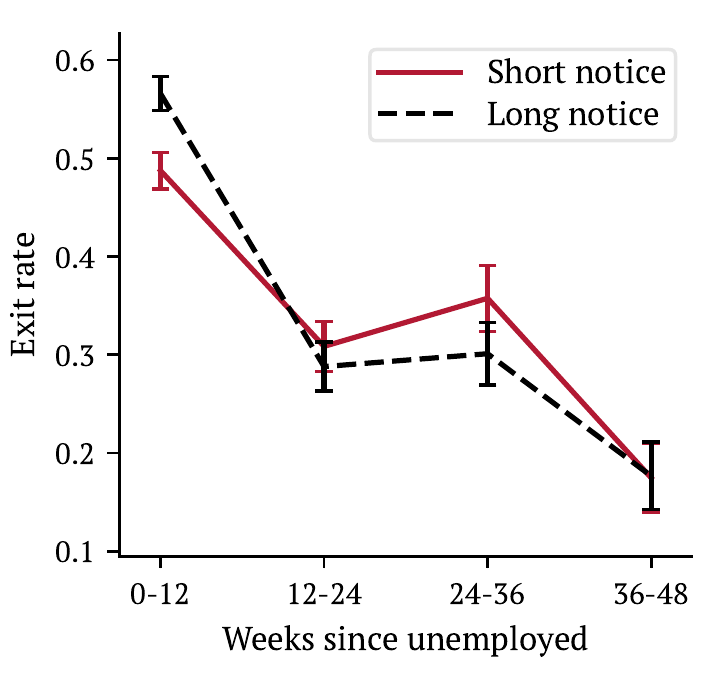}
\subcaption{Exit Rate}
\end{subfigure}
\begin{subfigure}{.475\textwidth}
\includegraphics{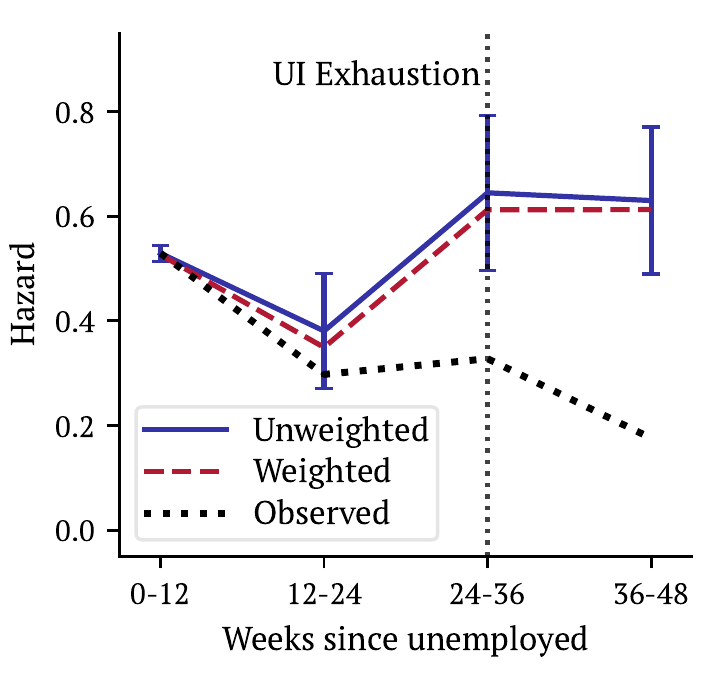}
\subcaption{Estimated Structural Hazard}
\end{subfigure}
\floatfoot{\textit{Note:} The figure presents data and estimates for the unweighted analytical sample. Panel A displays the exit rate from the data separately for long and short-notice workers. In panel B, the solid blue line shows the estimated structural hazard from the Mixed Hazard model using unweighted moments, the dashed red line shows the estimate using weighted moments, and the dotted black line represents the average exit rate for short and long-notice workers in the data.}
\end{figure}

\subsection{Alternative Notice Categories}\label{subsec_alt_notcats}

As noted in Section \ref{sec_data}, the question on notice length in the DWS is categorical and includes categories of no notice, <1 month, 1-2 months, and >2 months. So far, I have only focused on the latter two categories, referred to as short and long notice. I now test the robustness of my findings with respect to this restriction. First, I include individuals who received no notice as a third group in the original estimation sample. Then, I repeat the process, adding individuals with <1 month's notice as the third group.\footnote{I do not include all four categories simultaneously for two reasons. The first, and less significant, reason is that the overlap in propensity scores is already suboptimal with three categories and worsens with four. The primary reason is that the exit rate in the first period for the \textit{<1 month} group is lower than for the \textit{no notice} group, possibly due to selection, since separations with no notice may also include voluntary quits. However, the first-period exit rates for both are lower than those for the 1-2 month and >2 month notice groups.} I use the same covariates as before to fit the propensity score models using multinomial logit and weight the data accordingly. Observations where the propensity scores fall outside the 0.1 to 0.9 range are excluded. I then re-estimate the Mixed Hazard model with three notice categories for both the augmented samples.

\begin{figure}[t]\caption{Estimates with Alternative Notice Categories}\label{fig_robust_notice_cat}
\begin{subfigure}{.475\textwidth}
\centering
\includegraphics{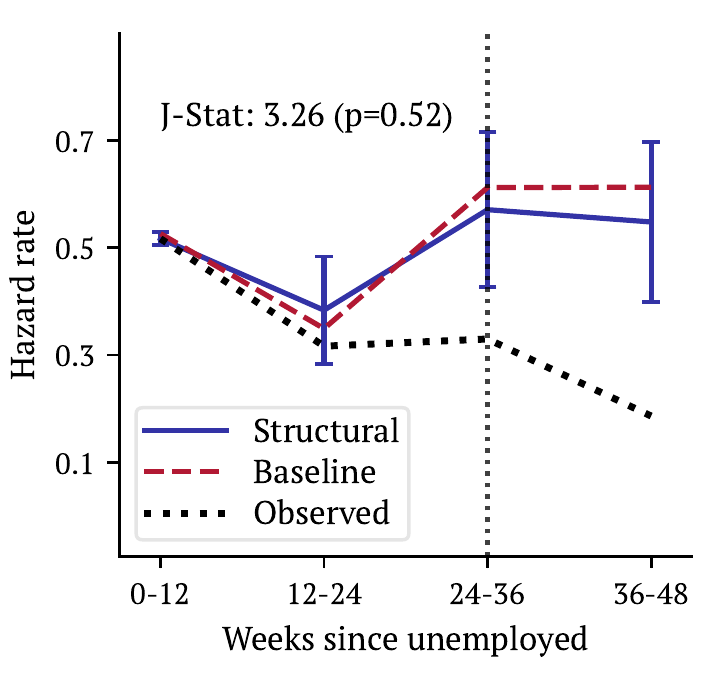}
\subcaption{No notice, 1-2 months, >2 months}
\end{subfigure}
\begin{subfigure}{.475\textwidth}
\includegraphics{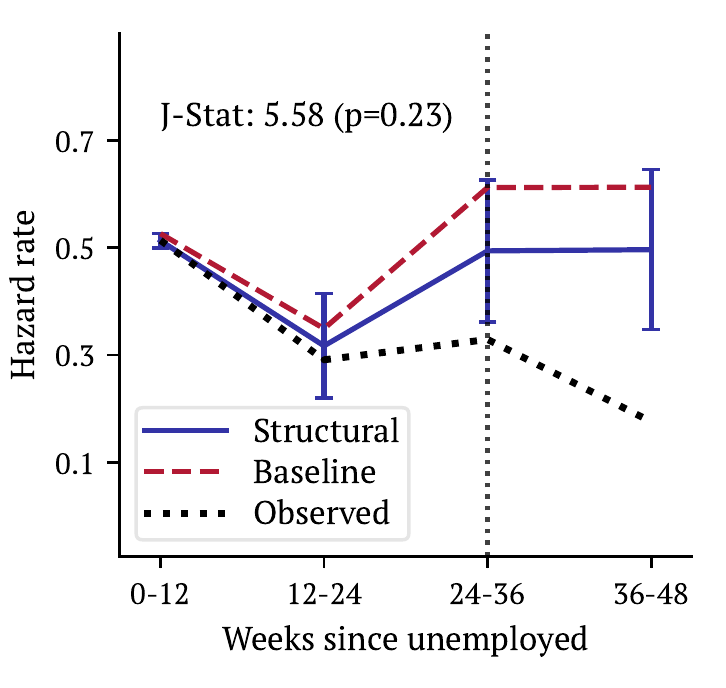}
\subcaption{<1 month, 1-2 months, >2 months}
\end{subfigure}
\floatfoot{\textit{Note:} Panel A presents the estimated hazard for the sample consisting of individuals who received no notice, along with those who received notice periods of 1-2 months and >2 months, as in the main analytical sample. Panel B estimates are for individuals who received notice of <1 month, in addition to the groups from the main analytical sample. J-Stat refers to the Sargan-Hansen statistic for overidentifying restrictions.}
\end{figure}

The estimates from the above exercise are presented in Figure \ref{fig_robust_notice_cat}. Panel A corresponds to the sample where the third category is no notice, while panel B corresponds to the sample where the third category is <1 month notice. After excluding extreme propensity score values, 7,122 of 12,061 observations remain for the data used in panel A, and 5,667 of 5,898 observations remain for the data used in panel B. In both cases, the estimates of structural duration dependence are qualitatively similar to the baseline results—the structural hazard consistently exceeds the exit rate from the data, initially decreasing, then increasing up to the third interval, and remaining constant thereafter. However, in panel B, the estimated structural hazard increases less leading up to benefit exhaustion compared to the baseline estimate. This could reflect differences in patterns of heterogeneity or duration dependence across the samples. More importantly, because we have more than two notice categories, the model is over-identified. With $3 \times 4 = 12$ moment conditions and 8 parameters, there are 4 degrees of freedom. The Sargan-Hansen J-statistic for testing overidentifying restrictions is reported in both panels, along with the respective $p$-values.\footnote{The estimated model in the main analysis is also over-identified due to the imposition of a log-logistic functional form on the structural hazard, resulting in one free parameter. However, here we are testing a more extensive set of overidentifying restrictions, with the caveat that the power of this test is lower.} In both cases, the large $p$-values indicate that there is no strong evidence against the null hypothesis.

\subsection{Binning Unemployment Duration}\label{subsec_binning}

Finally, a limitation of the application in the paper is the need to bin unemployment durations into 12-week intervals due to the small sample size of the DWS and the potentially noisy measurement of duration caused by retrospective bias. While Figure \ref{fig_altbins} shows that the patterns of how exit rates vary with duration and notice length are similar when binning data into 4- or 9-week intervals, I now present estimates from the MH model using unemployment duration data binned into 9-week intervals in Figure \ref{fig_robust_altbins}. Once again, we uncover the same pattern in estimates of structural duration dependence, with a decline from the first to the second interval, followed by a rise in the third interval corresponding to 18-27 weeks, which aligns with UI exhaustion at 26 weeks. I also present non-parametric estimates, which closely follow the log-logistic estimates.

\begin{figure}[t]\caption{Estimates with Unemployment Duration Binned in 9-Week Intervals}\label{fig_robust_altbins}
\vskip-0.3cm
\includegraphics{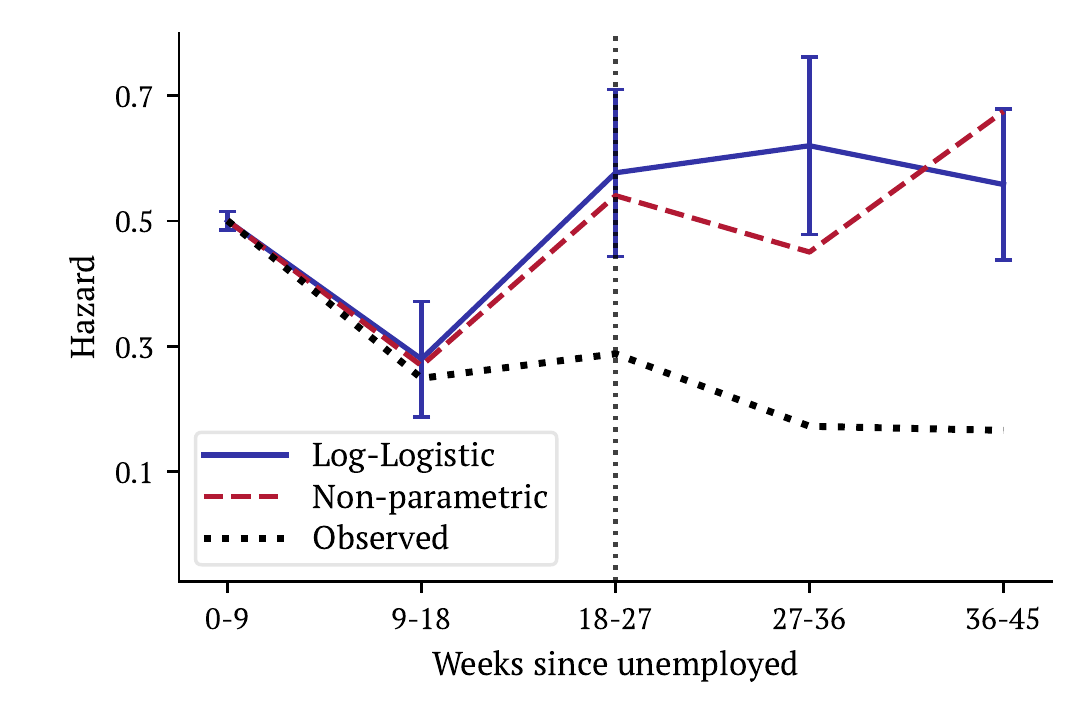}
\vskip-0.5cm
\floatfoot{\textit{Note:} The figure presents estimates from the Mixed Hazard model using data with unemployment duration binned in 9-week intervals. The solid blue line represents estimates for the structural hazard assuming the Log-Logistic functional form. The dashed red line represents non-parametric estimates, while the dotted black line represents the observed exit rate from the data.}
\end{figure}

I conduct another exercise to evaluate the effect of interval sizes directly on estimates from the MH model. Specifically, I take exit rates implied by the model for a given data-generating process (DGP) and compute the exit rates that would be observed if unemployment durations were binned into different-sized intervals. I then compare how the estimates change with the size of the bins. 

The data generating process (DGP) for this exercise is defined as follows: There are two groups of notice, each with a probability of 0.5, and the first-period structural hazards are set at 0.1 and 0.2 for the two groups, respectively. For periods after the first \((d > 1)\), the structural hazards, which are common to both groups, are generated in four different cases using the Weibull function \(bkd^{d-1}\)to ensure our conclusions are not an artifact of a specific DGP. These four cases are: (1) increasing \((b=0.2, k=1.2)\), (2) decreasing \((b=0.2, k=0.75)\), (3) constant \((b=0.15, k=1)\), and (4) non-monotonic, where the Weibull function is increasing as in (1) until \(d=7\) and then follows (2) multiplied by 1.75. The unobserved type \(\nu\) is a mixture of three Beta distributions: \(\nu_1 \sim \text{Beta}(0.1, 0.1)\) with probability 0.5, \(\nu_2 \sim \text{Beta}(0.3, 0.5)\) with probability 0.1, and \(\nu_3 \sim \text{Beta}(0.25, 0.5)\) with probability 0.4. I generate the exit rates for 12 time periods, bin them into intervals of four different sizes---1, 2, 3, and 4---and estimate the model in each case.

The dashed black line in Figure \ref{fig_sim_binA} represents the profile of exit rates corresponding to different bin sizes used in the estimation for the four DGPs from panels A-D. The solid dark-grey line shows the estimated structural hazard in each case. Comparing the estimates across different bin sizes in each panel, one can see that increasing the bin size leads to potentially missing the full extent of changes in the structural hazard over time. Specifically, when the structural hazard is increasing, the binned estimate increases by less (compare column 4 to column 1 in panel A). Similarly, when the structural hazard is decreasing, the binned estimate decreases by less (panel B). This is a direct consequence of the binned exit rates not fully capturing the extent of changes in the exit rates over time, making it difficult for the estimates using those exit rates to capture the full dynamics.

However, note that regardless of bin size, the estimates are able to uncover the broader trends in the structural hazard. This is because the estimates correspond to the cumulative structural hazard within a given interval for the average type, represented by the dotted red line on the plots. Specifically, the cumulative structural hazard is defined as the probability of an individual exiting unemployment in a given interval, provided that the individual hasn't exited by the start of the interval.\footnote{For instance, for a bin size of 2, the cumulative structural hazard corresponding to the second interval for the average worker is given by: \( \Pr(3 \leq D \leq 4 | D \geq 2, L, \nu = \mu) = \psi(3)\mu + [1-\psi(3)\mu]\psi(4)\mu \).} So, while the estimates may miss some nuances within the intervals, the extent of what is missed is limited to what would be missed if we were analyzing the true structural exit probabilities in broader intervals. This is also evident from the implied average type for each estimate presented in Figure \ref{fig_sim_binB}. From this figure, it can be seen that the average type in each panel falls similarly over time for different bin sizes, implying that we are able to capture the contribution of underlying heterogeneity in the observed exit rates very closely, regardless of the bin size. 

Overall, this shows that the assumptions of the model do not interact with bin sizes in a way that leads to systematically biased estimates, ensuring that the estimates reported in this paper correspond to meaningful underlying patterns.


\begin{figure}[p]\caption{Estimation on Binned Simulated Data: Hazard}\label{fig_sim_binA}
	\includegraphics{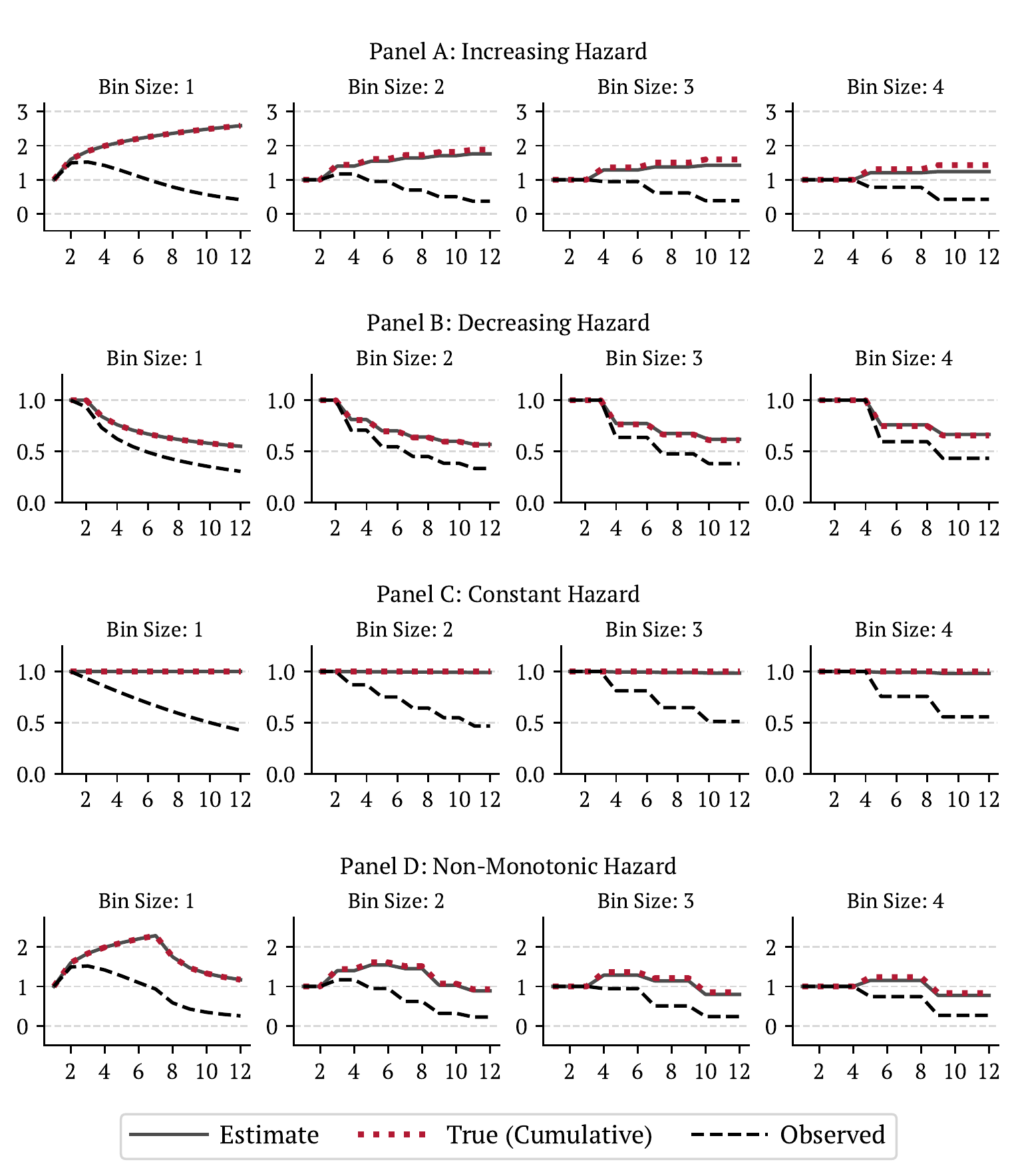}
\floatfoot{\textit{Note}: The figure presents results from a robustness exercise evaluating the impact of binning data into intervals. I calculate the moments used for estimation from the Mixed Hazard model with 12 durations. These moments are then binned into intervals of lengths ranging from 1 to 4, and the model is estimated using the binned moments. Along with the estimates of the structural hazard (dark grey solid line), I present two additional quantities binned similarly for comparison. The dotted red line represents the true structural cumulative hazard of exiting in the specific interval, while the black dashed line represents the average observed hazard.}
\end{figure}

\begin{figure}[p]\caption{Estimation on Binned Simulated Data: Average Type}\label{fig_sim_binB}
	\includegraphics{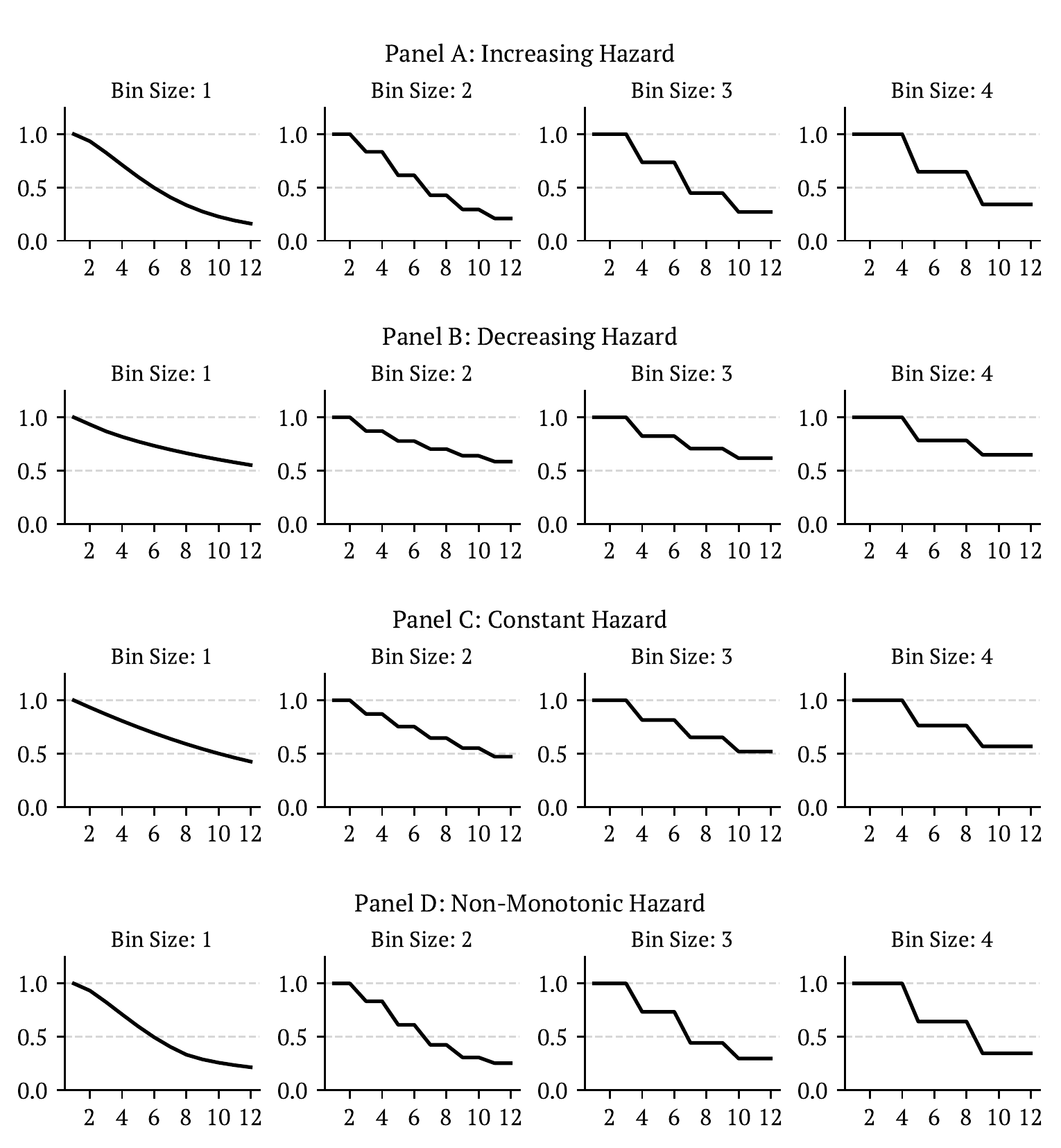}
\floatfoot{\textit{Note}: The figure presents results from a robustness exercise evaluating the impact of binning data into intervals. I calculate the moments used for estimation from the Mixed Hazard model with 12 durations. These moments are then binned into intervals of lengths ranging from 1 to 4, and the model is estimated using the binned moments. The plots represent the average type at each duration, as implied by the estimated structural hazard.}
\end{figure}

\clearpage
\setcounter{table}{0}
\setcounter{figure}{0}
\section{Generalization}\label{app_gen}

The main identification result in the paper relies on two crucial assumptions: (i) the notice length is independent of the worker type (conditional on observables), and (ii) the structural hazard after the initial period does not vary with notice length. In this section, I generalize the identification result showing that we can identify structural duration dependence and the moments of heterogeneity distribution, provided we know how the structural hazard following the initial period and the distribution of heterogeneity vary with notice length. Previously, the model operated under the assumption that the distribution of heterogeneity and the structural hazards at later durations did not vary with notice length. Now, we are considering alternative assumptions about how these factors might vary. This is useful as it allows us to assess estimates under different assumptions and compare them to the original results.

For brevity, I ignore observable characteristics while presenting the proof in this section and work with a version of the model where an individual worker's exit probability is given by \(h(d|\nu, L) = \psi_L(d) \nu\), with \(\nu\) being exogenous. However, the result can be straightforwardly extended to incorporate observable characteristics as in the model with \(h(d|\nu, L, X) = \psi_L(d) \phi(X) \nu\), under conditional independence. In fact, in the implementation of this result in Section \ref{app_subsec_impl}, I do use exit rates weighted by the inverse of propensity scores.

Unlike before, now we will not assume stationarity or independence; instead, we will allow the distribution of heterogeneity and structural hazards after the initial period to vary based on certain known parameters. In particular, for two lengths of notice $\ell$ and $\ell'$, define $\kappa_d$ as the difference between the $d^{th}$ moment of $\nu$ conditional on $\ell'$ and $\ell$ as follows:
$$ \kappa_d = \mu_{\ell',d}-\mu_{\ell,d} $$
where $\mu_{\ell,d} = E(\nu^d|\ell)$. So $\kappa_1$ is the difference between the average type of workers with $\ell'$ and $\ell$ notice lengths. Additionally, define $\gamma_d$ as the ratio of structural hazards at duration $d$ for two lengths of notice,
$$\gamma_d = \frac{\psi_{\ell'}(d)}{\psi_{\ell}(d)}$$

The identification result presented below states that if for some $\bar{D}$ we know $\kappa_d$ for $d=1,...,\bar{D}$ and $\gamma_d$ for $d=2,...,\bar{D}$, we can identify the first $\bar{D}$ structural hazards and moments of type distribution for each notice length up to scale.\footnote{Alternatively, we could know $\kappa_d$ for $d=2,...,\bar{D}$ and $\gamma_d$ for $d=1,...,\bar{D}$. Also, in theory, the choice of defining $\gamma_d$ and $\kappa_d$ as a ratio or a difference does not impact the proof of identification. In this case, I define $\kappa_d$ as a difference and $\gamma_d$ as a ratio for the convenience of varying these parameters when examining the changes in estimates.}

\subsection{General Identification Result}

\begin{theorem}\label{result_identification}
For some $\ell,\ell'$, define $\kappa_d = \mu_{\ell',d}-\mu_{\ell,d}$ and $\gamma_d = \psi_{\ell'}(d)/\psi_{\ell}(d)$. Then for some $\bar{D}$, if $\{\kappa_d\}_{d=1}^{\bar{D}}$ and $\{\gamma_d\}_{d=2}^{\bar{D}}$ are known, then the baseline hazards $\{\psi_l(d),\psi_{\ell'}(d)\}_{d=1}^{\bar{D}}$ and the conditional moments of the type distribution $\{\mu_{\ell,d},\mu_{\ell',d}\}_{d=1}^{\bar{D}}$ are identified up to a scale from $\{h(d|\ell),h(d|\ell')\}_{d=1}^{\bar{D}}$. 
\end{theorem}

\begin{proof}
First note that we can write,
\begin{align} 
g(d|\ell) = \psi_l(d) \sum_{k=1}^{d} c_k(\bm{\psi_{\ell,d-1}}) \mu_{\ell,k}   
\end{align}
where $\bm{\psi_{\ell,d-1}} = \{\psi_l(s)\}_{s=1}^{d-1}$, $c_k(\bm{\psi_{\ell,0}})=1$, and 
$$ c_k(\bm{\psi_{\ell,d-1}}) =  \begin{cases}
c_k(\bm{\psi_{\ell,d-2}}) & \text{for }  k=1 \\
c_k(\bm{\psi_{\ell,d-2}})-\psi_l(d-1) c_{k-1}(\bm{\psi_{\ell,d-2}}) & \text{for }   1<k \leq d \\
0 & \text{for }   k>d 
\end{cases} $$

Now, we can prove the statement of the theorem by induction. First, note that the statement is true for $\bar{D}=1$. To see this, note that 
$$ g(1|\ell) = \psi_l(1) \mu_{\ell,1} \quad \quad g(1|\ell') = \psi_{\ell'}(1)(\mu_{\ell,1}+\kappa_1)  $$
Normalizing $\mu_{\ell,1}=1$, we can solve for $\psi_l(1) = g(1|\ell)$ and $\psi_{\ell'}(1) =  \frac{g(1|\ell')}{1+\kappa_1}$.

Now, let us assume that the statement is true for $\bar{D}=d-1$. Then we can identify $\{\psi_l(s),\psi_{\ell'}(s)\}_{s=1}^{d-1}$ and $\{\mu_{\ell,s},\mu_{\ell',s}\}_{s=1}^{d-1}$ from $\{g(s|\ell),g(s|\ell')\}_{s=1}^{d-1}$. To complete the proof, we need to prove that the statement is true for $\bar{D}=d$ as well.

Denote $\Gamma_d =  \prod_{s=1}^{d} \gamma_s  $ and $\Psi_{\ell}(d) = \prod_{s=1}^{d} \psi_{\ell}(s)  $. Now note that, 
\begin{align*}
g(d|\ell) &= \psi_l(d)  \sum_{k=1}^{d} c_k(\bm{\psi_{\ell,d-1}}) \mu_{\ell,k} \\
&= \psi_l(d)  \left[ \sum_{k=1}^{d-1} c_k(\bm{\psi_{\ell,d-1}}) \mu_{\ell,k} + c_{d}(\bm{\psi_{\ell,d-1}}) \mu_{\ell,d} \right]   \\
&= \psi_l(d)  \left[ \sum_{k=1}^{d-1} c_k(\bm{\psi_{\ell,d-1}}) \mu_{\ell,k} + (-1)^{d-1}  \Psi_{\ell}(d-1) \mu_{\ell,d} \right]   
\end{align*}
From the above equation we can solve for $\mu_{\ell,d}$ as follows:
\begin{equation}\label{eq_mu}
\mu_{\ell,d} =  \frac{(-1)^{d} }{\Psi_{\ell}(d-1)}\left[  \sum_{k=1}^{d-1} c_k(\bm{\psi_{\ell,d-1}}) \mu_{\ell,k} - \frac{g(d|\ell)}{\psi_l(d)} \right]
\end{equation}
Using the fact that $\mu_{\ell',d} =\kappa_d+ \mu_{\ell,d} $, we can write $g(d|\ell')$ as follows:
\begin{align*}
g(d|\ell') &=  \psi_{\ell'}(d)  \left[ \sum_{k=1}^{d-1} c_k(\bm{\psi_{\ell',d-1}}) \mu_{\ell',k} + (-1)^{d-1} \Psi_{\ell'}(d-1) (\kappa_d+ \mu_{\ell,d}) \right]   
\end{align*}
By plugging in $\mu_{\ell,d}$ from equation (\ref{eq_mu}) in the above expression, we can solve for $\psi_{\ell'}(d) $ as follows:
$$ \psi_{\ell'}(d)  = \frac{ g(d|\ell')- \Gamma_{d}g(d|\ell)}{\sum_{k=1}^{d-1} c_k(\bm{\psi_{\ell',d-1}}) \mu_{\ell',k} -   \Gamma_{d-1} \sum_{k=1}^{d-1} c_k(\bm{\psi_{\ell,d-1}}) \mu_{\ell,k} + (-1)^{d-1} \kappa_{d} \Psi_{\ell'}(d-1) } $$

Plugging this back in expression for $\mu_{\ell',d} $, we can solve for
{\footnotesize $$ \mu_{\ell',d} =  \frac{(-1)^{d} }{\Psi_{\ell'}(d-1)}\left[ \frac{ g(d|\ell') \Gamma_{d-1} \sum_{k=1}^{d-1} c_k(\bm{\psi_{\ell,d-1}}) \mu_{\ell,k}-\Gamma_{d}g(d|\ell)\sum_{k=1}^{d-1} c_k(\bm{\psi_{\ell',d-1}}) \mu_{\ell',k} - (-1)^{d-1} g(d|\ell') \kappa_{d} \Psi_{\ell'}(d-1)}{g(d|\ell')- \Gamma_{d}g(d|\ell)} \right]  $$ }
So as long as the denominators in the expressions for $\psi_{\ell'}(d)$ and $\mu_{\ell',d}$ are not zero, we would have identification.
\end{proof}

We can see that with $\kappa_d = 0$ for $d=1,..,\bar{D}$ and $\gamma_d=1$ for $d=2,..,\bar{D}$, the above theorem is equivalent to the result in the main text. Also, note that the theorem can more generally be applied to situations with other observable characteristics. For instance, with $\kappa_d = 0$ for $d=1,..,\bar{D}$ and $\gamma_d=\gamma$ for $d=1,..,\bar{D}$, the above is equivalent to the discrete MPH model.  In the following subsection, I investigate how the estimates of structural hazard vary under different assumptions on $\kappa_d$ and $\gamma_d$.


\subsection{Implementation}\label{app_subsec_impl}

In our estimation, we utilized two lengths of notice, 1-2 months ($S$) and >2 months ($L$). Let's define $\kappa_d = \mu_{L,d}-\mu_{S,d}$ and $\gamma_d = \psi_{L}(d)/\psi_{S}(d)$. For our baseline estimates, we assumed that the distribution of heterogeneity for individuals with these different notice lengths was identical, i.e., $\kappa_d=0$ for all $d$. We also assumed that after the first period, the structural hazards for both the groups were the same, so $\gamma_d=1$ for $d>1$. I now study how our estimates change if the underlying distribution of heterogeneity and/or the structural hazards after the initial period are different for workers with different lengths of notice. In particular, I perform the following three exercises.  

\noindent \underline{1. Allow average type to vary} 

I relax the assumption that notice length is independent of a worker's type and let the mean of the heterogeneity distribution vary across the two groups. I assume that apart from the mean, the rest of the shape of the distribution for the two groups is identical. Since we have $\bar{D}=4$, this implies that in the 2nd, 3rd, and 4th central moment, the variance, skewness, and kurtosis for the two groups are identical. Scale changes would impact the non-central moments, so all four $\kappa_d$s will be non-zero. Denote central moments by $\tilde{\mu}$. Note that, $ \tilde{\mu}_{2} = \mu_{2}-\mu_{1}^2  $. Then since we need $\tilde{\mu}_{S,2} = \tilde{\mu}_{L,2}$, 
$$ \mu_{S,2}-\mu_{S,1}^2 = \mu_{S,2} + \kappa_2-(\mu_{S,1}+\kappa_1)^2 \rightarrow  \kappa_2 = \kappa_1 (\kappa_1 + 2\mu_{S,1}) $$
Similarly, noting that $\tilde{\mu}_3 = \mu_3 - 3 \mu_1 \mu_2 + 2 \mu_1^3$ and setting $\tilde{\mu}_{S,3} = \tilde{\mu}_{L,3}$, implies $\kappa_3 = \kappa_1(\kappa_1^2 + 3 \kappa_1 \mu_{S,1} + 3 \mu_{S,2}) $. And since, $\tilde{\mu}_4 = \mu_4 - 4 \mu_3 \mu_1 + 6 \mu_2 \mu_1^2 - 3 \mu_1^4$, then setting $\tilde{\mu}_{S,4} = \tilde{\mu}_{L,4}$, we would have $\kappa_4 = \kappa_1(\kappa_1^3 + 4 \kappa_1^2 \mu_{S,1} + 6 \kappa_1 \mu_{S,2} + 4 \mu_{S,3})$. 

\begin{figure}[p]\caption{Allow average type to vary}\label{fig_genHet}
\begin{subfigure}{.475\textwidth}
\includegraphics{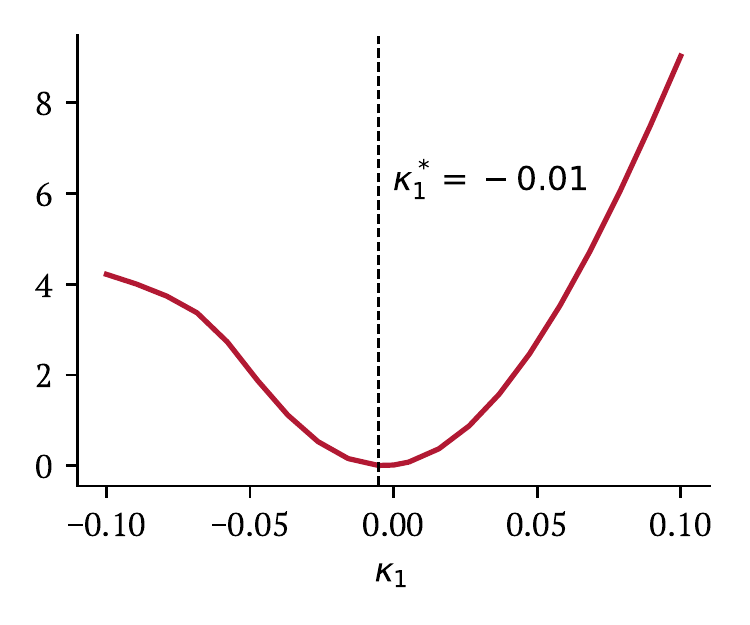}
\subcaption{Residuals}
\end{subfigure} 
\begin{subfigure}{.475\textwidth}
\includegraphics{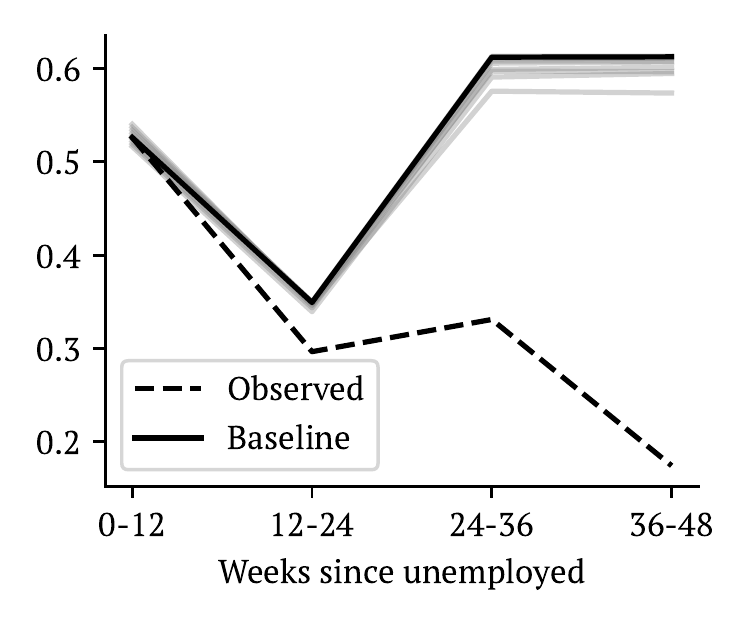}
\subcaption{10 lowest-residual estimates}
\end{subfigure} 
\vskip-0.2cm
\floatfoot{\textit{Note:} The figure presents results from the estimation of a more generalized Mixed Hazard model, where the mean of the heterogeneity distribution for individuals with different lengths of notice is allowed to vary according to the parameter $\kappa_1$. Panel A presents the residuals from GMM estimation for different values of $\kappa_1$. Panel B presents the structural hazard estimates from the ten best models with the lowest-valued residuals (light grey lines), compared to the baseline estimate (solid line) and the observed hazard in the data (dashed line).}
\end{figure} 

Now, assuming \(\gamma_d = 1\) for \(d > 1\) and normalizing \(\mu_{S,1} = 1\), I reestimate the model with \(\kappa_1\) values spaced at 0.01 intervals within the range \([-0.1, 0.1]\).\footnote{For values beyond this interval, the model fit deteriorates drastically, and the estimated moments of the heterogeneity distribution blow up in either direction.} $\kappa_2,\kappa_3$ and $\kappa_4$ are defined as above. Residuals from this exercise are presented in panel A of Figure \ref{fig_genHet}, from which we can see that the residual-minimizing value of \(\kappa\) is close to zero. In panel B, I present the estimates for structural duration dependence for the values of \(\kappa\) from the ten models with the lowest residuals. Results from this exercise suggest that the assumption \(\kappa = 0\) is not unreasonable and likely does not affect the conclusions.

\noindent \underline{2. Allow structural hazards after the first period to vary} 

Now,  I assume notice length to be independent of worker type but allow structural hazards beyond the initial period to vary for workers with different lengths of notice up to some constant $\gamma$. This corresponds to assuming $\kappa_d = 0$ for $d=1,..,\bar{D}$ and $\gamma_d=\gamma$ for $d=2,..,\bar{D}$. I estimate the model for values of $\gamma$ spaced at intervals of 0.01 within the range $[0.9, 1.1]$. Results from this exercise are presented in Figure \ref{fig_genStr}. As before, panel A shows model residuals for different \(\gamma\) values, while panel B presents the estimates for structural duration dependence for the values of \(\gamma\) from the ten models with the lowest residuals. As we can see, the residual-minimizing value of \(\gamma\) is close to one, giving support to the assumption \(\gamma = 1\).


\begin{figure}[p]\caption{Allow structural hazards after the first period to vary}\label{fig_genStr}
\begin{subfigure}{.475\textwidth}
\includegraphics{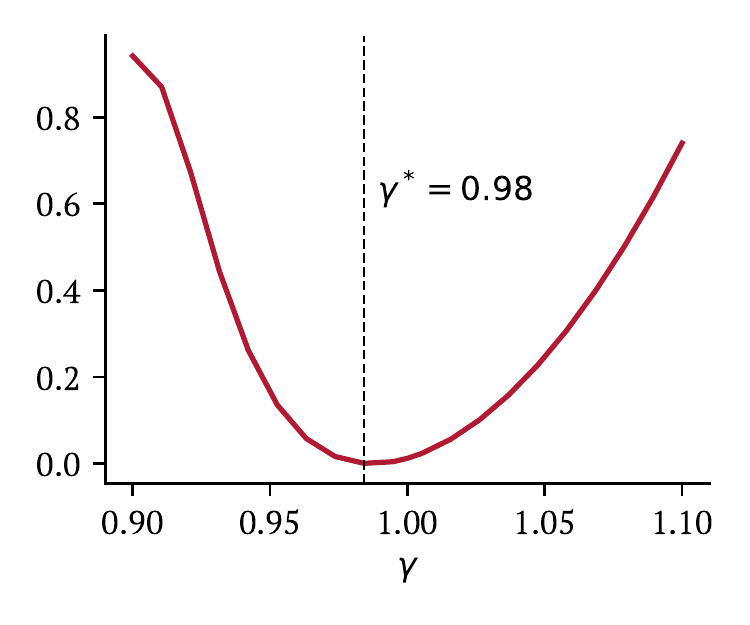}
\subcaption{Residuals}
\end{subfigure} 
\begin{subfigure}{.475\textwidth}
\includegraphics{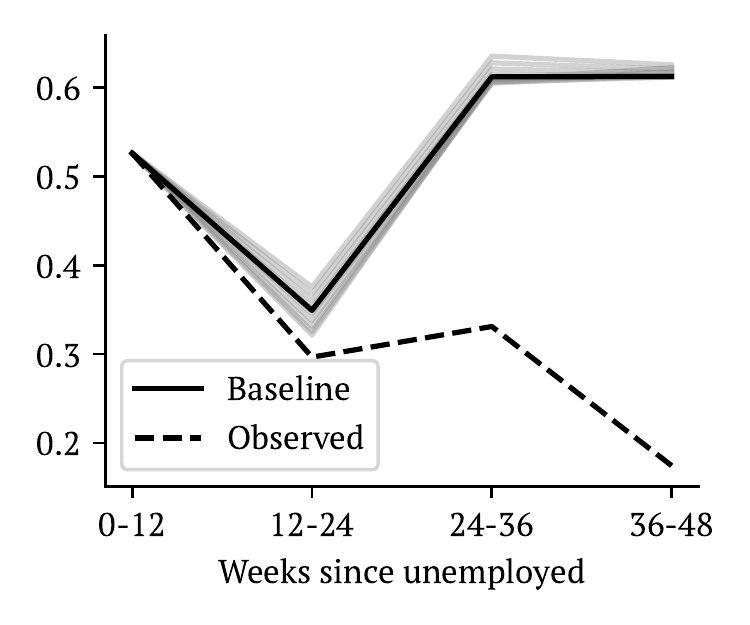}
\subcaption{10 lowest-residual estimates}
\end{subfigure} 
\vskip-0.2cm
\floatfoot{\textit{Note:}  The figure presents results from the estimation of a more generalized Mixed Hazard model, where the structural hazard after the initial period for individuals with different lengths of notice is allowed to vary according to the parameter $\gamma$. Panel A presents the residuals from GMM estimation for different values of $\gamma$. Panel B presents the structural hazard estimates from the 10 best models with the lowest-valued residuals (light grey lines), compared to the baseline estimate (solid line) and the observed hazard in the data (dashed line).}
\end{figure} 

\noindent \underline{3. Allow the average type and structural hazards after the first period to vary} 

Finally, I create a grid for values of \(\kappa\) and \(\gamma\) used in the above two exercises and reestimate the model for each point in the grid. Panel A of Figure \ref{fig_genHetStr} presents the residuals for different values in the grid, while panel B of Figure \ref{fig_genHetStr} shows the estimates for the set of \(\kappa\) and \(\gamma\) values that result in the 25 models with the lowest residuals. Mirroring the results from the above two exercises, the findings indicate no significant mean differences between short- and long-notice groups, nor a difference in structural hazards for long- vs. short-notice workers beyond the initial period.

\begin{figure}[b]\caption{Alternative Assumptions on Structural Hazards and Heterogeneity Distribution}\label{fig_genHetStr}
\makebox[\textwidth]{
\begin{subfigure}{.475\textwidth}
\includegraphics{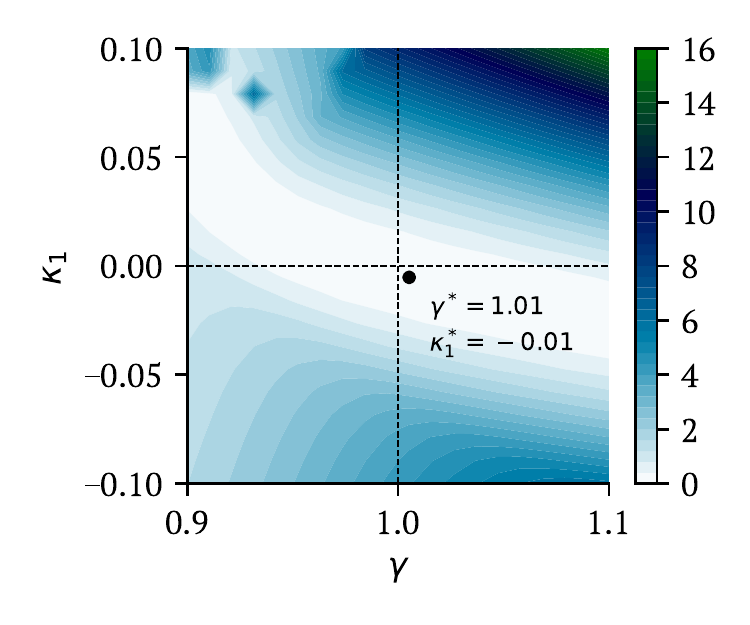}
\subcaption{Residuals}
\end{subfigure} \hspace{1em}
\begin{subfigure}{.475\textwidth}
\includegraphics{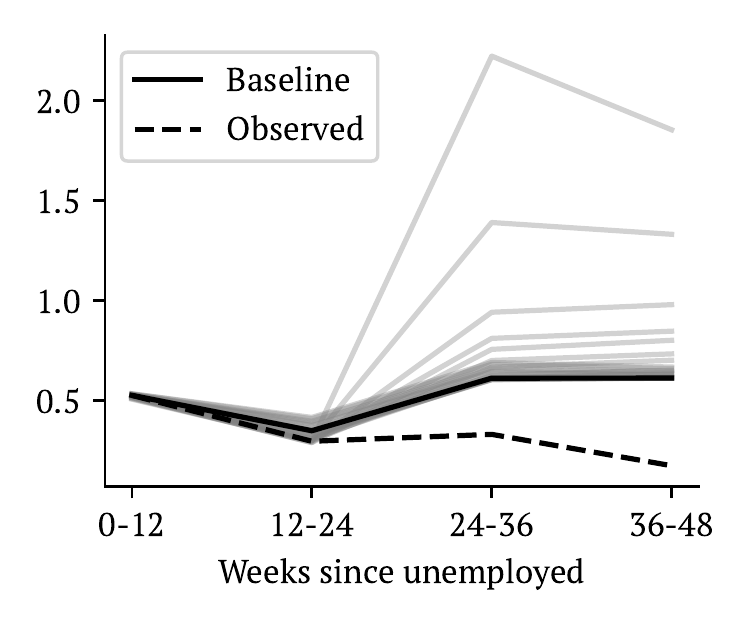}
\subcaption{25 lowest-residual estimates}
\end{subfigure} 
\floatfoot{\textit{Note:}  The figure presents results from the estimation of a more generalized Mixed Hazard model. The mean of the heterogeneity distribution for individuals with different lengths of notice is allowed to vary according to the parameter $\kappa_1$. The structural hazard after the initial period for individuals with different lengths of notice is allowed to vary according to the parameter $\gamma$. Panel A presents the residuals from GMM estimation for different values of $\kappa_1$ and $\gamma$. Panel B presents the structural hazard estimates from the 25 best models with the lowest-valued residuals (light grey lines), compared to the baseline estimate (solid line) and the observed hazard in the data (dashed line).} 
}
\end{figure}

\clearpage
\setcounter{table}{0}
\setcounter{figure}{0}
\section{Search Model Simulation}\label{app_sm_sim}

In this section, I simulate data from the search model presented in the main text. To incorporate multiple notice periods, I let the offer rate in the first period be different for long (L) and short (S) notice individuals. I set $\nu_H=1$, $\nu_L=0.5$ and $\pi=0.5$, $\delta_L(1)=1.25, \delta_S(1)=1$, and $\delta(d)=0.95$ for $d=2,3,4$. The rest of the parameters are set as in the calibration of the model in the main text. I assume there are 3000 individuals, half of whom receive the $L$ length notice. I simulate data on exit rates for this model 1000 times. The average of estimates for the structural hazard is presented in Figure \ref{fig_simulation_average}, while the distribution of the estimates is presented in Figure \ref{fig_simulation_distribution}.

\begin{figure}[h]\caption{Simulation: Average Estimate}\label{fig_simulation_average}
\includegraphics{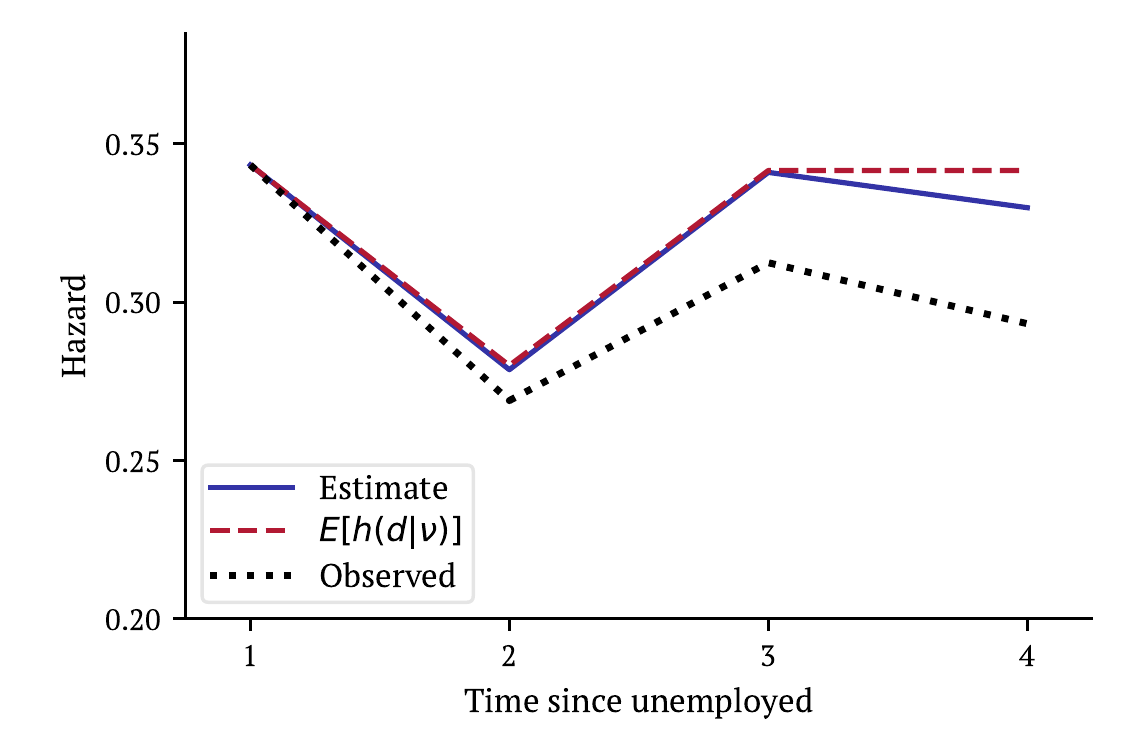}
\floatfoot{\textit{Note:}  The solid blue line presents the average estimate from 1000 simulations of the search model. The dashed red line presents the structural duration dependence $\E [h(d|\nu)]$ implied by the model. While the dotted black line presents the observed structural duration dependence $\E [h(d|\nu)|D \geq d]$ implied by the model.}
\end{figure}

\begin{figure}[p]\caption{Estimates using Simulated Data from the Search Model}\label{fig_simulation_distribution} 
\begin{subfigure}{.4\textwidth}
\includegraphics{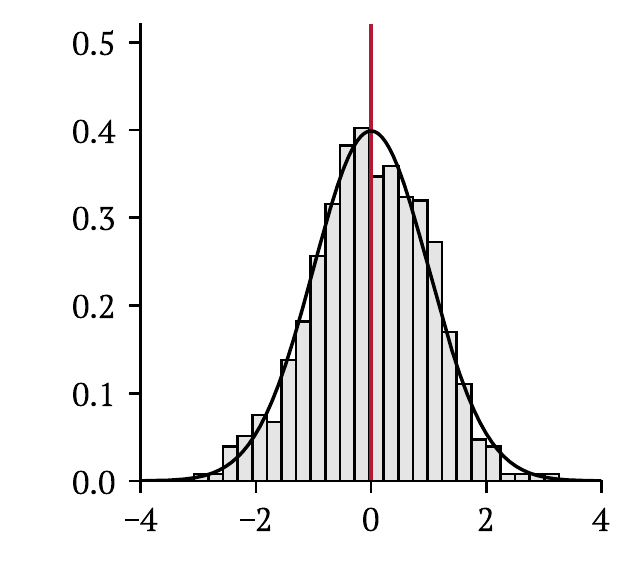}
\subcaption{$\psi(1)$}
\end{subfigure} 
\begin{subfigure}{.4\textwidth}
\includegraphics{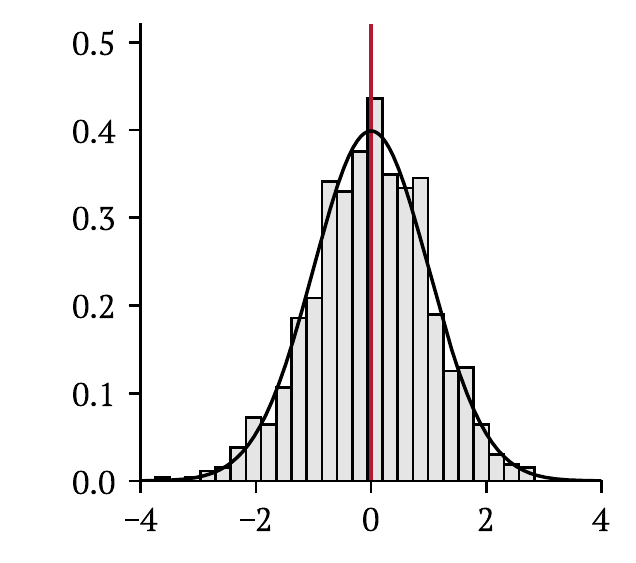}
\subcaption{$\psi(2)$}
\end{subfigure} 
\begin{subfigure}{.4\textwidth} 
\includegraphics{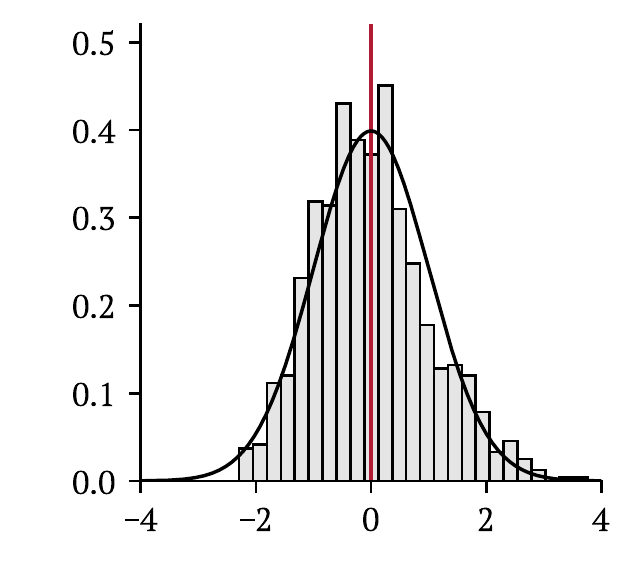}
\subcaption{$\psi(3)$}
\end{subfigure} 
\begin{subfigure}{.4\textwidth}
\includegraphics{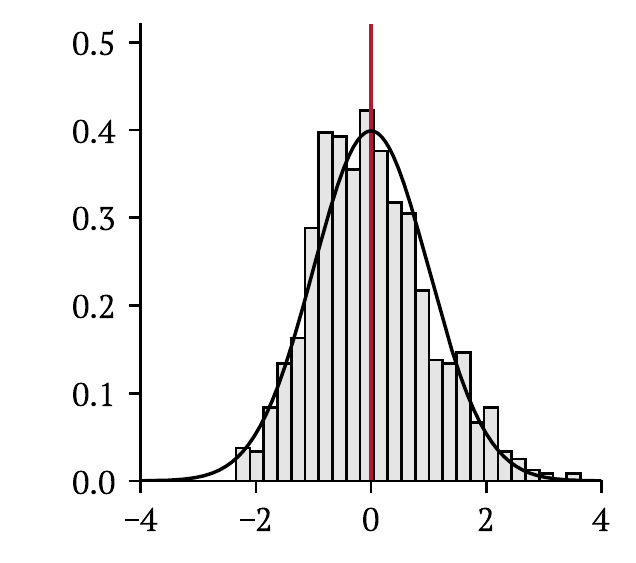}
\subcaption{$\psi(4)$}
\end{subfigure}
\floatfoot{\textit{Note:}  The figure presents the normalized distribution of structural duration dependence estimated on simulated data from the search model. The vertical lines represent the mean of the distribution for each structural hazard. Standard normal density is overlaid for reference.}
\end{figure}

\end{document}